%% file: main.tex
\documentclass[reprint,superscriptaddress,aps,pra,floatfix]{revtex4-2}
\usepackage{dynkin-diagrams}
\usepackage[export]{adjustbox}
\usepackage{amsmath}
\usepackage{amsthm}
\usepackage{thm-restate}
\usepackage{braket}
\usepackage[caption=false]{subfig}
\usepackage{makecell} % Split row tables
\usepackage{bm}
\usepackage{hyperref}
\usepackage[capitalize]{cleveref}
\newcommand{\circleRoot[1]}{\draw (root #1) circle (3pt);}
\newcommand{\nkd}[3]{[\![#1,#2,#3]\!]}
\newcommand{\col}[1]{\mathrm{col}({#1})}
\newcommand{\ketbra}[2]{\mathinner{|{#1}\rangle\langle{#2}|}}
\DeclareMathOperator{\argmax}{arg\,max}
\DeclareMathOperator{\supp}{supp}
\DeclareMathOperator{\rank}{rank}
\DeclareMathOperator{\St}{St}
\newtheorem{theorem}{Theorem}
\newtheorem{lemma}[theorem]{Lemma}
\theoremstyle{definition}
\newtheorem{example}[theorem]{Example}
\hypersetup{
    colorlinks,
    linkcolor={red!50!black},
    citecolor={green!50!black},
    urlcolor={blue!50!black}
}

\begin{document}

\title{Morphing quantum codes}

\author{Michael Vasmer}
\affiliation{Perimeter Institute for Theoretical Physics, Waterloo, ON N2L 2Y5, Canada}
\affiliation{Institute for Quantum Computing, University of Waterloo, Waterloo, ON N2L 3G1, Canada}

\author{Aleksander Kubica}
\affiliation{Perimeter Institute for Theoretical Physics, Waterloo, ON N2L 2Y5, Canada}
\affiliation{Institute for Quantum Computing, University of Waterloo, Waterloo, ON N2L 3G1, Canada}
\affiliation{AWS Center for Quantum Computing, Pasadena, CA 91125, USA}
\affiliation{California Institute of Technology, Pasadena, CA 91125, USA}

\date{\today}

%%%%%%%%%%%%%%
%% Abstract %%
%%%%%%%%%%%%%%
\begin{abstract}
    We introduce a morphing procedure that can be used to generate new quantum codes from existing quantum codes.
	In particular, we morph the 15-qubit Reed-Muller code to obtain a \nkd{10}{1}{2} code that is the smallest known stabilizer code with a fault-tolerant logical $T$ gate. 
	In addition, we construct a family of hybrid color-toric codes by morphing the color code. 
	Our code family inherits the fault-tolerant gates of the original color code, implemented via constant-depth local unitaries.
	As a special case of this construction, we obtain toric codes with fault-tolerant multi-qubit control-$Z$ gates. 
	We also provide an efficient decoding algorithm for hybrid color-toric codes in two dimensions, and numerically benchmark its performance for phase-flip noise.
	We expect that morphing may also be a useful technique for modifying other code families such as triorthogonal codes.  
\end{abstract}

\maketitle

%%%%%%%%%%%%%%%%%%
%% Introduction %%
%%%%%%%%%%%%%%%%%%
%\section{Introduction \label{sec:intro}}

The techniques of quantum error correction seem to be indispensable in the quest to build a large-scale quantum computer. 
Not only do we need to protect fragile qubits from environmental noise, but we also need to perform quantum circuits fault-tolerantly.
The celebrated threshold theorem~\cite{aharonov2008,knill1998a,kitaev1997} proves that scalable fault-tolerant quantum computation is possible in principle, and moreover fault-tolerance can be achieved with constant overhead for particular families of quantum error-correcting codes~\cite{gottesman2014,fawzi2018a}. 
However this is not the end of the story, as these asymptotic results can hide large constant factors and the requirements of quantum error correction~\cite{litinski2019,fowler2019,chamberland2020b,kim2021,beverland2021a} are still beyond the capabilities of today's hardware~\cite{egan2021,andersen2020,chen2021,chen2021a,erhard2021}. 
Therefore, it is imperative to develop new quantum codes that are more efficient, more resilient to noise, and more tailored to hardware. 

We introduce a theoretical tool for modifying quantum codes, called morphing, that allows us to systematically construct new codes from existing codes.
Morphing does not change the number of logical qubits encoded into a code, however it changes other properties of a code including the number of physical qubits, and may also change the code distance, the stabilizer weights, and the implementation of logical gates. 
Consequently, morphing enables us to trade off different code parameters against each other in order to satisfy the constraints of a particular hardware platform, e.g., the number of physical qubits can be decreased at the cost of complicating the implementation of some logical gates.

We apply morphing to the 15-qubit Reed-Muller code thereby constructing a \nkd{10}{1}{2} code with a fault-tolerant $T$ gate.
To our knowledge, this is the smallest known stabilizer code with a fault-tolerant implementation of this gate.
We investigate the performance of the \nkd{10}{1}{2} code in magic state distillation (MSD), finding that it may lead to advantages over previous MSD protocols based on small codes~\cite{bravyi2005,meier2012}.
Our \nkd{10}{1}{2} protocol is the first example of an MSD protocol mixing different input magic states, and we expect that our approach may be used to improve other MSD protocols.

By morphing the color code~\cite{bombin2006,bombin2007a,kubica2018a}, we construct a family of hybrid color-toric (HCT) codes, which inherit the fault-tolerant logical gates of the color code. 
Starting from an initial color code, we can construct HCT codes with different stabilizer weights, qubit connectivity, and logical gate implementations, allowing us to construct codes tailored to a particular hardware platform or fault-tolerant protocol.
In three or more dimensions, our HCT code family includes codes with fault-tolerant non-Clifford gates.  
In particular, our family subsumes previous examples of toric codes with transversal non-Clifford gates~\cite{kubica2015a,vasmer2019,jochym-oconnor2021}.

Lastly, we study the problem of efficient decoding for the family of the HCT codes.
We propose a decoding algorithm that connects two distinct decoders: the minimum-weight perfect matching (MWPM) toric code decoder~\cite{dennis2002,fowler2015} and the color code restriction decoder~\cite{kubica2019}.
We benchmark our algorithm for two-dimensional (2D) HCT codes subject to phase-flip noise.
Our decoder achieves good performance across the entire family, and recovers previous results in the toric code and color code limits.

The remainder of this manuscript is structured as follows.
We begin by introducing morphing and discussing a simple example in \cref{sec:morphing}.
Then, in \cref{sec:msd}, we present our first application: morphing the 15-qubit Reed-Muller code. 
Next, we apply morphing to color codes in \cref{sec:morph_color} to produce a family of HCT codes.
We examine the decoding problem for HCT codes in \cref{sec:decoding}, where we propose and benchmark a decoder for 2D HCT codes.
Finally, in \cref{sec:conc} we summarize our results and suggest future applications of morphing. 

%%%%%%%%%%%%%%%%%%%%%%%%%%%%
%% Morphing quantum codes %% 
%%%%%%%%%%%%%%%%%%%%%%%%%%%%
\section{Morphing quantum codes \label{sec:morphing}}

\begin{figure*}
	\centering
	\subfloat[]{
		\centering
		\includegraphics[width=0.21\linewidth,valign=b]{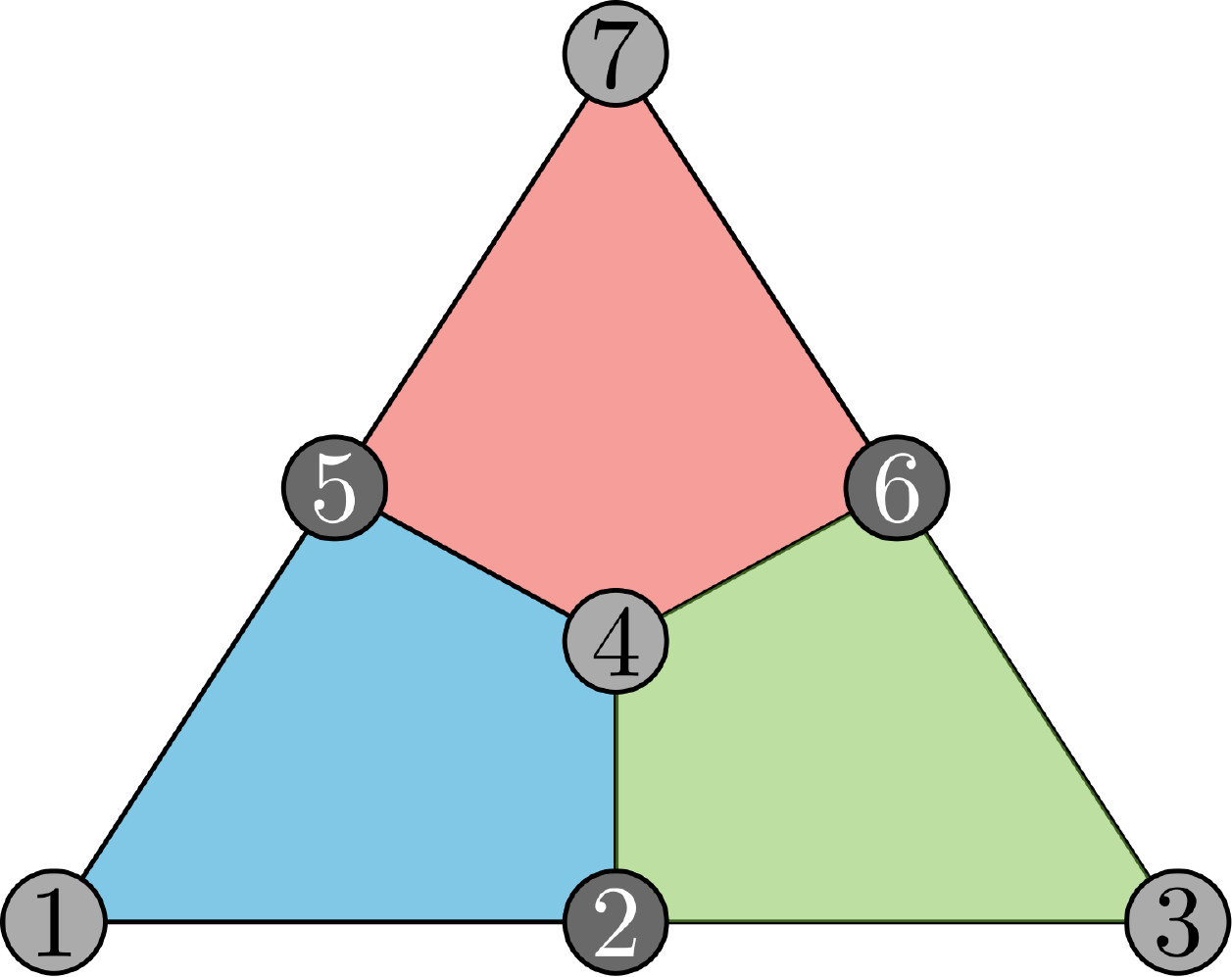}
		\label{sfig:steane}
	}
	\quad
	\subfloat[]{
		\centering
		\includegraphics[width=0.21\linewidth,valign=b]{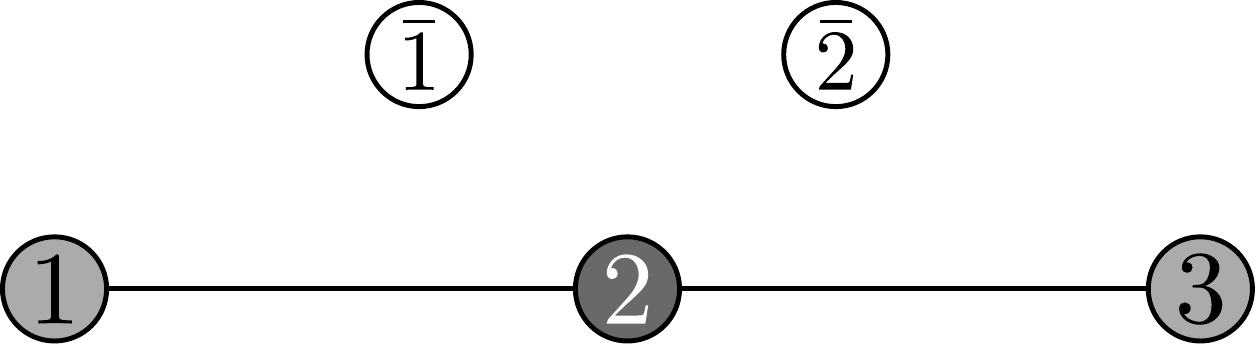}
		\label{sfig:morph_steane}
	}
	\quad
	\subfloat[]{
		\centering
		\includegraphics[width=0.21\linewidth,valign=b]{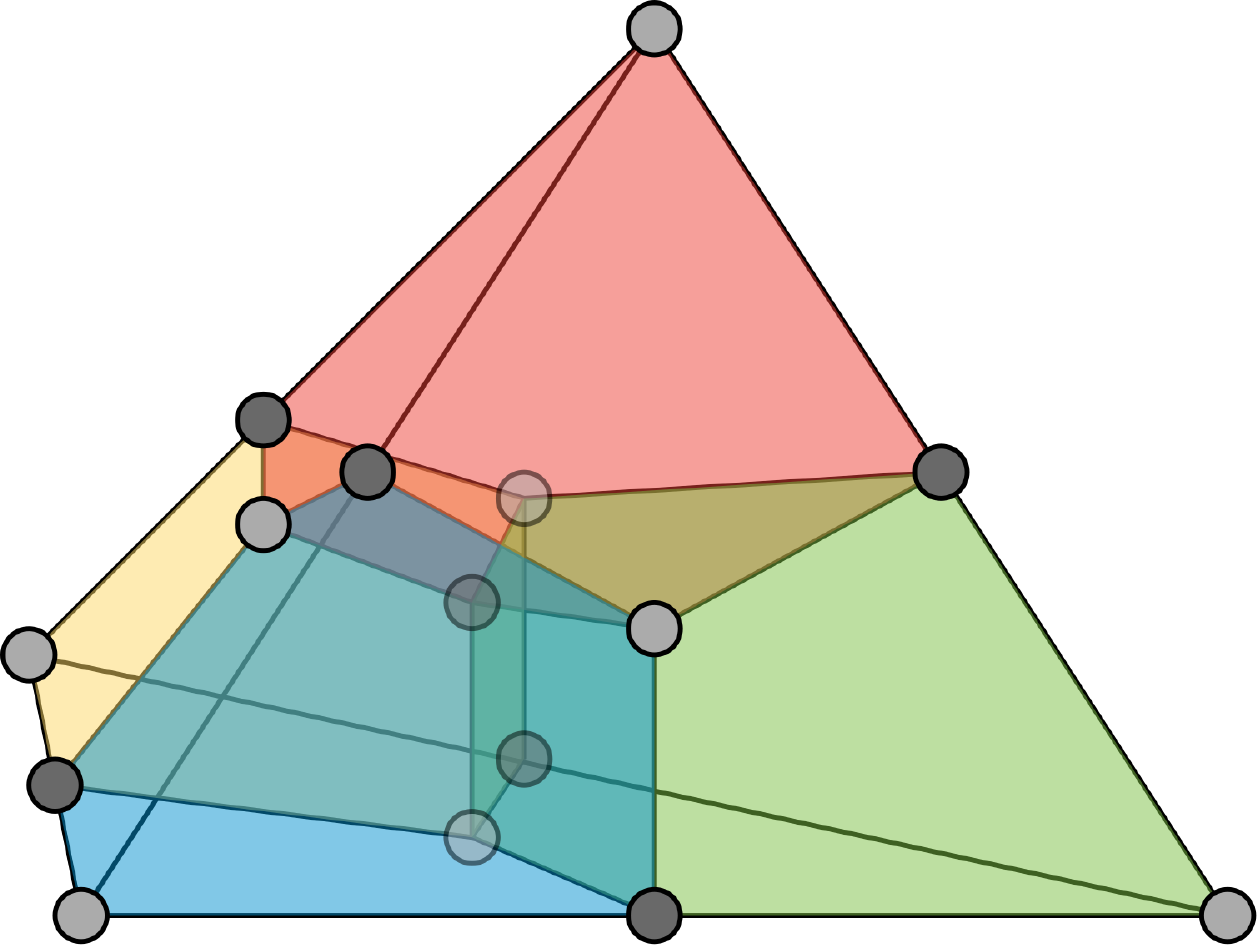}
		\label{sfig:15qrm}
	}
	\quad
	\subfloat[]{
		\centering
		\includegraphics[width=0.21\linewidth,valign=b]{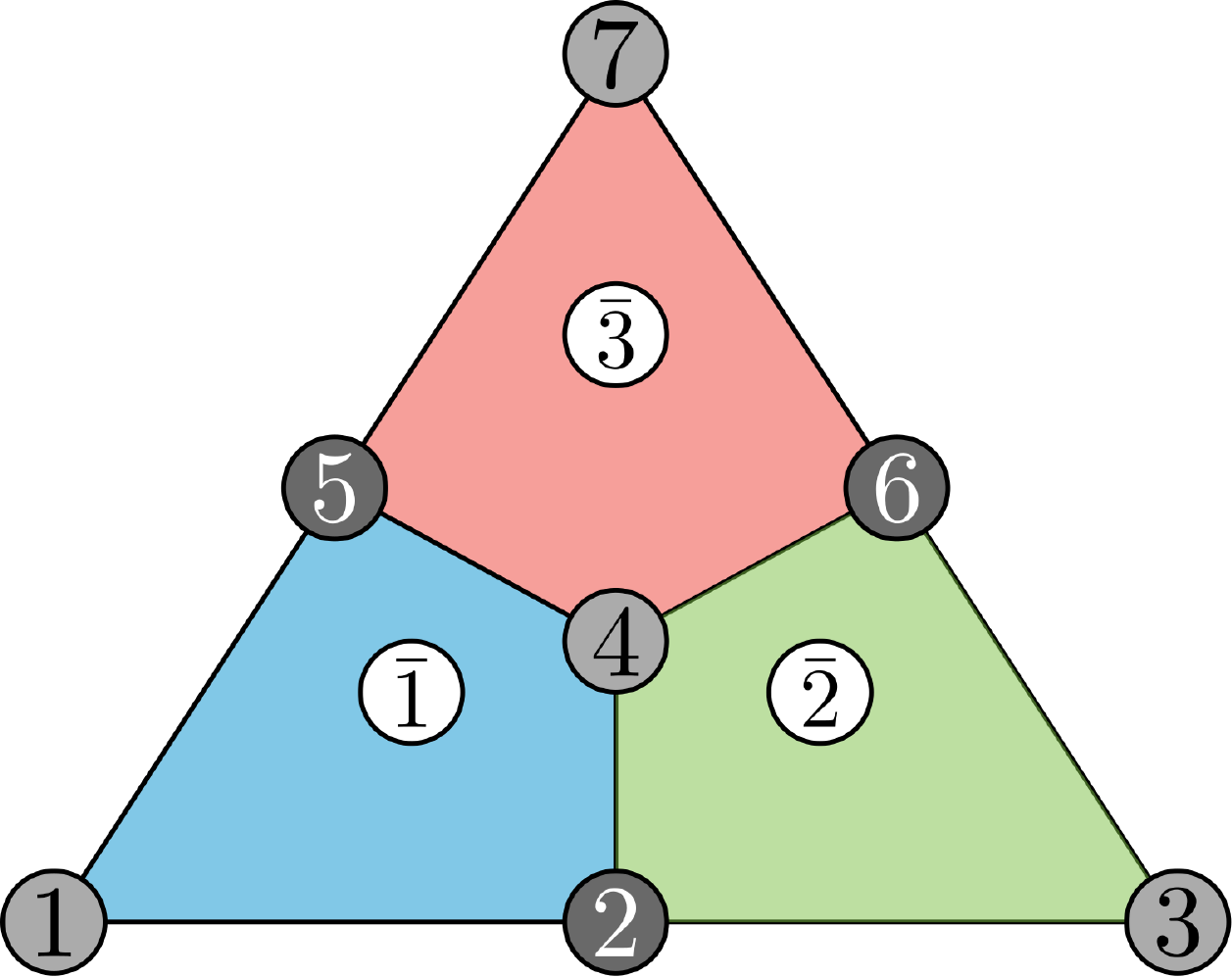}
		\label{sfig:morph_15qrm}
	}
	\caption{
		(a) 
		Steane's code. Qubits are on vertices; $X$- and $Z$-type stabilizers are on faces.
		(b)
		A \nkd{5}{1}{2} code obtained by morphing (a) with $R$ comprising the qubits in the red face, where the corresponding child code is a \nkd{4}{2}{2} code with logical qubits $\bar 1$ and $\bar 2$.
		The stabilizer group of the \nkd{5}{1}{2} code is $\langle X_1 X_2 X_{\bar 1}, X_2 X_3 X_{\bar 2}, Z_1 Z_2 Z_{\bar 2}, Z_2 Z_3 Z_{\bar 1}\rangle$.
		To implement $\overline S$ we apply $S^{\pm 1}$ to the light/dark grey qubits and $CZ$ to the white qubits; $\overline H=H^{\otimes 5}SWAP_{\bar 1 \bar 2}$.
		(c)
		The \nkd{15}{1}{3} Reed-Muller code. 
		Qubits are on vertices; $X$- and $Z$-type stabilizers are on cells and faces, respectively. 
		(d)
		A \nkd{10}{1}{2} code obtained by morphing (c) with $R$ comprising the qubits in the yellow cell, where the corresponding child code is an \nkd{8}{3}{2} code with logical qubits $\bar 1$, $\bar 2$ and $\bar 3$.
		The stabilizer group of the \nkd{10}{1}{2} code is generated by $X_1 X_2 X_4 X_5 X_{\bar 1}$, $Z_1 Z_2 Z_4 Z_5$, $Z_4 Z_5 Z_{\bar 2}$, and the analogous operators for the red and green faces. 
		To implement $\overline T$ we apply $T^{\pm 1}$ to the light/dark grey qubits and $CCZ$ to the white qubits.
	}
	\label{fig:morph_rm}
\end{figure*}

Our key observation is that for a subset of the physical qubits of a stabilizer code~\cite{gottesman1997}, the stabilizer group induces a (smaller) stabilizer code on the qubits in the subset. 
Consider some stabilizer code $\mathcal C$ with the stabilizer group $\mathcal S$. 
Let $Q$ denote the set of physical qubits of $\mathcal C$ and let $R \subseteq Q$.
We define $\mathcal S(R)$ to be the group generated by all the stabilizer generators of $\mathcal S$ that are fully supported within $R$, i.e.,
\begin{equation}
	\mathcal S(R) = \{ P \in \mathcal S : \supp (P) \subseteq R \},
	\label{eq:stab_subgroup}
\end{equation}
Since $\mathcal S(R)$ is a subgroup of $\mathcal S$, it is itself an abelian subgroup of the Pauli group not containing $-I$, and therefore defines a stabilizer code, $\mathcal C(R)$, with physical qubits $R$.
We refer to $\mathcal C$ as the parent code and $\mathcal C(R)$ as the child code.

As a stabilizer code, $\mathcal C(R)$ has a Clifford encoding circuit that can be implemented with depth $O(|R|)$ using a simple gate set of $\{ H, S, CNOT \}$~\cite{gottesman1997}, where $|R|$ is the size of $R$.
Suppose that we implement the inverse of this encoding circuit.
Then the logical qubits of $\mathcal C(R)$ will still be entangled with the rest of the parent code, but the remaining qubits of $\mathcal C(R)$ will be in a fixed (product) state and are free to be discarded~\footnote{We observe that morphing is similar to the disentangling step of entanglement renormalization~\cite{vidal2007}.} (assuming we were originally in the codespace of $\mathcal C$). 
In effect, we have produced a new code $\mathcal C_{\setminus R}$ whose physical qubits are the union of $Q \setminus R$ and the logical qubits of $\mathcal C(R)$.
We call this procedure morphing, we call the inverse of the encoding circuit a morphing circuit and we call $\mathcal C_{\setminus R}$ the morphed code.
Note that we do not necessarily anticipate implementing morphing circuits in a real device, rather we use morphing as a systematic procedure to generate new quantum codes. 

The morphed code $\mathcal C_{\setminus R}$ clearly has fewer physical qubits than $\mathcal C$. 
But crucially, the number of logical qubits in both codes is the same, as morphing decreases the number of physical qubits by the same amount as the number of independent stabilizer generators.
However, the implementation of a subset of the logical operators of $\mathcal C_{\setminus R}$ may be different to the implementation of the corresponding operators in $\mathcal C$.
As $\mathcal S(R) \subseteq \mathcal S$, the logical operators of the parent code act as logical operators in the child code (when restricted to $R$). 
Therefore, given the logical Pauli operators of $\mathcal C(R)$, we can compute the transformation of the logical operators of $\mathcal C$ under morphing.
This highlights an important degree of freedom inherent in morphing: the choice of the logical Pauli operators of the child code. 
All choices are equally valid, but different choices can lead to different implementations of the logical operators of $\mathcal C_{\setminus R}$; see \cref{app:non_cliff_ddim}.

We note that the code distance of $\mathcal C_{\setminus R}$ may be smaller than that of $\mathcal C$ and by the same reasoning the weight of some stabilizers of $\mathcal C_{\setminus R}$ may be smaller than the corresponding stabilizers in $\mathcal C$. 
Less obviously, morphing can simplify the decoding problem of the parent code; see \cref{sec:decoding}.

In general, morphing gives us the ability to modify quantum codes, trading off various code properties against each other in order to construct a code that is suited to the constraints of a particular quantum computing platform. 
We now make the above discussion concrete by considering a simple example which nevertheless demonstrates many of the changes induced by morphing.

\begin{example}
	Consider the Steane code~\cite{steane1996a}, which has parameters \nkd{7}{1}{3} and a transversal implementation of the Clifford group.
	We can pick a particular subset, $R$, of the physical qubits of the Steane code such that $\mathcal C(R)$ is the well-known \nkd{4}{2}{2} code~\cite{knill2005}; see \cref{sfig:steane}.
	Morphing the Steane code produces a \nkd{5}{1}{2} code with a fault-tolerant implementation of the Clifford group; see \cref{sfig:morph_steane}.
	\label{ex:steane}
\end{example}

%%%%%%%%%%%%%%%%%%%%%%%%%%%%%%%%%%%%%%%%%%%%%%%%%%%%%%%%%%%%%%%%
%% Magic state distillation with the morphed Reed-Muller code %%
%%%%%%%%%%%%%%%%%%%%%%%%%%%%%%%%%%%%%%%%%%%%%%%%%%%%%%%%%%%%%%%%
\section{Magic state distillation with the morphed Reed-Muller code \label{sec:msd}}

\begin{figure*}
	\centering
	\subfloat[]{
		\centering
		\includegraphics[valign=b]{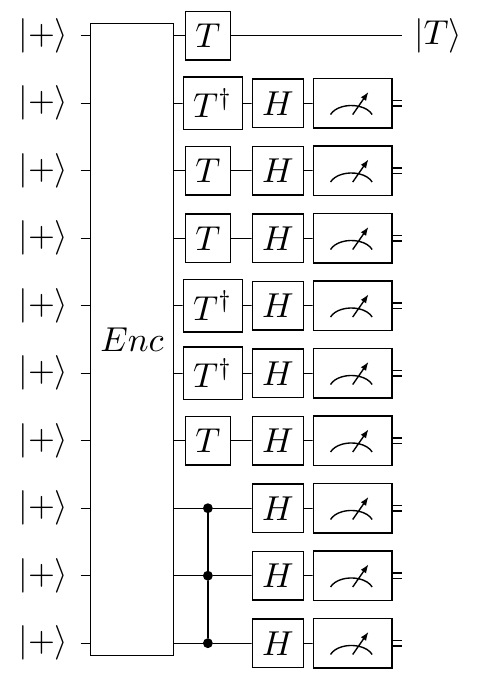}
	}
	\subfloat[]{
		\centering
		\begin{minipage}[b]{0.6\linewidth}
		\includegraphics[]{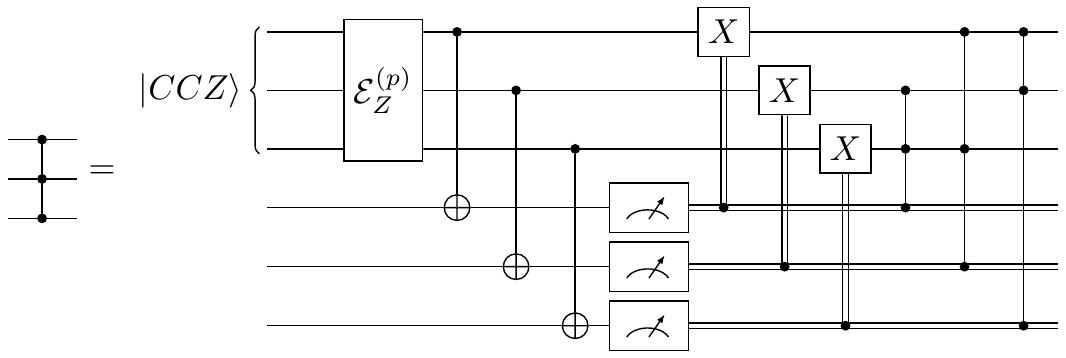}
		\\
		\vspace{0.75cm}
		\includegraphics[valign=b]{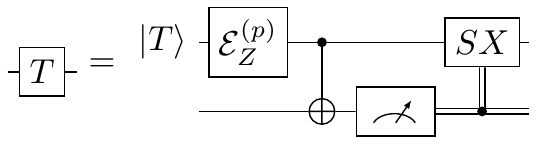}
		\end{minipage}
	}
	\caption{
		(a) The \nkd{10}{1}{2} MSD protocol, consisting of the encoding circuit, noisy application of the logical $T$ gate, and measurement. 
		The protocol is successful if the measurements give a trivial error syndrome.
		The labelling of the qubits is the same as in \cref{sfig:morph_15qrm}.
		(b) The state injection circuits for implementing $T$ and $CCZ$ gates using input magic states.
		We assume that all components of the circuit are ideal except for the preparation of the magic states, which we model as ideal preparation followed by the noise channels implicit in \cref{eq:rho_diag,eq:tau_diag}.
	}
	\label{fig:10_1_2_protocol}
\end{figure*}

In an MSD protocol~\cite{knill2004,bravyi2005}, one starts with some number of noisy magic states and distils these states into fewer magic states of higher quality.
Magic states are used to implement challenging non-Clifford gates in fault-tolerant quantum computing architectures and hence finding efficient MSD protocols is an ongoing area of research~\cite{bravyi2012,campbell2017a,hastings2018,haah2017,haah2018a}.
In this section, we use our morphing procedure to construct a family of codes with applications to MSD.

%%%%%%%%%%%%%%%%%%%%%%%%%%%%%%%%%%%%%%%%
%% Morphing quantum Reed-Muller codes %%
%%%%%%%%%%%%%%%%%%%%%%%%%%%%%%%%%%%%%%%%
\subsection{Morphing quantum Reed-Muller codes \label{subsec:qrm}}

Quantum Reed-Muller (QRM) codes~\cite{steane1999,anderson2014} are a generalization of classical Reed-Muller codes~\cite{macwilliams1977}.
We apply our morphing procedure to QRM codes of distance three, where the codes in this subfamily have parameters $\nkd{2^{d+1}-1}{1}{3}$ for $d\geq 2$.
The $d=2$ code is actually the Steane code discussed in \cref{ex:steane}. 
In this section, we concentrate on the $d=3$ code, often called the 15-qubit Reed-Muller code, which has a transversal logical $T$ gate.
We defer the discussion for general $d$ until \cref{app:qrm}.

\Cref{sfig:15qrm,sfig:morph_15qrm} illustrate the \nkd{15}{1}{3} Reed-Muller code and its morphed counterpart, a \nkd{10}{1}{2} code. 
One can select a subset, $R$, of the physical qubits of the \nkd{15}{1}{3} code such that $\mathcal S(R)$ defines an \nkd{8}{3}{2} code (the so-called `smallest interesting color code'~\cite{kubica2015a,campbellblog}). 
Morphing the \nkd{15}{1}{3} code with this subset of qubits gives a \nkd{10}{1}{2} code that inherits the logical gates of the \nkd{15}{1}{3} code and therefore has a logical $T$ gate.
This gate is not transversal (with respect to a single-qubit partition) but it is fault-tolerant as any single fault during the implementation of the gate can only lead to a detectable error.

%%%%%%%%%%%%%%%%%%%%%%%%%%%%%
%% [[10,1,2]] MSD protocol %%
%%%%%%%%%%%%%%%%%%%%%%%%%%%%%
\subsection{\texorpdfstring{$\nkd{10}{1}{2}$}{[[10,1,2]]} MSD protocol \label{subsec:msd_protocol}}

Here we compare the MSD performance of the \nkd{15}{1}{3} code with its morphed counterpart, the \nkd{10}{1}{2} code. 
We denote the relevant magic states by 
\begin{equation}
\begin{split}
	&\ket {T} = T \ket +, \\
	&\ket {CCZ} = CCZ \ket {+++}.
	\label{eq:magic-states}
\end{split}
\end{equation}
In the 15-to-1 protocol~\cite{bravyi2005}, one distils fifteen $\ket{T}$ states into one higher-quality $\ket{T}$ state using the \nkd{15}{1}{3} code.
However, in the 10-to-1 protocol, we instead distill seven $\ket{T}$ and one $\ket{CCZ}$ states into one higher-quality $\ket{T}$ state using the \nkd{10}{1}{2} code; see \cref{fig:10_1_2_protocol}. 
To our knowledge, our 10-to-1 protocol is the first example of an MSD protocol with more than one type of magic state as input. 

We proceed with the standard analysis~\cite{bravyi2005,bravyi2012,paetznick2013a,chamberland2020b} of the 10-to-1 protocol.
We assume that the only noise in the protocol comes from the input magic states; see \cref{fig:10_1_2_protocol}.
This assumption is justified by the fact that the remaining parts of the protocol are Clifford circuits, preparation of Pauli eigenstates, and single-qubit measurements. 
It is usually straightforward to make these operations fault-tolerant using quantum error correcting codes.

\begin{figure*}
	\subfloat[]{
		\centering
		\includegraphics[width=0.425\linewidth,valign=t]{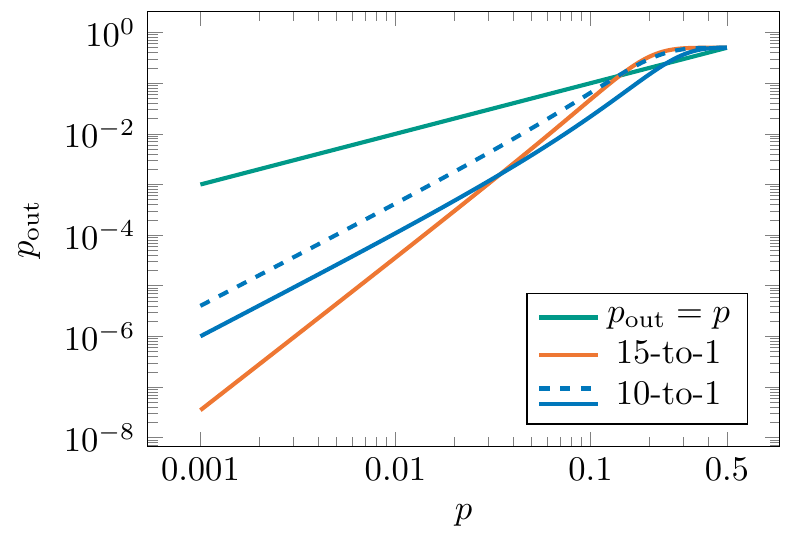}
		\label{sfig:pout}
	}
	\quad
	\subfloat[]{
		\centering
		\includegraphics[width=0.425\linewidth,valign=t]{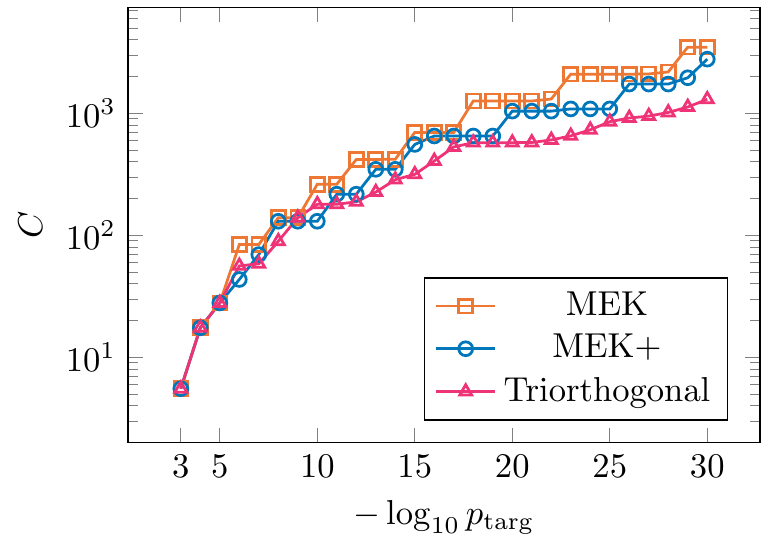}
		\label{sfig:cost}
	}
	\caption{
		(a) A plot showing the output error $p_{\mathrm{out}}$ as a function of the input error $p$ for the 15-to-1 and 10-to-1 MSD protocols. The solid and dashed blue lines correspond to optimistic (\cref{eq:tau_diag}) and pessimistic (\cite{jones2013a}) noise models for the $\ket{CCZ}$ state, respectively.
		(b) A plot showing the distillation cost $C$ for different target error rates $p_{\mathrm{targ}}$, assuming an input error rate of $p=0.01$ and the optimistic noise model for $\ket{CCZ}$ states.
		The cost for our MEK+ protocol (blue) is always smaller than the MEK protocol cost (orange), and even beats the triorthogonal protocol cost (magenta) for some values of $p_{\mathrm{targ}}$.
	}
	\label{fig:msd_compare}
\end{figure*}

We start with seven copies of a mixed one-qubit state $\rho$ such that $\bra{T} \rho \ket{T} = 1-p$ and one mixed three-qubit state $\tau$ such that $\bra{CCZ} \tau \ket{CCZ} = 1-p$, where we refer to $p$ as the input error. 
We assume that $\rho$ is diagonal in the $\ket T$ basis, i.e.,
\begin{equation}
	\rho = (1 - p) \ketbra{T}{T}+ p\, Z \ketbra{T}{T} Z,
	\label{eq:rho_diag}
\end{equation}
which can always be accomplished using Clifford twirling. 
And we assume that $\tau$ has the following form, 
\begin{equation}
	\begin{split}
	\tau &= 
	(1 - p) \ketbra{CCZ}{CCZ} \\
	&+ \frac{p}{7} \sum_{\substack{\bm{b} \in \{0,1\}^3 \\ |\bm{b}| \neq 0}} Z^{\bm{b}} \ketbra{CCZ}{CCZ} Z^{\bm{b}},
	\end{split}
    \label{eq:tau_diag}
\end{equation}
where $|\bm{b}|$ denotes the Hamming weight of $\bm{b} = (b_1, b_2, b_3)$ and $Z^{\bm{b}} = Z_1^{b_1} Z_2^{b_2} Z_3^{b_3}$.
We use this noise model because in many methods of preparing noisy $\ket{CCZ}$ states errors such as $ZII$ have a comparable probability to higher weight errors such as $ZZZ$~\cite{chamberland2020b}. 
In addition, as $\{ Z^{\bm{b}}\ket{CCZ} : \bm{b} \in \{0,1\}^3 \}$ is an orthonormal basis, we can express a $\ket{CCZ}$ state subject to $X$ and $Z$ noise as a $\ket{CCZ}$ subject to $Z$ noise only.
Furthermore, we can eliminate the off-diagonal terms in the density matrix via Clifford twirling to bring the noisy $\ket{CCZ}$ state into the form shown in \cref{eq:tau_diag} (though the coefficients may be different).

The probability that the protocol succeeds (the syndrome is trivial) is 
\begin{equation}
	p_{\mathrm s} = 1 - 8p + 29 p^2 + O(p^3).
	\label{eq:1012_psucc}
\end{equation}
When the protocol is successful, we obtain the output state $\rho_{\mathrm{out}} = (1-p_{\mathrm{out}}) \ketbra{T}{T} + p_{\mathrm{out}} Z \ketbra{T}{T} Z$, where 
\begin{equation}
	p_{\mathrm{out}} = p^2 + 9 p^3 + O(p^4).
	\label{eq:1012_pout}
\end{equation} 
We refer to $p_{\mathrm{out}}$ as the output error.
The leading order prefactor in \cref{eq:1012_pout} is exceptionally small; in the 15-to-1 protocol the output error is $35 p^3 + O(p^4)$~\cite{bravyi2005}.
This is due to an unusual property of the \nkd{10}{1}{2} code. 
Namely, for each of the seven possible $Z^{\bm{b}}$ errors on the input $\ket{CCZ}$ state, there is only one logical $Z$ operator that can be constructed by applying $Z$ to one of the $\ket T$ states. 

In \cref{sfig:pout} we compare the performance of the 10-to-1 protocol with the 15-to-1 protocol. 
The success probability of the 10-to-1 protocol is always greater than or equal to the success probability of the 15-to-1 protocol, and for $p \geq 0.034$ the 10-to-1 protocol has a smaller output error.
The high success probability of the 10-to-1 protocol suggests it may be useful in early rounds of distillation, as part of a multi-round protocol.
To test this hypothesis, we compute the distillation cost~\cite{bravyi2012} $C$ (the expected number of magic states required) to achieve a given target error rate $p_{\mathrm{targ}}$ assuming an input error rate of $p = 0.01$, for a number of multi-round MSD protocols.
Following~\cite{meier2012}, we define the Meier, Eastin, and Knill (MEK) protocol to consist of up to five rounds of distillation, where in each round either 15-to-1~\cite{bravyi2005} or 10-to-2~\cite{meier2012} distillation can be used. 
And we define the MEK+ protocol in the same way as the MEK protocol, except at each round we can choose between 15-to-1, 10-to-2 and 10-to-1 distillation.
In \cref{sfig:cost} we compare the distillation costs of the MEK and MEK+ protocols with the cost of a protocol where at each round we can choose between distillation protocols based on triorthogonal codes~\cite{bravyi2012}. 
We find that the MEK+ protocol distillation cost is always less than or equal to the MEK protocol cost, and even betters the triorthogonal protocol cost for $p_{\mathrm{targ}} \in \{ 10^{-6}, 10^{-9}, 10^{-10} \}$; see \cref{tab:cost} in \cref{app:cost} for the full data including the optimal distillation sequences.

One may question the reasonableness of the noise model assumed in \cref{eq:tau_diag} as it seems plausible that $\ket{CCZ}$ states could be more difficult to prepare than $\ket{T}$ states.
One option for preparing $\ket{CCZ}$ states is to use the circuit of Jones~\cite{jones2013a}, which requires four $\ket{T}$ states to produce a $\ket{CCZ}$ state. 
Assuming the $\ket T$ states are as in \cref{eq:rho_diag}, the output error of the \nkd{10}{1}{2} protocol will be $p_{\mathrm{out}} = 4p^2 + 21p^3 + O(p^4)$.
In this case, the distillation cost of the MEK+ protocol does not improve on the cost of the MEK protocol.

Definitively answering the question of which MSD protocol is preferable depends on the error model of quantum computing architecture under consideration, for example see~\cite{litinski2019a,chamberland2020b}.
However, often simpler MSD protocols using fewer qubits are preferable~\cite{litinski2019,litinski2019a,gidney2019a}, which suggest that our \nkd{10}{1}{2} protocol may transpire to be useful in practice. 
Furthermore, we have presented a single example of using morphing to modify an existing MSD protocol, but we conjecture that morphing could be fruitfully applied to other MSD protocols, e.g., those based on triorthogonal codes~\cite{bravyi2012,campbell2017,haah2018a}.

%%%%%%%%%%%%%%%%%%%%%%%%%%%%%
%% Morphing the color code %%
%%%%%%%%%%%%%%%%%%%%%%%%%%%%%
\section{Morphing the color code \label{sec:morph_color}}

In this section, we apply our morphing procedure to the color code~\cite{bombin2006,bombin2007a,kubica2018a}, a family of topological codes.
We first consider 2D color codes, before discussing the generalization to three and higher dimensions. 
Then we explain how morphing naturally leads to a family of HCT codes. 
Next we comment on the relation of our results to previous work relating the color code and the toric code. 
And lastly, we detail the fault-tolerant logical gates of HCT codes, which in three or more dimensions include non-Clifford gates. 

%%%%%%%%%%%%%%%%%%%%%%%%%%%%%%%%%%%%%%%%%%%%%%%%%%%%
%% Morphing color codes in two or more dimensions %%
%%%%%%%%%%%%%%%%%%%%%%%%%%%%%%%%%%%%%%%%%%%%%%%%%%%%
\subsection{Morphing color codes in two or more dimensions \label{subsec:morph_2d_cc}}

We use the dual lattice picture of color codes.
In this picture, a 2D color code lattice $\mathcal L$ is a triangulation of a 2D manifold (possibly with boundary) whose vertices are three-colorable.
That is, each vertex of the lattice can be assigned one of three colors ($r$, $g$, or $b$), such that any two vertices sharing an edge have different colors.
Let $\mathcal L$ denote such a lattice and let $\mathcal L_0$, $\mathcal L_1$, and $\mathcal L_2$ denote its vertices, edges, and faces, respectively. 
The edges and faces of the lattice inherit the colors of their vertices, e.g., an edge linking an $r$-vertex and a $g$-vertex is an $rg$-edge.
The qubits of the code are on the faces of $\mathcal L$ and stabilizer generators are associated with interior vertices. 
That is, for each vertex $v$ not on the boundary of $\mathcal L$, we have the stabilizers $X(v) = \prod_{f \supseteq v} X_f$ and $Z(v) = \prod_{f \supseteq v} Z_f$, where $X_f$ and $Z_f$ denote Pauli $X$ and $Z$ operators acting on the qubit on face $f$, respectively, and we are implicitly interpreting faces as subsets of vertices.

To apply our morphing procedure, we need to select subsets of the color code qubits. 
The subsets we choose are disk-like regions of $\mathcal L$, but we refer to them as ball-like in anticipation of generalization to three (and higher) dimensions. 
Namely, for each vertex $v$ in $\mathcal L$, we define the ball-like region
\begin{equation} 
	\mathcal B^v = \bigcup_{k=0}^d \{ \kappa \in \mathcal L_k \mid \kappa \supseteq v \},
	\label{eq:ball}
\end{equation}
where $d=2$.
We define the color of a ball-like region $\mathcal B^v$ to be the color of its central vertex $v$.
We again use the notation $\mathcal B^v_0$ to refer to the vertices of $\mathcal B^v$ etc.
Restricting the stabilizer group to a ball-like region $\mathcal B^v$ gives us a child code with parameters $\nkd{|\mathcal B^v_2|}{|\mathcal B^v_2|-2}{2}$, as the code has two stabilizer generators that act on all the qubits.  
We call such codes ball codes. 
In the interests of brevity, we use the phrase `morphing a ball-like region of a color code' as a shorthand for morphing a color code with the subset $R$ of the physical qubits given by the qubits contained in a ball-like region of the code.
\Cref{fig:morph} shows an example of morphing a ball-like region of a 2D color code.

\begin{figure}
	\centering
	\includegraphics[width=0.9\linewidth]{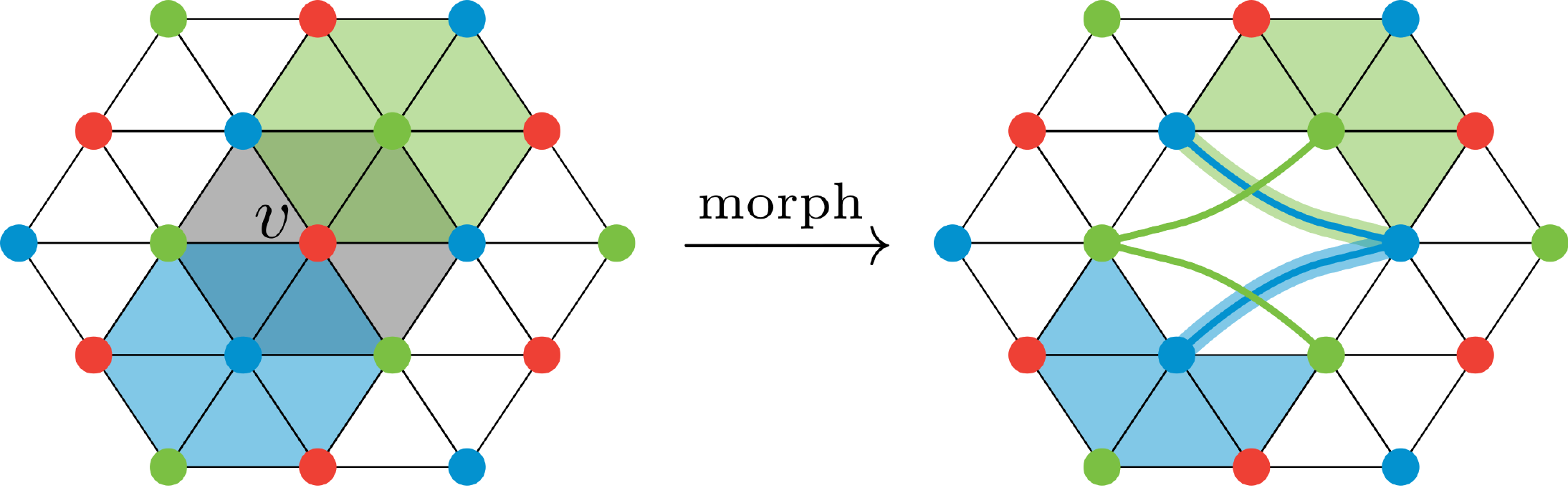}
	\caption{
		Morphing a ball-like region of a 2D color code.
		Qubits are on triangles and stabilizers are associated with vertices. 
		On the left, we shade the ball-like region $\mathcal B^v$, which has an associated \nkd{6}{4}{2} ball code.
		We also shade the support of an $X$-type and $Z$-type stabilizer in blue and green, respectively. 
		On the right, we show the new lattice formed by morphing $\mathcal B^v$, replacing the original qubits with the four logical qubits of the ball code (blue and green edges).
		This changes the support of the highlighted stabilizers, as shown.
	}
	\label{fig:morph}
\end{figure}

We now introduce a canonical geometric picture of morphing, which will prove especially useful when we consider decoding HCT codes in \cref{sec:decoding}.
The following corresponds to a particular choice of logical Pauli basis for ball codes, which we call the canonical basis (see \cref{app:morph_ddim}). 
We emphasize that other basis choices are equally valid and could be preferable in different contexts. 
Consider some ball-like region $\mathcal B^v$ of a 2D color code lattice $\mathcal L$. 
We note that the vertices on the boundary of $\mathcal B^v$ always have two colors.
When we morph $\mathcal B^v$ we remove the faces of $\mathcal B^v$ from $\mathcal L$ and add new edges as follows. 
For each boundary vertex color $c$, we choose one of the $c$-vertices on the boundary of $\mathcal B^v$ and connect it via an edge to every other $c$-vertex on the boundary of $\mathcal B^v$, where we refer to the new edges as $cc$-edges. 
We place a qubit on each new $cc$-edge; these qubits are the logical qubits of the ball code. 
In this picture, morphing takes $X$-type stabilizers of the parent color code to operators consisting partly of color code $X$-type stabilizers and partly of toric code $X$-type stabilizers (and similarly for $Z$-type stabilizers).
See \cref{fig:morph} for an example.
We defer a detailed discussion of the canonical geometric picture (including its generalization to higher dimensions) until \cref{app:morph_ddim}.

Morphing higher dimensional color codes works in much the same way as the 2D case. 
We now briefly review the definition of the $d$-dimensional color code.
Let $\mathcal L$ be a $d$-dimensional lattice formed by attaching $d$-simplices along their $(d-1)$-dimensional faces. 
We require that the vertices of $\mathcal L$ are $(d+1)$-colorable.
We place a qubit on each $d$-simplex of $\mathcal L$ and we associate $X$ and $Z$-type stabilizer generators respectively with vertices and $(d-2)$-simplices not on the boundary of $\mathcal L$. 
In 3D, qubits are on tetrahedra, $X$-type stabilizer generators are associated with vertices and $Z$-type stabilizers are associated with edges. 

In much the same way as in 2D, in higher dimensions we again morph ball-like regions of the color code. 
For a given vertex $v \in \mathcal L_0$, the $d$-dimensional ball-like region $\mathcal B^v$ consists of all simplices that contain $v$, alongside $v$ itself; see \cref{eq:ball}.
\Cref{fig:3d_ball_ex} shows an example of such a ball-like region in a 3D color code.
Restricting the stabilizer group of the color code to a ball-like region $\mathcal B^v$ gives us a $d$-dimensional ball code, which we characterize in the following lemma (whose proof we defer until \cref{app:morph_ddim}).
\begin{restatable}{lemma}{params} 
	The ball code defined on the $d$-dimensional ball-like region $\mathcal B^v$ has parameters 
	\begin{equation}
		N=|\mathcal B^v_d|, \quad K = |\mathcal B^v_1| - d, \quad D = 2.
		\label{eq:ball_code_params_ddim}
	\end{equation}
	\label{lem:ball_code_params}
\end{restatable}
We note that the canonical geometric picture can be generalized to higher dimensions, see \cref{fig:3d_ball_ex} and \cref{app:morph_ddim}.

\begin{figure}
	\centering
	\includegraphics[width=0.7\linewidth]{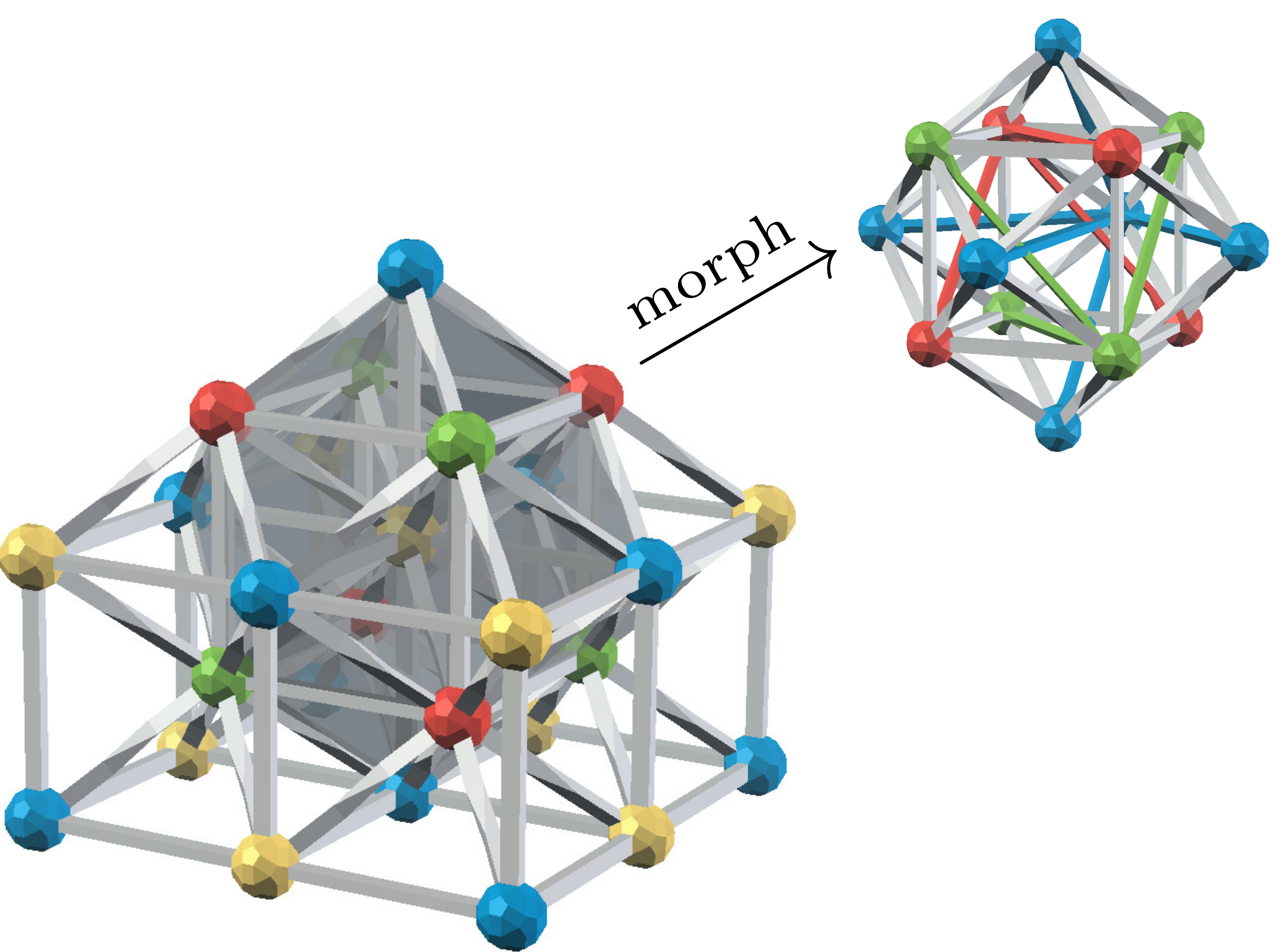}	
	\caption{
		Morphing a ball-like region of a 3D color code. 
		Qubits are on tetrahedra, $X/Z$ stabilizers are associated with vertices and edges, respectively.
		We shade a ball-like region (with an associated \nkd{24}{11}{2} ball code) and show its transformation under morphing.
	}
	\label{fig:3d_ball_ex}
\end{figure}

%%%%%%%%%%%%%%%%%%%%%%%%%%%%%%
%% Hybrid color-toric codes %%
%%%%%%%%%%%%%%%%%%%%%%%%%%%%%%
\subsection{Hybrid color-toric codes \label{subsec:hct}}

We use morphing to construct a family of HCT codes in $d$-dimensions, for $d\geq 2$.
Given an initial color code defined on a lattice $\mathcal L$, we obtain different HCT codes by morphing subsets of the ball-like regions of $\mathcal L$. 
Concretely, suppose that $\mathcal L$ is a closed manifold and consider the case where we morph a fraction $q \in [0,1]$ of $r$ ball-like regions (the choice of color is arbitrary).
For $q=0$ we of course have the original color code and for $q=1$ we obtain $d$ copies of the toric code (up to adding/removing ancilla qubits). 
Let $\widetilde{\mathcal L}$ denote the HCT code lattice produced from $\mathcal L$ via morphing. 
For $q=1$, qubits will be exclusively associated with edges as all the original color code qubits were contained in some $r$ ball-like region of $\mathcal L$.
As we show in \cref{app:morph_ddim}, the $X$ and $Z$-type stabilizers of $\widetilde{\mathcal L}$ can be partitioned into $d$ subgroups such that each subgroup only contains operators acting on $cc$-edges, where $c$ is any of the $d+1$ colors except for $r$. 
In 2D, each of these subgroups defines a toric code but in higher dimensions to obtain toric codes we need to introduce ancillas, unless we consider lattices with special structure (see \cref{app:coxeter}). 
We can also consider the case where we morph a fraction $q$ of all the ball-like regions in $\mathcal L$, but then we need to be careful because ball-like regions of different colors can overlap; see \cref{sec:decoding}.

In $d\geq 3$ dimensions color codes have transversal non-Clifford gates~\cite{bombin2007a,bombin2015,kubica2015,kubica2015a,watson2015,bombin2018a}, which our HCT codes inherit.
As we show in \cref{app:non_cliff_ddim}, if a $d$-dimensional color code has a transversal non-Clifford logical operator $L$ implemented by $R_d^{\pm 1}$ gates, then HCT codes obtained from this code will have a fault-tolerant logical operator $L$ implemented by $R_d^{\pm 1}$ and multi-control-$Z$ gates.
We recall that $R_d = \exp (\frac{i \pi}{2^d} Z)$ and that the $d$-qubit multi-control-$Z$ gate is specified by its action on the states in the computational basis, namely
\begin{equation}
	MCZ \ket{b_1,\ldots,b_d} = \left(1-2\prod_{i=1}^d b_i\right)\ket{b_1,\ldots,b_d}
	\label{eq:mcz_def}
\end{equation}

We now consider some examples of HCT codes with fault-tolerant non-Clifford gates.
The \nkd{10}{1}{2} code introduced in \cref{subsec:qrm} is one such example as the \nkd{15}{1}{3} code fits into the color code family~\cite{kubica2015}.
Furthermore, we can construct a HCT code with a fault-tolerant logical $R_d$ gate from any color code defined on a cellulation of a $d$-simplex.  
Another example is given by the family of hypercubic color codes~\cite{kubica2015a}, which we illustrate in \cref{fig:hypercubic}.
Hypercubic color codes have transversal logical $MCZ$ gates and therefore HCT codes obtained from these codes have fault-tolerant logical $MCZ$ gates.

\begin{figure}
	\centering
	\subfloat[]{
		\centering
		\includegraphics[width=0.4\linewidth,valign=b]{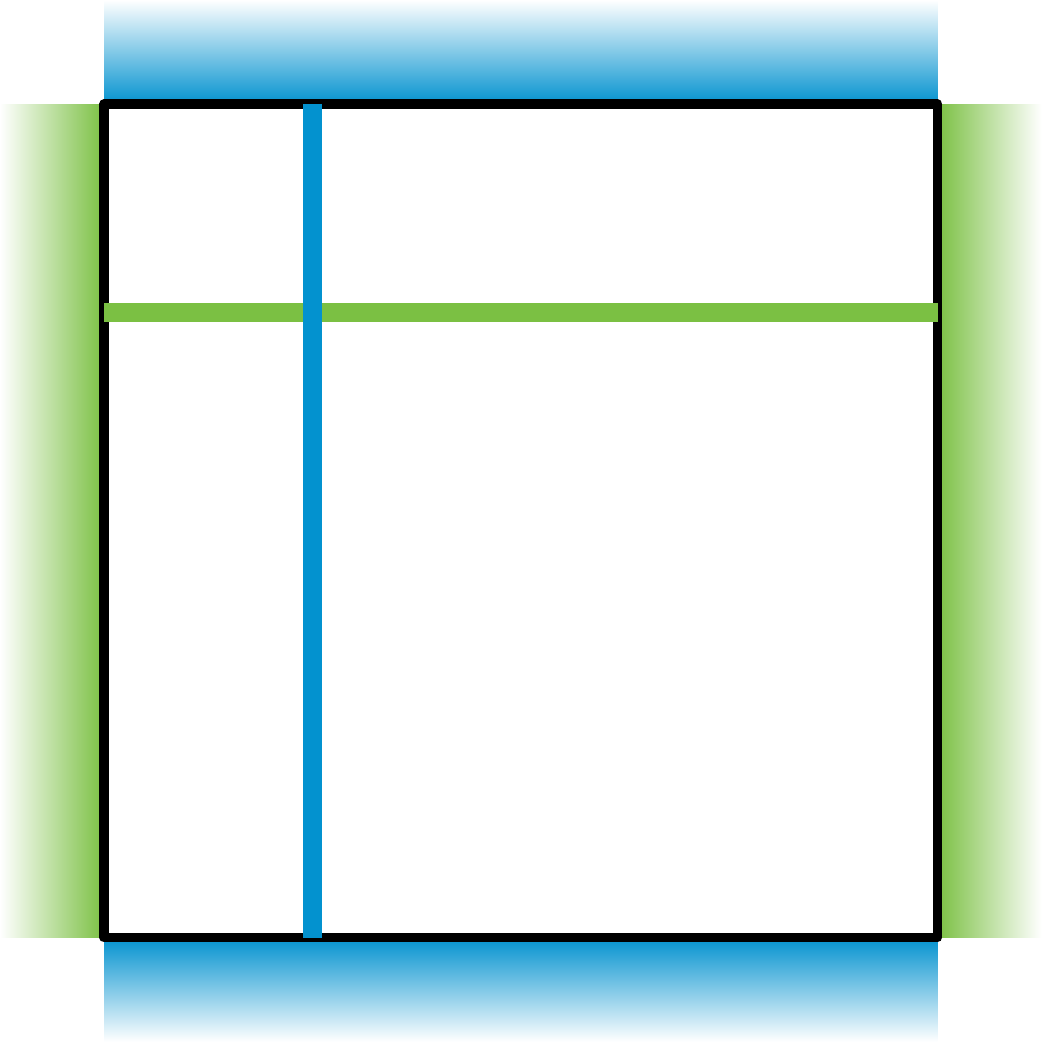}
	}
	\quad
	\subfloat[]{
		\centering
		\includegraphics[width=0.4\linewidth,valign=b]{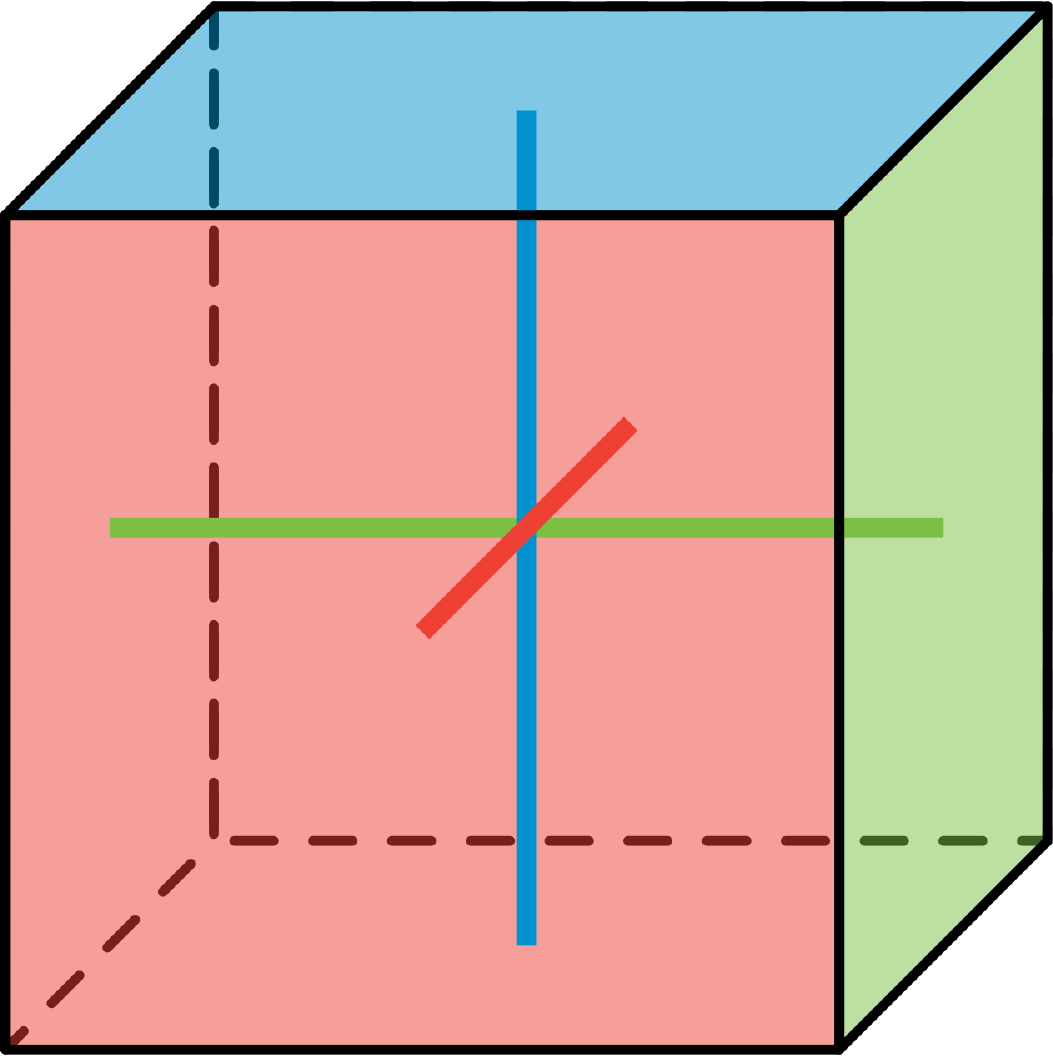}
	}
	\caption{
		Hypercubic color codes are defined on lattices that tessellate the hypercube, where opposite boundaries have the same color.
		(a) The $d=2$ case. 
		Logical $X_1$ and $Z_2$ are supported on the green line (as products of $X$ and $Z$ operators, respectively).
		Similarly, $X_2$ and $Z_1$ are supported on the blue line.
		(b) The $d=3$ case. 
		Logical $X$ operators are supported on the boundaries and the corresponding logical $Z$ operators connect boundaries of the same color. 
	}
	\label{fig:hypercubic}
\end{figure}

We now briefly discuss how our results relate to previous work. 
In Ref.~\cite{kubica2015a}, the authors proved that the $d$-dimensional color code is local-Clifford equivalent to $d$ copies of the $d$-dimensional toric code.
Their argument was non-constructive, i.e., they showed that for a $d$-dimensional color code defined on a lattice $\mathcal L$, there exists a Clifford isometry that transforms the color code into $d$ copies of the toric code. 
Furthermore, the isometry has the following form
\begin{equation}
	V = \bigotimes_{\substack{u \in \mathcal L_0 \\ \col{u} = c}} V_u,
	\label{eq:unfolding_V}
\end{equation}
where $V_u$ is a Clifford isometry acting on the physical qubits of the color code in the neighborhood of the vertex $u$.
Morphing can be understood as giving an explicit implementation of the local Clifford isometry $V$, as each morphing circuit is a Clifford isometry applied in the neighborhood of a vertex. 
In Refs.~\cite{criger2016} and \cite{vasmer2019}, the authors show that one can obtain color codes by concatenating multiple toric code with small error-detecting codes.
This is a special case of morphing, but in reverse. 
In Refs.~\cite{vasmer2019} and \cite{jochym-oconnor2021}, the authors present constructions of 3D and 4D toric codes with transversal logical $MCZ$ gates. 
We can use morphing to recover these results, as we explain in \cref{app:coxeter}.

%%%%%%%%%%%%%%%%%%%%%%%%
%% Decoding HCT codes %%
%%%%%%%%%%%%%%%%%%%%%%%%
\section{Decoding HCT codes \label{sec:decoding}}

In this section, we propose an efficient decoding algorithm for 2D HCT codes, and we numerically benchmark its performance. 
Remarkably, our decoder reduces to the MWPM decoder~\cite{dennis2002,fowler2015} in one limit and to the restriction decoder in the opposite limit. 
We recall that the restriction decoder is an efficient decoder of the color code in $d\geq 2$ dimensions that builds upon the ideas in Refs.~\cite{bombin2012}~and~\cite{delfosse2014}.

%%%%%%%%%%%%%%%%%%%%%%%%
%% Decoding algorithm %%
%%%%%%%%%%%%%%%%%%%%%%%%
\subsection{Decoding algorithm \label{subsec:decoder}}

\begin{figure*}
	\centering
	\subfloat[]{
		\centering
		\includegraphics[width=0.22\linewidth]{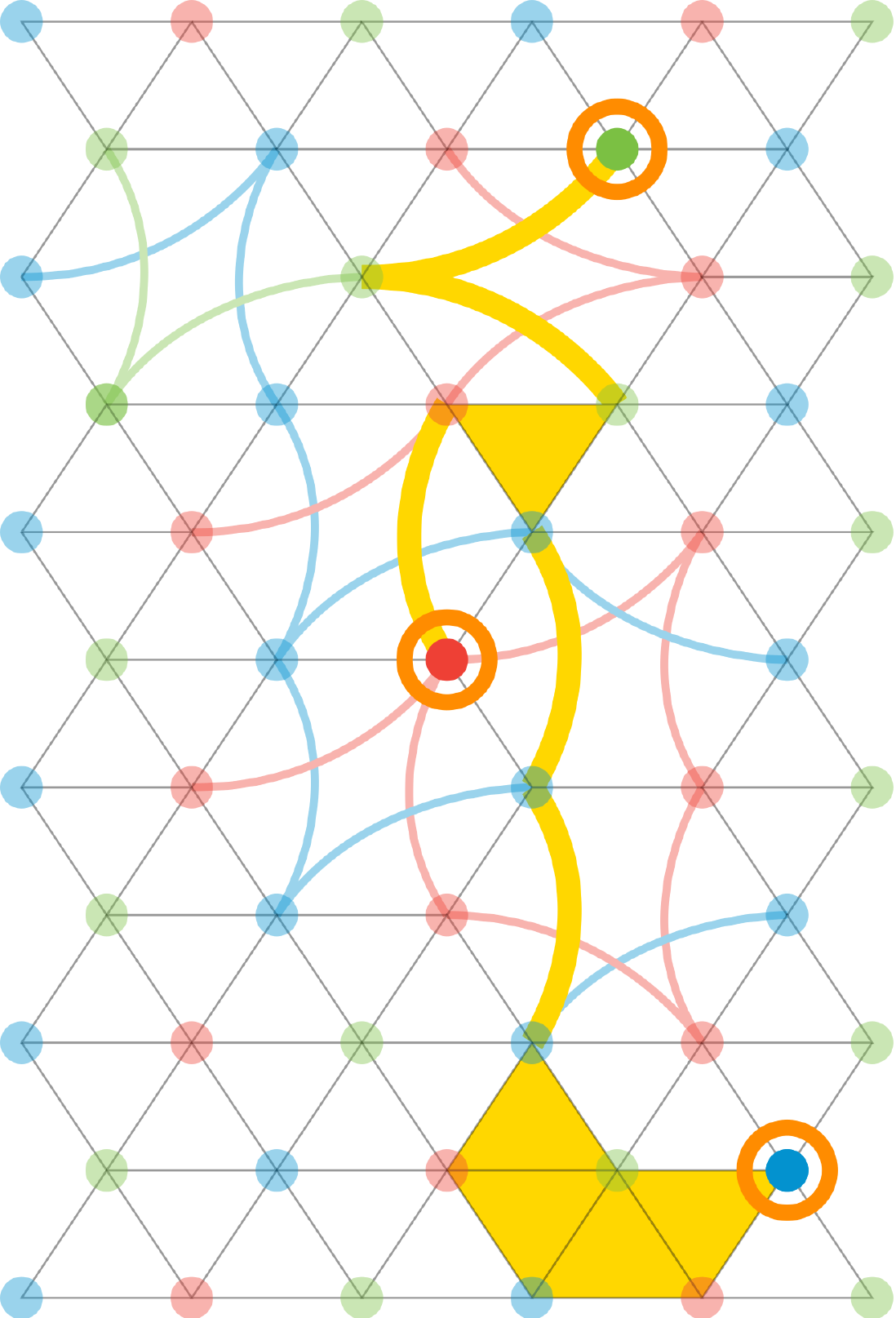}
	}
	\hspace{0.1mm}
	\subfloat[]{
		\centering
		\includegraphics[width=0.22\linewidth]{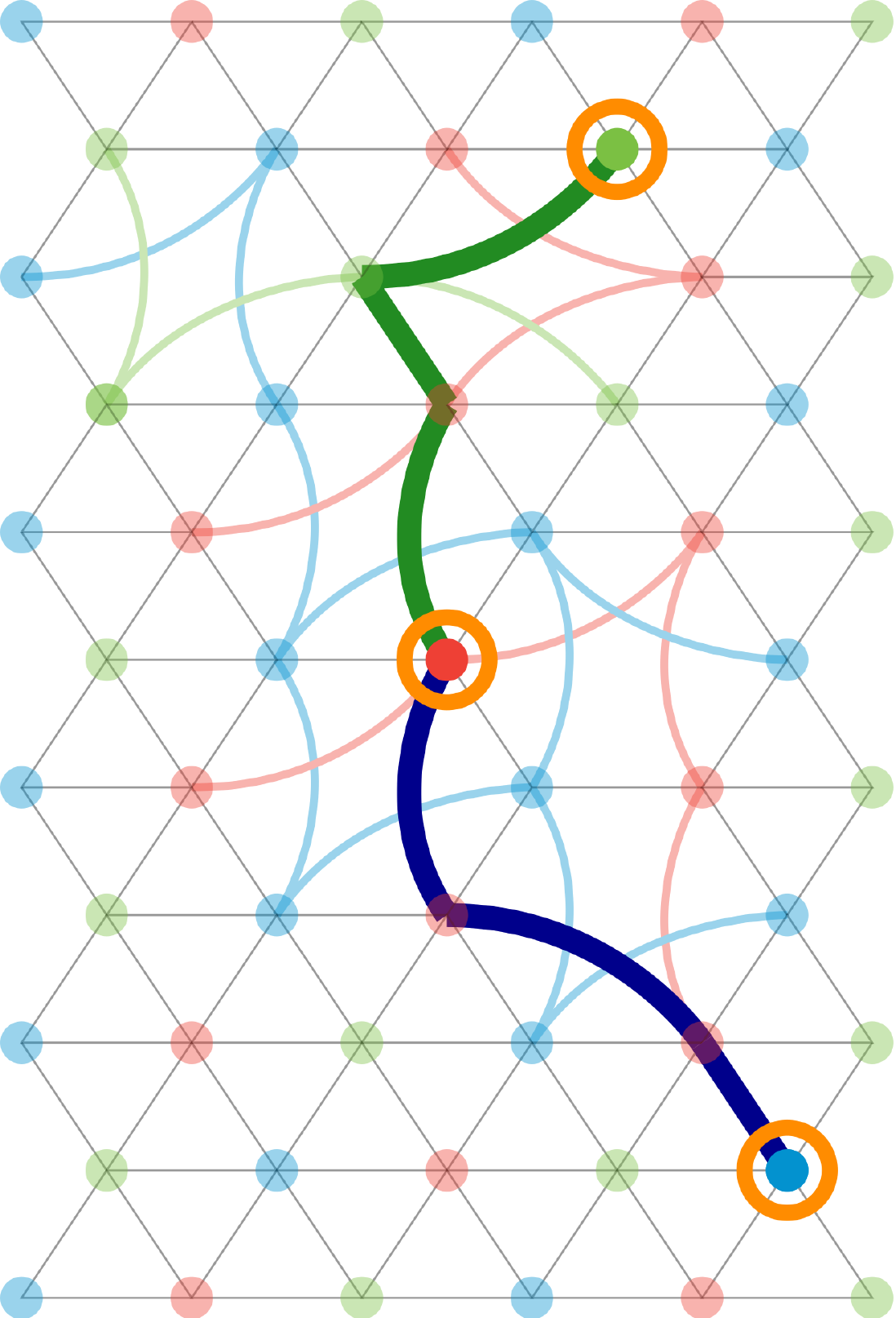}
	}
	\hspace{0.1mm}
	\subfloat[]{
		\centering
		\includegraphics[width=0.22\linewidth]{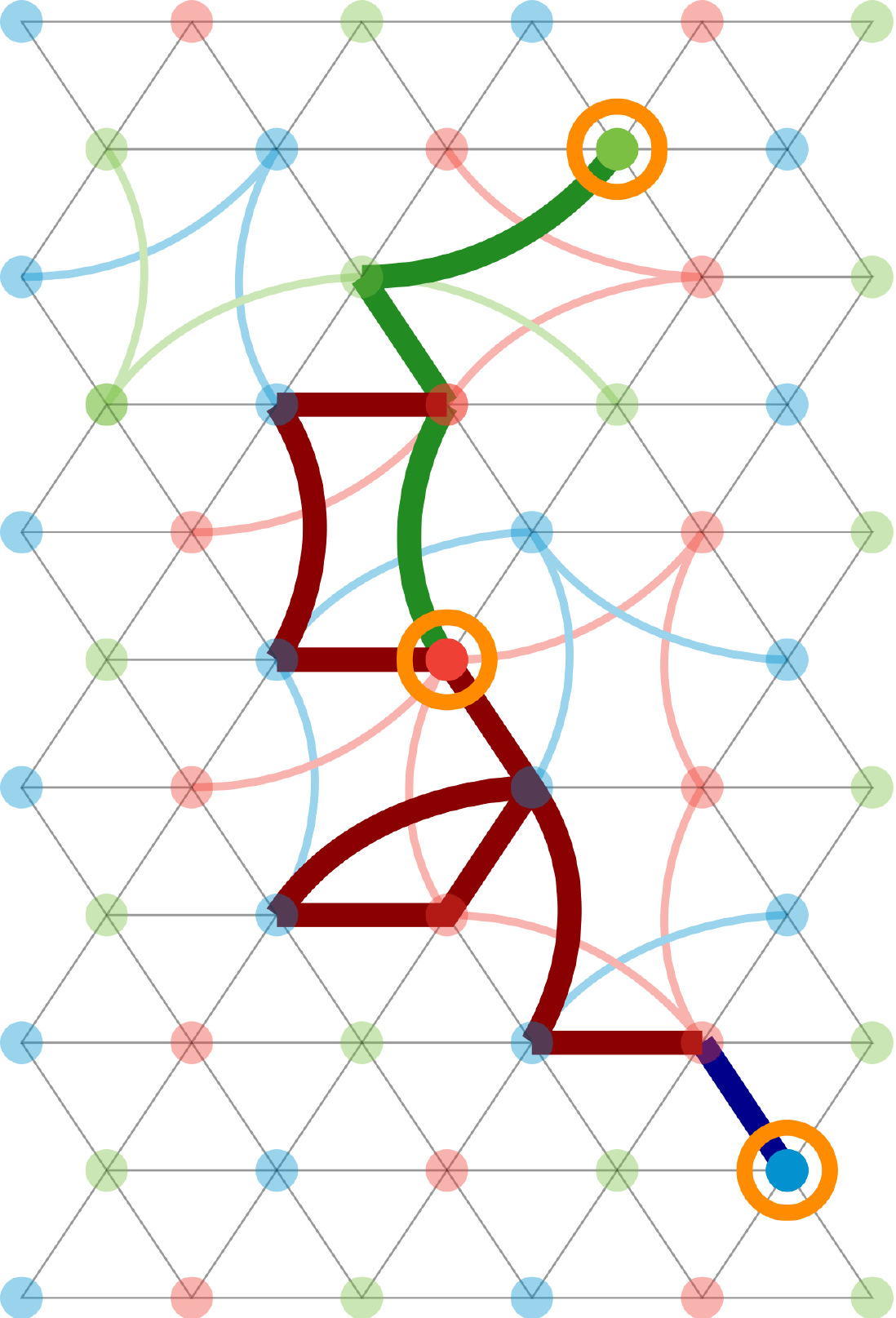}
	}
	\hspace{0.1mm}
	\subfloat[]{
		\centering
		\includegraphics[width=0.22\linewidth]{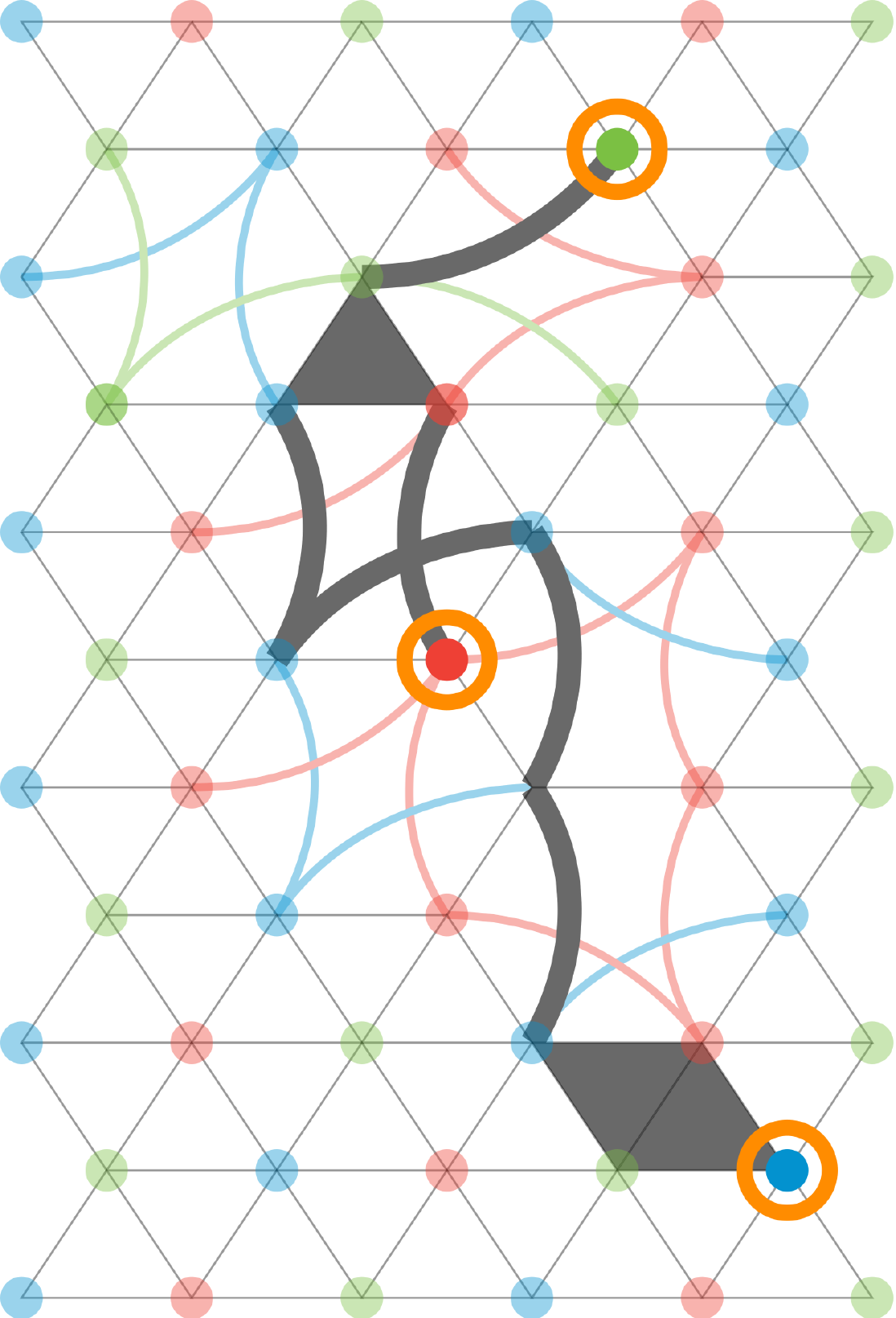}
	}
	\caption{
		An illustration of our decoding algorithm for HCT codes.
		(a) A phase-flip error (yellow faces and edges) and its syndrome (ringed vertices). 
		(b)-(d) The three steps of the HCT code decoder: 
		(1) Matching: use MWPM to compute a matching (dark green and blue edges) of the syndrome in the restricted lattices $\widetilde{\mathcal L}^{rg}$ and $\widetilde{\mathcal L}^{rb}$.
		(2) Local Modification: locally modify the matching; new edges included in the matching are highlighted in red.
		(3) Local Lift: use the local lifting procedure to compute a correction (grey faces and edges) from the modified matching. 
	}
	\label{fig:alg_ex}
\end{figure*}

We consider phase-flip noise as the correction of bit-flips is analogous and can be done separately.
Let $\widetilde{\mathcal L}$ denote a 2D HCT code lattice, produced from a color code lattice $\mathcal L$ (without boundary) by morphing some subset of its ball-like regions.
We label the three colors of the vertices by $r$, $b$, and $g$.
We use the canonical geometric picture from \cref{subsec:morph_2d_cc}, so qubits in $\widetilde{\mathcal L}$ are associated with $rgb$-faces, i.e., faces $f \in \widetilde{\mathcal L}_2 \cap \mathcal L_2$ and $cc$-edges, i.e., edges $e \in \widetilde{\mathcal L}_1 \setminus \mathcal L_1$.
And $X$-type stabilizer generators are associated with vertices $v \in \widetilde{\mathcal L}_0$.
We emphasize that $\widetilde{\mathcal L}$ should not be treated as a topological object but rather as a combinatorial object. 
We also make use of a well-known concept in color code decoding: restricted lattices~\cite{bombin2006,kubica2019}. 
Let $\widetilde{\mathcal L}^{c_1 c_2}$ be the lattice formed from $\widetilde{\mathcal L}$ by deleting all vertices of color $c \neq c_1 \neq c_2$, along with all edges and faces that contain these vertices. 

We can now state our decoding problem: given an error syndrome $\sigma \subseteq \widetilde{\mathcal L}_0$ caused by an unknown phase-flip error $\epsilon \subseteq (\widetilde{\mathcal L}_2 \cap \mathcal L_2) \cup (\widetilde{\mathcal L}_1\setminus\mathcal L_1)$, output a correction operator $\zeta \subseteq (\widetilde{\mathcal L}_2 \cap \mathcal L_2) \cup (\widetilde{\mathcal L}_1\setminus\mathcal L_1)$ with the same syndrome as $\epsilon$. 
Decoding is successful if $\zeta$ and $\epsilon$ differ by a stabilizer.
Our decoding algorithm consists of the following three steps.
\begin{enumerate}
	\item \textbf{Matching}: we use MWPM to pair up the vertices of $\sigma$ in the restricted lattices $\widetilde{\mathcal L}^{rg}$ and $\widetilde{\mathcal L}^{rb}$.
	That is, we obtain a set of edges $E^{rg} \subseteq \widetilde{\mathcal L}_1^{rg}$ whose boundary is exactly $\sigma \cap \widetilde{\mathcal L}_0^{rg}$ (and similarly for $\widetilde{\mathcal L}^{rb}$).
	In a slight abuse of terminology, we call such a set of edges a matching.
	\item \textbf{Local Modification}: we modify the matchings as follows, focusing on $E^{rg}$ (for $E^{rb}$ exchange $g \leftrightarrow b$ below).
	For each $rr$-edge $e \in E^{rg}$, we find the color $c$ of its corresponding ball-like region $\mathcal B^v$ in $\mathcal L$. 
	If $c \neq g$, we remove $e$ from the matching. 
	Let $u_1$ and $u_2$ be the endpoints of $e$.
	We find two vertices $w_1,w_2 \in \widetilde{\mathcal L}_0$ such that $(u_1, w_1, v)$ and $(u_2, w_2, v)$ are contained in $\mathcal L_2$ for some $v \in \mathcal L_0$. 
	Then we find a path $(e_1, \ldots, e_m)$ from $w_1$ to $w_2$ consisting of $gg$-edges.
	We add the edges $(u_1, w_1)$, $(u_2, w_2)$, $e_1, \ldots, e_m$ to the matching. 
	\item \textbf{Local Lift}: first, we add each $cc$-edge in the combined matching $E^{rg} \cup E^{rb}$ to the correction $\zeta$.
	Then, for each $r$-vertex $v \in e$ such that $e \in E^{rg} \cup E^{rb}$ we apply a local lift. 
	For regions of $\widetilde{\mathcal L}$ that are the same as the corresponding regions $\mathcal L$, the local lift is the color code lift described in~\cite{kubica2019}.
	In the morphed regions of $\widetilde{\mathcal L}$, the local lift must be modified to account for the changes due to morphing; see \cref{app:local_lift}.
	Let $\zeta_v \subseteq (\widetilde{\mathcal L}_2 \cap \mathcal L_2) \cup (\widetilde{\mathcal L}_1\setminus\mathcal L_1)$ be the output of the local lift at $v$. 
	We update the correction $\zeta = \zeta \bigtriangleup \zeta_v$.
\end{enumerate}

\Cref{fig:alg_ex} illustrates our algorithm via and example, and a C++ implementation of our algorithm (utilizing Blossom V~\cite{kolmogorov2009}) is available online~\cite{github}.
We can now explain our earlier statement that morphing can simplify the decoding problem. 
Morphing (applied judiciously) transforms a color code into copies of the toric code, changing the decoding problem from color code decoding to the simpler toric code decoding problem. 
In particular, the local lift in the toric code regions is  trivial, compared
with the local lift in the color code regions.
Therefore one can argue that the decoding problem is simplified in the regions of the lattice that have been morphed.

Our algorithm has runtime complexity $O(N^3)$ (see \cref{app:local_lift}), where $N$ is the number of physical qubits in the code, due to the MWPM subroutine.
As with the restriction decoder, MWPM can be replaced with any algorithm for 2D toric code decoding.
For example, replacing MWPM with union-find~\cite{delfosse2017} would give our algorithm almost linear time complexity. 
We note that our decoder is an instance of the `decoding using emergent symmetries' paradigm, which underlies many topological code decoders~\cite{kubica2019,brown2020b,bonillaataides2021}. 
As in the color code, the emergent symmetries in our case are the products of all stabilizers of the same color. 

One example application for our decoder is color code-toric code lattice surgery~\cite{PoulsenNautrup2017}.
There can be advantages to preparing magic states in a color code before transferring them to a toric code~\cite{shutty2022finding}. 
This transfer is accomplished via lattice surgery, during which the color code and toric code will be temporarily merged to form an HCT code that must be decoded as part of the lattice surgery protocol. 
For our decoder to be directly applicable to this situation, we would need to modify our algorithm to work for codes with boundaries, but we expect that this modification would proceed in a similar way to the restriction decoder~\cite{chamberland2020}.

\begin{figure*}
	\centering
	\subfloat[]{
		\centering
		\includegraphics[width=0.32\linewidth,valign=t]{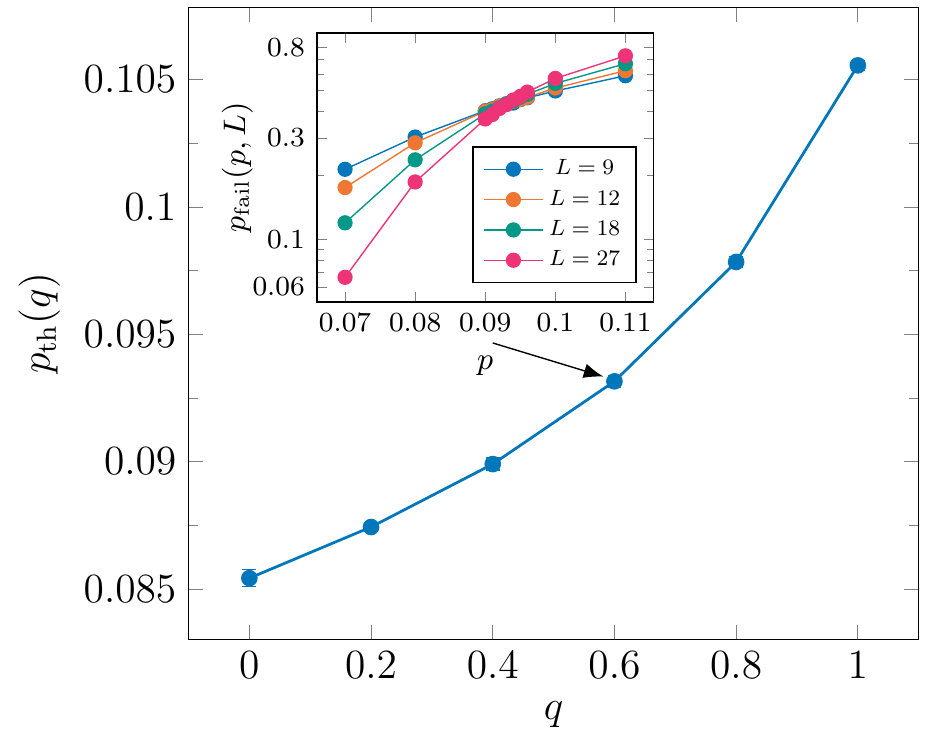}
		\label{sfig:pth_q_rr}
	}
	\subfloat[]{
		\centering
		\includegraphics[width=0.32\linewidth,valign=t]{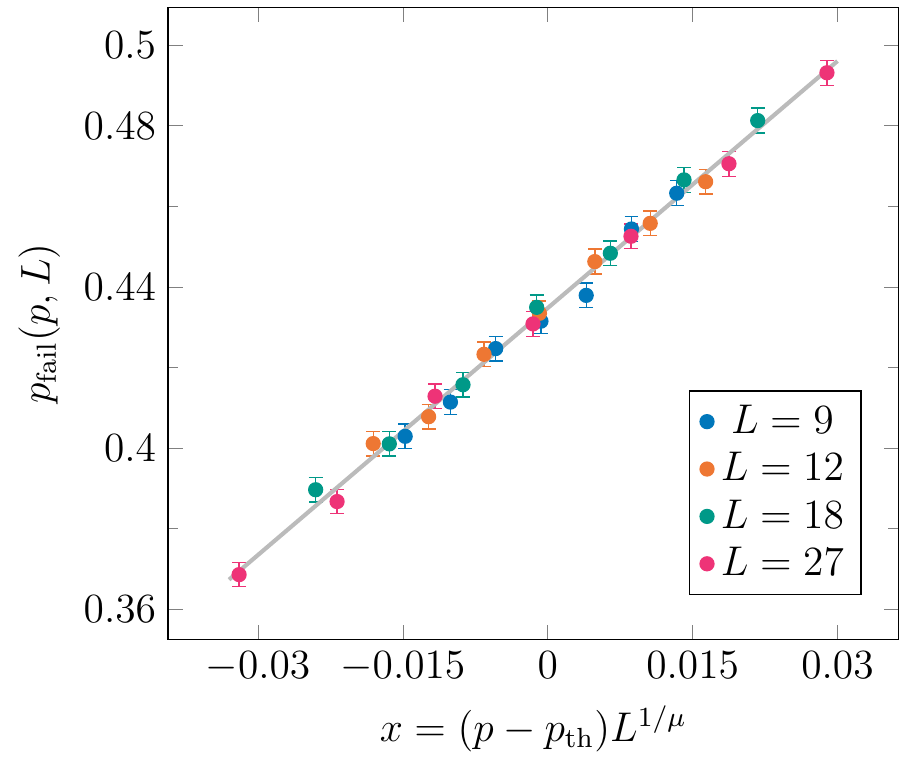}
		\label{sfig:collapse}
	}
	\subfloat[]{
		\centering
		\includegraphics[width=0.32\linewidth,valign=t]{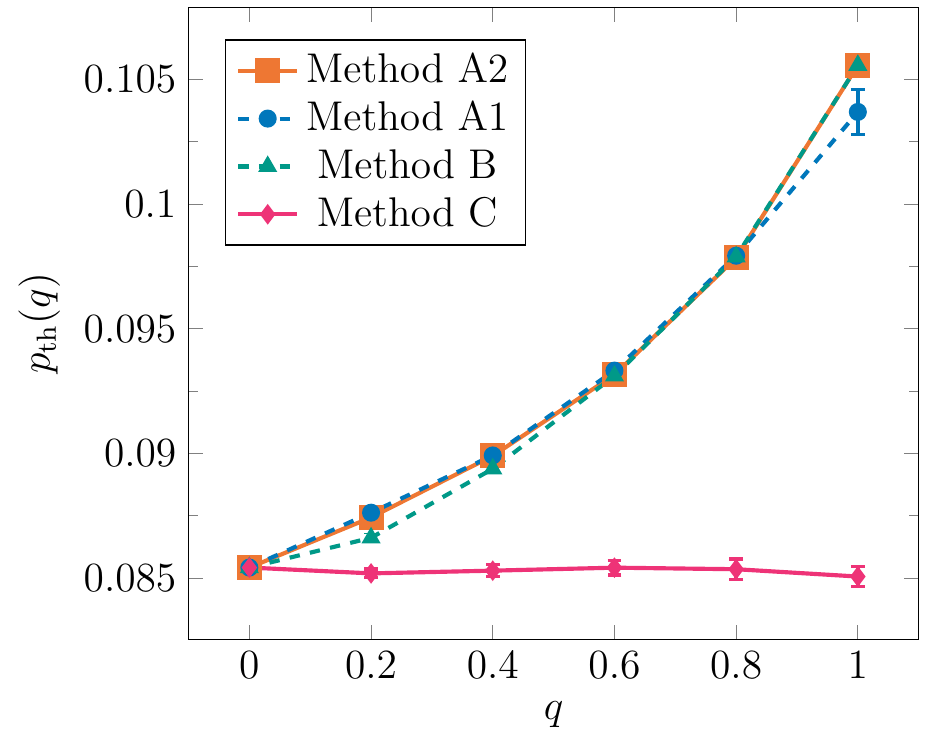}
		\label{sfig:ensemble_comparison}
	}
	\caption{
		(a) 
		A plot of the error threshold $p_{\mathrm{th}}$ as a function of $q$ for the ensemble of HCT codes produced using Method A1. 
		The inset shows a plot of the logical error rate $p_{\mathrm{fail}}$ as a function of the phase-flip error rate $p$ for $q=0.6$. 
		(b)
		The rescaled data for the inset in (a), using finite-size scaling analysis to estimate the threshold.
		The fitting parameters for the collapse are $p_{\mathrm{th}}=0.0932(1)$ and $\mu=1.42(8)$ and the grey line shows a quadratic fit.  
		(c)
		A comparison of the behavior of the error threshold as a function of $q$ for different methods of constructing HCT codes.
	}
	\label{fig:numerics}
\end{figure*}

%%%%%%%%%%%%%%
%% Numerics %%
%%%%%%%%%%%%%%
\subsection{Numerics}
\label{subsec:num}

We investigate the performance of our decoding algorithm using Monte Carlo simulations. 
We begin with the family of 2D color codes defined on $L \times L$ triangular tilings where six triangles meet at each vertex.
We use periodic boundary conditions, i.e., we identify the coordinates $(x,y)=(x+L,y)=(x,y+L)=(x+L,y+L)$.
From this starting point, we explore three ways of constructing HCT codes. 

\begin{enumerate}
	\item \textbf{Method A:} 
	We choose $q \in [0,1]$ and construct a random HCT code by morphing each $r$ ball-like region with probability $q$.
	We consider both choosing the same logical Pauli basis for each region (A1), and choosing the basis randomly (A2). 
	\item \textbf{Method B:}
	We choose $q \in [0,1]$ and iterate through the $r$ ball-like regions, morphing with probability $q$.
	Next we iterate through the $g$ ball-like regions and morph with probability $q$, except if an overlapping region has already been morphed. 
	Lastly, we iterate through the $b$ ball-like regions and morph with probability $q$ (again unless an overlapping region has already been morphed).
	In all cases we choose the logical Pauli basis randomly.
	\item \textbf{Method C:}
	We choose $q \in [0,1]$ and put the vertices of the lattice into a random order.
	For each vertex in the list, with probability $q$ we morph the ball-like region centred at the vertex (choosing the logical Pauli basis randomly), except if an overlapping region has already been morphed.
\end{enumerate}

By definition, the parameter $q$ captures the proportion of morphed $r$ ball-like regions in Methods A1, A2 and B. 
However in Method C the parameter $q$ does not correspond directly to the proportion of morphed ball-like regions because of the non-overlapping condition.
Furthermore, in the $q=1$ limit Methods A1, A2 and B all produce toric codes with certainty, whereas Method C with high probability does not produce a toric code as this would require all the vertices of one color to occur at the start of the random ordering.

We investigate the performance of our decoder as a function of $q$ for each of the above methods.
In each case, we use an error model where iid $Z$ errors affect the qubits with probability $p$ and stabilizer measurements are perfect. 
We use Monte Carlo simulations to estimate the logical error rate $p_{\mathrm{fail}}$ as a function of $p$, $L$, and $q$; see \cref{sfig:pth_q_rr}.
For each value of $q > 0$, we randomly generate 100 HCT code lattices for each $L$, and for each value of $p$ we carry out 1000 Monte Carlo trials per lattice. 
For each value of $q$, we estimate the error threshold $p_{\mathrm{th}} (q)$ using finite-size scaling methods~\cite{harrington2004}, as illustrated in \cref{sfig:collapse}.

For Method A1, we reproduce the known color code error threshold of $p_{\mathrm{fail}} (0) \approx 0.085$~\cite{delfosse2014} and the known toric code threshold of $p_{\mathrm{fail}} (1) \approx 0.103$~\cite{wang2003}. 
The threshold monotonically increases from $q=0$ to $q=1$; see \cref{sfig:ensemble_comparison}. 
For Method A2 we observe an increased threshold of $p_{\mathrm{fail}} (1) \approx 0.105$; see \cref{sfig:ensemble_comparison}
This is not surprising, as by picking the logical Pauli basis randomly we are highly unlikely to produce the standard square lattice toric code for $q=1$.

For Method B, we observe similar results to the case where we morphed only $r$ ball-like regions using random logical Pauli bases.   
And for Method C, we find that the error threshold does not vary with $q$ and is approximately equal to the color code error threshold of $p_{\mathrm{fail}} (0) \approx 0.085$.
The data for these cases are shown in~\cref{sfig:ensemble_comparison}.

%%%%%%%%%%%%%%%%
%% Conclusion %%
%%%%%%%%%%%%%%%%
\section{Conclusion \label{sec:conc}}

In this manuscript we introduced morphing as a systematic method for modifying quantum codes.
By applying morphing to quantum Reed-Muller codes, we obtained a family of codes with fault-tolerant non-Clifford gates and applications in MSD. 
We found that in the case of the color code, morphing allows us to change the implementation of fault-tolerant logical gates, reduce the weight of stabilizers and the number of physical qubits (at the cost of also reducing the code distance), and simplify the decoding problem.
One can apply morphing to tetrahedral color codes~\cite{bombin2007a,bombin2015,bombin2018a} to construct HCT codes with a fault-tolerant logical $T$ gate. 
The \nkd{10}{1}{2} code is actually the smallest example of this family, which is thoroughly explored in~\cite{iverson2020a,iverson2022}.

We examined the decoding problem for 2D HCT codes and we detailed an efficient decoding algorithm based on existing toric and color code decoders.
Our decoding algorithm achieves good performance across the whole family and we expect that our approach should generalize to higher dimensions. 
A natural next step would be to study more realistic noise models that include measurement errors.
We expect that in this case standard techniques such as repeating the stabilizer measurements~\cite{raussendorf2007a,wang2010} would work for HCT codes.
In addition, it would be interesting to investigate the optimal error threshold of 2D HCT codes using tensor network decoders~\cite{bravyi2014a,chubb2021}.
We expect that the optimal threshold will interpolate between the optimal thresholds of the color code and the derived toric codes. 
This would potentially explain why the optimal threshold of the color code defined on the triangular lattice~\cite{katzgraber2009} matches the optimal threshold of the toric code defined on the square lattice~\cite{honecker2001,merz2002,ohzeki2009}, as the latter can be obtained from the former via morphing.

Morphing can be straightforwardly generalized to stabilizer subsystem codes~\cite{poulin2005}. 
We again select a subset of the physical qubits of the subsystem code and consider the code generated by all the stabilizers of the subsystem code that are fully supported within the subset. 
We remark that we cannot instead choose gauge operators supported within the subset as then some of the logical operators of the subsystem codes may not act as logical operators in the child code. 
We expect that it may be interesting to apply morphing to, e.g., the gauge color code~\cite{bombin2015,kubica2015} in order to define a HCT code family interpolating between the gauge color code and the recently introduced 3D subsystem toric code~\cite{kubica2021}.

Lastly, we expect that morphing could be fruitfully applied to triorthogonal codes~\cite{bravyi2012} and pin codes~\cite{vuillot2019a}, two families of codes that generalize color codes in different ways.  
We conjecture that morphing may allow us to improve the performance of MSD protocols based on triorthogonal codes (especially when mixing different input magic states), and to modify pin codes in much the same way as we did in the case of color codes. 

%%%%%%%%%%%%%%%%%%%%%%
%% Acknowledgements %%
%%%%%%%%%%%%%%%%%%%%%%
\section*{Acknowledgements}
M.V.\ thanks Dan Browne, Simon Burton and Chris Chubb for insightful discussions, and Raymond Laflamme for valuable comments on an earlier draft of this manuscript. 
A.K.\ acknowledges funding provided by the Simons Foundation through the ``It from Qubit'' Collaboration. This work was completed prior to A.K.\ joining the AWS Center for Quantum Computing.
Research at Perimeter Institute is supported in part by the Government of Canada through the Department of Innovation, Science and Economic Development Canada and by the Province of Ontario through the Ministry of Colleges and Universities.

%\clearpage
\onecolumngrid

\appendix

%%%%%%%%%%%%%%%%%%%%%%%%%%%%%%%%%%%%%%%%
%% Morphing d-dimensional color codes %%
%%%%%%%%%%%%%%%%%%%%%%%%%%%%%%%%%%%%%%%% 
\section{Morphing \texorpdfstring{$d$-dimensional}{d-dimensional} color codes \label{app:morph_ddim}}

In this appendix, we explain how to morph $d$-dimensional color codes. 
We begin by fixing the notation and reviewing the definitions of the toric code and the color code (\cref{subsec:prelim}).
Then, in \cref{subsec:ball_codes}, we examine the regions of the color code to which we apply our morphing procedure: colorable $d$-balls.
Finally, in \cref{subsec:hct_app} we introduce a family of HCT codes, which we construct by morphing the color code.

%%%%%%%%%%%%%%%%%%%
%% Preliminaries %%
%%%%%%%%%%%%%%%%%%%
\subsection{Preliminaries \label{subsec:prelim}}

A $d$-dimensional lattice $\mathcal L$ can be constructed by attaching $d$-dimensional cells to one another along their $(d-1)$-dimensional faces. 
We denote the $k$-cells of a lattice $\mathcal L$ by $\mathcal L_k$.
We say that an $m$-cell $\mu \in \mathcal L_m$ is contained in a $k$-cell $\kappa \in \mathcal L_k$ and write $\mu \subseteq \kappa$, if the vertices of $\mu$ are a subset of the vertices of $\kappa$ (where $m \leq k$).
And we define the $n$-star of a $k$-cell $\kappa \in \mathcal L_k$ to be $\St_n(\kappa) = \{ \nu \in \mathcal L_n \mid \kappa \subseteq \nu \}$. 

The lattices we consider may have a boundary, denoted by $\partial \mathcal L$, which comprises all $(d-1)$-cells of $\mathcal L$ that are contained in a single $d$-cell, along with all $k$-cells contained in these $(d-1)$-cells, for $0 \leq k \leq d-2$. 
We denote the interior of a lattice by $\mathcal L^\circ = \mathcal L \setminus \partial \mathcal L$ and we denote the internal $k$-cells by $\mathcal L_k^\circ = \mathcal L_k \cap \mathcal L^\circ$. 

We often consider lattices whose cells are simplices and whose vertices are $(d+1)$-colorable, i.e., one can introduce a function 
\begin{equation}
	\mathrm{col} : \mathcal L_0 \rightarrow [d+1],
	\label{eq:color-f}
\end{equation}
such that for any two vertices $u$ and $v$ sharing an edge, $\col{u} \neq \col{v}$, and we use the shorthand $[n] = \{ 1,2,\ldots, n \}$.
Such a lattice is called a $d$-colex~\cite{bombin2007}.
We abuse the notation and write $\col{\kappa} = \bigcup_{v \subseteq \kappa} \col{v}$ for some $k$-simplex $\kappa$, i.e., the color of $\kappa$ is the set of the colors of its vertices.
We denote the subset of $k$-cells in $\mathcal L$ of color $c$ by $\mathcal L_k^c = \{ \kappa \in \mathcal L_k \mid \col{\kappa} = c \}$.
Given two simplices $\mu \in \mathcal L_m$ and $\nu\in \mathcal L_n$, we define their join $\mu \ast \nu$ to be the smallest simplex in $\mathcal L$ that contains them both. 
If no such simplex exist, then we define their join to be empty.

Anticipating the definition of the color code, we now introduce generalized boundary operators.
Given a lattice $\mathcal L$, let $C_k (\mathcal L)$ be a vector space over $\mathbb F_2$ with a basis given by the $k$-cells of $\mathcal L$. 
Subsets of $k$-cells are consequently isomorphic to vectors in $C_k (\mathcal L)$.
For all $k \neq n$, we define the generalized boundary operator $\partial_{k,n}$ as a linear map specified for every basis element $\kappa \in \mathcal L_k$ as
\begin{equation}
	\partial_{k,n} \kappa = 
	\begin{cases}
		\sum_{\nu \subseteq \kappa} \nu, & \text{for } k > n, \\
		\sum_{\nu \supseteq \kappa} \nu, & \text{for } k < n.
	\end{cases}
	\label{eq:gen_bound}
\end{equation}
We note that the `standard' boundary operator $\partial_k = \partial_{k, k-1}$ is a special case of the generalized boundary operator. 

%%%%%%%%%%%%%%%%%%%%%%%%%%%%%%%%%
%% Toric codes and color codes %%
%%%%%%%%%%%%%%%%%%%%%%%%%%%%%%%%%
\subsubsection{Toric codes and color codes \label{subsec:toric_color_def}}

We now review the definition of the toric code~\cite{kitaev1997a,kitaev2003,dennis2002} in $d$-dimensions, where $d \geq 2$. 
Let $\mathcal L$ be a $d$-dimensional lattice with boundary $\partial \mathcal L$.
We divide $\partial \mathcal L$ into two subsets, the `rough' boundaries $\partial \mathcal L^{(R)}$ and the `smooth' boundaries $\partial \mathcal L^{(S)}$.
We place a qubit on each edge in $\mathcal L_1 \setminus \partial \mathcal L^{(R)}_1$.
The stabilizer of the code is
\begin{equation}
	\mathcal S = 
	\langle
	X (v), Z (f) \mid v \in \mathcal L_0 \setminus \partial \mathcal L^{(R)}_0, f \in \mathcal L_2 \setminus \partial \mathcal L^{(R)}_2
	\rangle,
	\label{eq:tc_stab}
\end{equation}
where $X(v) = \prod_{e \supseteq v} X_e$, $Z(f) = \prod_{e \subseteq f} Z_e$, $Z_e$ denotes a Pauli $Z$ operator acting on the qubit on edge $e$, and $X_e$ denotes a Pauli $X$ operator acting on the qubit on edge $e$. 
For a face $f$ adjacent to a rough boundary, the set of edges contained in $f$ may include edges with no associated qubits. 
In this case we define $Z_e$ to be the identity.

The other important family of codes for us are color codes~\cite{bombin2006,bombin2007,kubica2015a}.
Let $\mathcal L$ be a $(d\geq 2)$-colex, with boundary $\partial \mathcal L$. 
The color code defined on $\mathcal L$ has qubits on $d$-cells, $X$-type stabilizer generators associated with internal vertices, and $Z$-type stabilizer generators associated with internal $(d-2)$-simplices. 
That is, the stabilizer group of the code is 
\begin{equation}
	\mathcal S =
	\langle
		X (v), Z (\mu) \mid v \in \mathcal L_0^\circ, \mu \in \mathcal L_{d-2}^\circ
	\rangle,
	\label{eq:cc_stab}
\end{equation}
where $U (\kappa) = \prod_{\delta \in \St_d (\kappa)} U_\delta$ and $U_\delta$ means $U$ applied to the qubit on the $d$-simplex $\delta \in \mathcal L_d$. 
This notation clashes with the notation used for toric codes above, but it should always be clear from context what kind of code we are referring to.

%%%%%%%%%%%%%%%%
%% Ball codes %%
%%%%%%%%%%%%%%%%
\subsection{Ball codes \label{subsec:ball_codes}}

Let $\mathcal L$ be a $d$-colex and pick a vertex $v \in \mathcal L_1$. 
We define the colorable $d$-ball $\mathcal B^v$ centred at $v$ to be 
\begin{equation}
	\mathcal B^v = \bigcup_{k=0}^{d} \{ \kappa \in \mathcal L_k \mid \kappa \supseteq v \},
	\label{eq:ball_def}
\end{equation}
i.e., $\mathcal B^v$ is the union of all $k$-simplices in $\mathcal L$ that contain $v$. 
We define the boundary of $\mathcal B^v$ as follows,
\begin{equation}
	\partial \mathcal B^v = \bigcup_{k=0}^{d} \{ \kappa \in \mathcal L_k \setminus \mathcal B^v_k \mid \exists \delta \in \mathcal B^v_d : \delta \supseteq \kappa \}.
	\label{eq:ball_bound_def}
\end{equation}
Namely, $\partial \mathcal B^v$ consists of all the $k$-simplices of $\mathcal L$ that are contained in a $d$-simplex of $\mathcal B^v$, but are not themselves contained in $\mathcal B^v$.
We use $\mathcal B^v_k$ to denote the set of $k$-simplices contained in $\mathcal B^v$, with the obvious extension to $\partial \mathcal B^v$.
And we define the color of a colorable $d$-ball $\mathcal B^v$ to be $\col{v}$.
\Cref{fig:ball-ex} shows some example colorable $d$-balls in 2D and 3D.

\begin{figure}
	\centering
	\subfloat[]{
		\centering
		\includegraphics[width=0.15\linewidth]{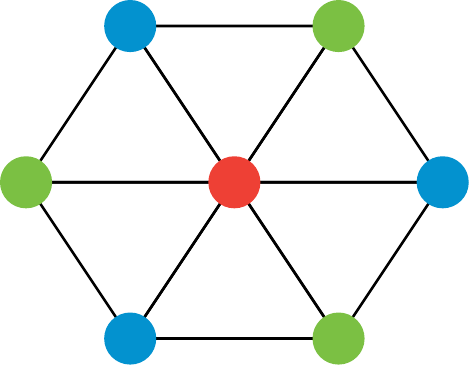}
		\label{sfig:2ball}
	}
	\subfloat[]{
		\centering
		\includegraphics[width=0.15\linewidth]{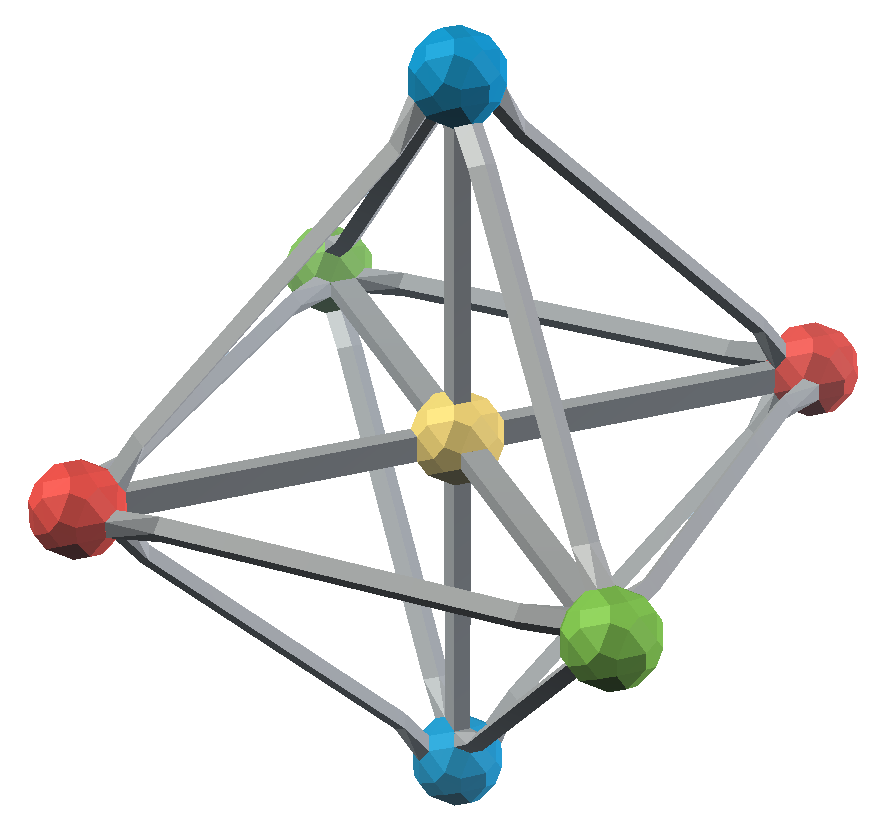}
		\label{sfig:3ball}
	}
	\caption{
		(a) A colorable $2$-ball and its boundary. 
		The colorable $2$-ball consists of the central $r$-vertex and all simplices containing it. 
		The boundary consists of the $b$- and $g$-vertices, alongside the edges connecting them.
		(b) A 3D colorable $3$-ball and its boundary.
		The colorable $3$-ball consists of the central yellow vertex and all simplices containing it.
		The boundary consists of the $r$-, $g$- and $b$-vertices, alongside the triangles and edges whose vertices are subsets of these vertices. 
	}
	\label{fig:ball-ex}
\end{figure}

When we morph a color code, we pick the subset of qubits $R$ (see \cref{sec:morphing}) to be the qubits contained in a colorable $d$-ball $\mathcal B^v$. 
The corresponding code has the stabilizer
\begin{equation}
	\mathcal S_{\mathcal B^v} = \langle 
		X(v), Z(\mu) \mid \mu \in \mathcal B^v_{d-2}
	\rangle.
	\label{eq:ball_stabs}
\end{equation}
This is exactly the stabilizer group of a color code defined on a $d$-colex $\mathcal B^v \cup \partial B^v$ with boundary $\partial B^v$ (recall \cref{eq:cc_stab}).
We call such a code a `ball code'.
We now come to our first lemma about ball codes.

\params*

\begin{proof}
	As the ball code defined on $\mathcal B_v$ is a color code, we immediately have $N = |\mathcal B^v_d|$.

	To find $K$, we first observe that one can define a $(d-1)$-dimensional color code on $\partial \mathcal B^v$, which we call the boundary code. 
	In the boundary code, qubits are on $(d-1)$-simplices, $X$-type stabilizers are associated with vertices, and $Z$-type stabilizers are associated with $(d-3)$-simplices. 
	There is an isomorphism between the $k$-simplices of $\mathcal B^v$ and the $(k-1)$-simplices of $\partial \mathcal B^v$, for $1 \leq k \leq d$. We have $\kappa \cong \tau$ for $\kappa \in \mathcal B^v_k$ and $\tau \in \partial \mathcal B^v_{k-1}$, if $\kappa = \tau \ast v$. 
	Furthermore, the $Z$-type stabilizers of the ball code and the boundary code have the same support, because $\St_d (\kappa) = \St_d (\tau)$ for $\kappa = \tau \ast v$.
	In particular, independent $Z$-type stabilizer generators of the ball code correspond to independent $Z$-type stabilizer generators of the boundary code.

	As the boundary code is a color code defined on a $(d-1)$-colex that is homeomorphic to a $(d-1)$-sphere, it encodes no logical qubits. 
	Let $\mathcal S^Z_{\mathcal B^v}$ and $\mathcal S^Z_{\partial \mathcal B^v}$ denote the $Z$-type stabilizer groups of the ball and boundary codes, respectively (and analogously for the $X$-type stabilizer groups). 
	We have
	\begin{equation}
		N = \rank \mathcal S^Z_{\partial \mathcal B^v} + \rank \mathcal S^X_{\partial \mathcal B^v}  
		= \rank \mathcal S^Z_{\mathcal B^v} + \rank \mathcal S^X_{\mathcal B^v} + K.
		\label{eq:k_Bv_1}
	\end{equation}
	Solving for $k$, we obtain
	\begin{equation}
		K = \rank \mathcal S^X_{\partial \mathcal B^v} - 1,
		\label{eq:k_Bv_2}
	\end{equation}
	as $\rank \mathcal S^Z_{\mathcal B^v} = \rank \mathcal S^Z_{\partial \mathcal B^v}$ and $\rank \mathcal S^X_{\mathcal B^v} = 1$.

	The $X$-type stabilizer generators of the boundary code are associated with the vertices of $\partial \mathcal B^v$, and every $(d-1)$-simplex of $\partial \mathcal B^v$ contains exactly $d$ vertices, each with a different color. 
	Consequently, we have $(d-1)$ independent relations between subsets of $X$-type stabilizer generators, namely
	\begin{equation}
		\prod_{\substack{u \in \mathcal B^v_0 \\ \col{u} = 2}} X(u) =
		\prod_{\substack{u \in \mathcal B^v_0 \\ \col{u} = 3}} X(u) =
		\ldots =
		\prod_{\substack{u \in \mathcal B^v_0 \\ \col{u} = d+1}} X(u),
		\label{eq:relations}
	\end{equation}
	where we assume without loss of generality that $\col{v} = 1$.
	This gives $\rank S^X_{\partial \mathcal B_v} = |\partial \mathcal B^v_0| - (d-1)$ and substituting into \cref{eq:k_Bv_2} gives the desired result.

	Finally, we come to the code distance $D$.
	As ball codes are CSS codes~\cite{calderbank1996,steane1996a}, we need only consider logical operators that are tensor products of $X$ and $I$ or tensor products of $Z$ and $I$. 
	We note that $X(e)$ is a logical $X$ operator of the ball code, where $e \in \mathcal B^v_1$.
	This operator is clearly not a stabilizer, but it commutes with the stabilizers of the ball code.
	To see this, consider a $Z$-type stabilizer generator $Z(\mu)$ where $\mu \in \mathcal B^v_{d-2}$.
	The overlap of $X(e)$ and $Z(\mu)$ (the qubits they both act non-trivially on) is given by $\St_d (\mu \ast e)$.
	The join of $\mu$ and $e$ is either empty, equal to $\mu$ (if $e \subseteq \mu$), or equal a $(d-1)$-simplex of $\mathcal B^v$. 
	In each case, $|\St_d (\mu \ast e)| = 0 \mod 2$; see the (Even Support) Lemma 4 in Ref.~\cite{kubica2015}. 

	The ball code has a single $X$-type stabilizer acting on all its physical qubits, so any weight-two $Z$ operator commutes with the stabilizers of the ball code.
	Consider the operator $Z(\lambda)$, where $\lambda \in \mathcal B^v_{d-1}$.
	$Z(\lambda)$ has weight two as each $(d-1)$-simplex of $\mathcal B^v$ is contained in exactly two $d$-simplices.
	For any such operator, we can always find a logical $X$ operator $X(e)$ such that $e \ast \lambda = \delta \in \mathcal B^v_D$, i.e., $Z(\lambda)$ and $X(e)$ anticommute. 
	Therefore, $Z(\lambda)$ is a logical $Z$ operator and consequently the ball code has a $Z$-distance of two. 
	The $X$-distance is also at least two as each qubit is in the support of two or more $Z$-type stabilizers. 
\end{proof}

We now specify a canonical generating set for the logical Pauli operators of a ball code defined on $\mathcal B^v$.
Following on from the proof of \cref{lem:ball_code_params}, we choose the following generating set for the logical $X$ operators, 
\begin{equation}
	\mathcal X = \bigcup_{c \in [d+1] \setminus \col{v}} \{ X(e) \mid e \in \mathcal B^v_1 \setminus e_c^* : \col{e} = c \},
	\label{eq:ball_x_logicals}
\end{equation}
where $e_c^*$ is some edge of color $c$ that we are free to choose.
Namely, our generating set consists of operators $X(e)$ for each edge $e$ in $\mathcal B^v$, except for one edge, $e_c^*$, of each color $c \neq \col{v}$. 
We choose each operator $X(e)$ in our generating set to act on a single logical qubit, and we define the color of a logical qubit to be the color of its associated edge. 
This implies that the operator $X(e_c^*)$ implements a logical $X$ on all logical qubits of color $c$.  
The set $\mathcal Z = \{ Z(\lambda) \mid \lambda \in \mathcal B^v_{d-1} \}$ generates the logical $Z$ operators of the ball code (see \cref{lem:ball_code_params}).
The action of these operators is fixed via commutation. 
Concretely, the logical operator $Z(\lambda)$ acts non-trivially on the logical qubits with associated logical operators $X(e)$ such that $|\St_d (\lambda) \cap \St_d (e)| = 1$ (or equivalently $\lambda \ast e = \delta \in \mathcal B^v_d$).

We now examine some examples of ball codes. 
The ball code defined on the colorable $2$-ball shown in \cref{sfig:2ball} is a \nkd{6}{4}{2} code.
And the ball code defined on the colorable $3$-ball shown in \cref{sfig:3ball} is an \nkd{8}{3}{2} code. 
This code is exactly the so-called `smallest interesting color code'~\cite{kubica2015,campbellblog}, and it is part of a family of color codes, which we call hyperoctahedron codes.

\begin{example}[Hyperoctahedron codes]
	A hyperoctahedron (also called a cross-polytope) is the $d$-dimensional generalization of the octahedron and can be constructed as follows.  
	In $d$-dimensional Euclidean space, place vertices at all the permutations of $(\pm 1, 0, 0, \ldots, 0)$.
	The $d$-hyperoctahedron is the convex hull of these vertices. 
	We remark that the $k$-cells of the $d$-hyperoctahedron are $k$-simplices, for $0 \leq k < d$.
	Let $\mathcal O$ denote the $d$-hyperoctahedron, which we transform into a colorable $d$-ball as follows.  
	First we create a vertex $v$ in the centre of $\mathcal O$.
	Next, for each $k$-simplex $\kappa$ of $\mathcal O$ we add the $k$-simplex $\kappa \ast v$, for $1 \leq k \leq d$.
	The result is a colorable $d$-ball, with $\mathcal B^v$ comprising the simplices that contain $v$ and $\partial \mathcal B^v$ consisting of the original simplices of $\mathcal O$. 
	The hyperoctahedron code is then simply the ball code defined on $\mathcal B_v$.
	Hyperoctahedron codes generalize the \nkd{4}{2}{2} code~\cite{knill2005} and the aforementioned \nkd{8}{3}{2} code (the dual of a $d$-hyperoctahedron is a $d$-hypercube). 
	The parameters of the $d$-hyperoctahedron code are $\nkd{2^d}{d}{2}$, as can be derived from \cref{eq:ball_code_params_ddim} and the fact that a $d$-hyperoctahedron has $2d$ vertices and $2^d$ $(d-1)$-cells.
	We note that for $d$-hyperoctahedron codes there is no freedom in the choice of canonical logical Pauli basis, as there is only one logical qubit of each color.
\end{example}

\begin{example}[3D ball codes]
	One can construct colorable 3-balls by taking the dual of Archimedean solids whose faces all have an even number of sides.
	The polyhedra satisfying these restrictions are the truncated octahedron, the truncated cuboctahedron, and the truncated icosidodecahedron. 
	The corresponding ball codes have parameters \nkd{24}{11}{2} (this is the code shown in \cref{fig:3d_ball_ex}), \nkd{48}{23}{2} and \nkd{120}{59}{2}, respectively.
	Each of these codes will have logical $CCZ$-type gates (see \cref{thm:ckz}) implemented by physical $T^{\pm 1}$ gates, and could therefore be used to distill magic states. 
	The associated MSD protocol has asymptotic average overhead $O(\log^{\gamma}(1/\epsilon))$ as the output error $\epsilon \rightarrow 0$, where  $\gamma = \log(n/k)/\log(d)$. 
	The ball codes described above have $\gamma$ values close to one, especially the \nkd{120}{59}{2} code where $\gamma = 1.02$. 
	This is significant as it was previously conjectured that $\gamma \geq 1$ for all MSD protocols~\cite{bravyi2012}, and the only known protocols with $\gamma < 1$ require either many ($\approx 2^{58}$) qubits~\cite{hastings2018} or large ($41$) local qudit dimension~\cite{krishna2019a}.
\end{example}

%%%%%%%%%%%%%%%%%%%%%%%%%%%%%%
%% Hybrid color-toric codes %%
%%%%%%%%%%%%%%%%%%%%%%%%%%%%%%
\subsection{Hybrid color-toric codes \label{subsec:hct_app}}

Using our morphing procedure, we construct a family of HCT codes. 
Namely, we start with a color code and morph a subset of its colorable $d$-balls to obtain an HCT code.
We note that we use the phrase `morphing a colorable $d$-ball' as a shorthand for morphing a color code with the subset $R$ of the physical qubits given by the qubits contained in a colorable $d$-ball. 
The HCT code family contains the color code and copies of the toric code, but also contains codes with regions of toric code and regions of color code. 
We now explain how the parameters of an HCT code change when we morph a colorable $d$-ball $\mathcal B_v$, which allows us to derive the parameters of any HCT code (by starting from the parent color code).
Suppose that the initial HCT code has parameters \nkd{N}{K}{D}.
Then after morphing the colorable $d$-ball $\mathcal B_v$, the resultant code has parameters
\begin{equation}
\begin{split}
	&N' = N - |\mathcal B^v_d| + |\mathcal B^v_1| - d, \\
	&K' = K, \\
	&D' \geq D - \max_{e \in \mathcal B^v_1} |\St_d(e)| + 1. 
	\label{eq:param_change_ball}
\end{split}
\end{equation}
We now briefly justify \cref{eq:param_change_ball}.
The value for $N'$ follows directly from \cref{lem:ball_code_params}.
As we discussed in \cref{sec:morphing}, the number of encoded qubits remains the same as we lose the same number of physical qubits and independent stabilizer generators.
Finally, the code distance decreases because some logical operators (acting on two or more physical qubits) of the ball code will be mapped to single-qubit operators by morphing. 
The highest-weight generator of the logical Pauli group of the ball code defined on $\mathcal B^v$ is the logical $X$ operator of highest weight (the $Z$ operators are weight two), i.e., $X(e^*)$ where $e^* = \argmax_{e \in \mathcal B^v_1} |\St_d(e)|$.
Importantly, for color codes $\max_{e \in \mathcal B^v_1} |\St_d(e)|$ is a constant that does not scale with the size of the code, so the code distance remains macroscopic.

In the remainder of this section, we introduce a geometric picture of morphing and we show that the HCT code family contains toric codes. 

%%%%%%%%%%%%%%%%%%%%%%%%%%%%%%%%%%%
%% Geometric picture of morphing %%
%%%%%%%%%%%%%%%%%%%%%%%%%%%%%%%%%%%
\subsubsection{Geometric picture of morphing \label{subsec:geom_pic}}

We begin by introducing a canonical geometric picture of morphing, which will enable us to see the connection to toric codes more clearly. 
Let $\mathcal B^v$ be a colorable $d$-ball of a $d$-colex $\mathcal L$, and assume that $\col{v} = d+1$.
When morphing $\mathcal B^v$, we represent our choice of logical Pauli basis pictorially by removing $\mathcal B_v$ from $\mathcal L$ and creating new edges linking vertices of the same color in $\partial \mathcal B^v$. 
Let $u \in \partial \mathcal B^v_0$ be the endpoint of the edge $e_c^*$ (recall \cref{eq:ball_x_logicals}) on the boundary of $\mathcal B^v$.
We add edges connecting $u$ to all the other vertices of color $c$ on the boundary of $\mathcal B^v$, and we place physical qubits on these new edges. 
We refer to these new qubits as edge qubits and we refer to the original color code qubits as simplex qubits.
We use $\widetilde{\mathcal L}$ to denote the HCT code lattice produced from $\mathcal L$ according to the above procedure.  
Also, we note that we must modify the generalized boundary operators to reflect the changes to the lattice structure; see \cref{app:local_lift} for a concrete example of this.

This geometric picture enables us to continue to associate the $X$-type stabilizer generators of an HCT code with vertices.
Concretely, the $X$-type stabilizer associated with a vertex $u \in \widetilde{\mathcal L}_0$ is
\begin{equation}
	X(u) = \prod_{\delta \supseteq u} X_{\delta} \prod_{\substack{e \supseteq u \\ \col{e} = \{c,c\}}} X_e,
	\label{eq:geom_pic_X_stab}
\end{equation} 
where $c = \col{u}$, $\delta \in \widetilde{\mathcal L}_d$, and $e \in \widetilde{\mathcal L}_1$ has an associated edge qubit. 
This last point is guaranteed by the condition that $\col{e} = \{c,c\}$, as, unlike the edges of $\mathcal L$, the edges created during morphing connect vertices of the same color. 
\Cref{fig:stab_morph} shows the intersection of a color code $X$-type stabilizer with a colorable $3$-ball. 

\begin{figure*}
	\centering
	\subfloat[]{
		\centering
		\includegraphics[width=0.2\linewidth]{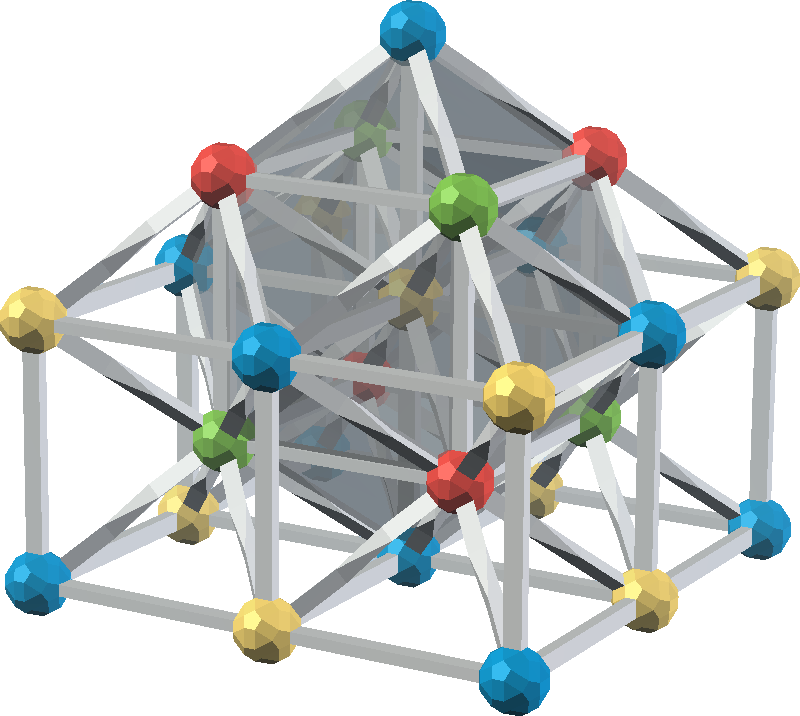}	
		\label{sfig:3d-morph-start}
	}
	\subfloat[]{
		\centering
		\includegraphics[width=0.2\linewidth]{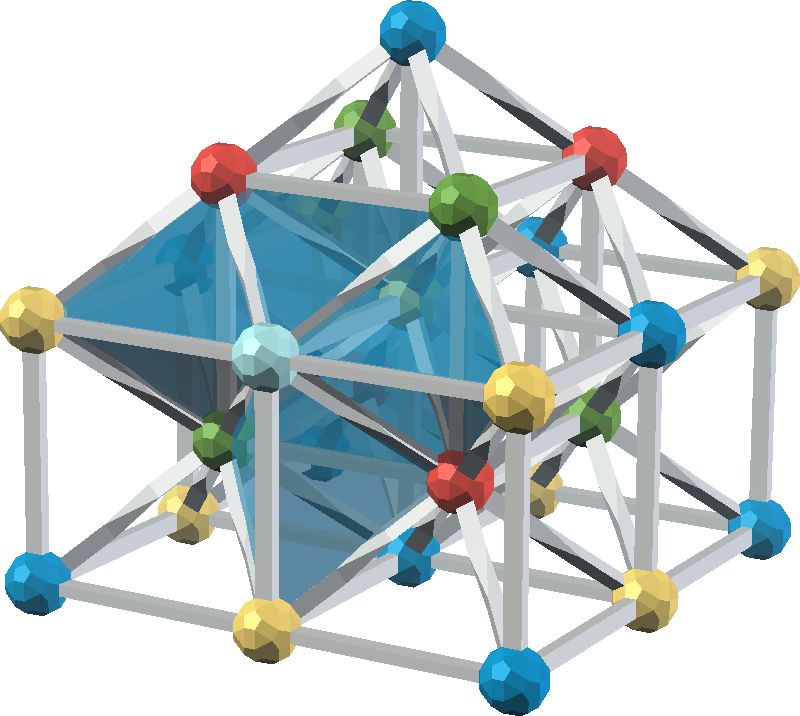}	
	}
	\subfloat[]{
		\centering
		\includegraphics[width=0.2\linewidth]{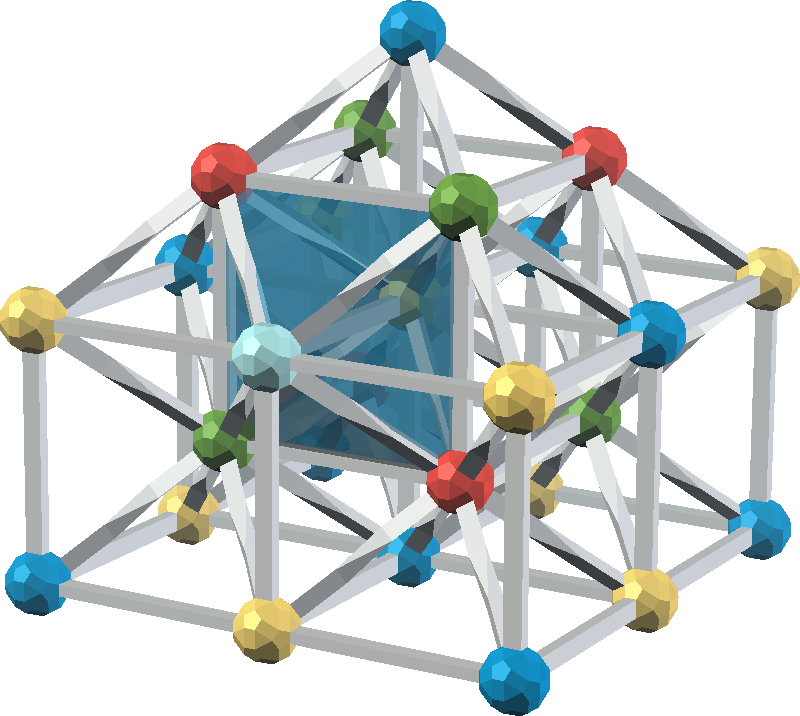}	
	}
	\caption{
		(a) We show part of a $3$-colex with a colorable $3$-ball highlighted in grey.
		(b) We highlight (shaded blue tetrahedra) the support of the $X$-type stabilizer associated with the light blue vertex. 
		(c) We show the intersection of the support of the stabilizer (b) with the colorable $3$-ball (a). 
		These qubits comprise the support of the ball code logical $X$ operator that is associated with the edge linking the light blue vertex and the central yellow vertex. 
	}
	\label{fig:stab_morph}
\end{figure*}

The $Z$-type stabilizers of a HCT code can also be understood using our geometric picture, though the definition is more obscure.
We continue to associate HCT code $Z$-type stabilizer generators with $(d-2)$-cells, as follows
\begin{equation}
	Z(\mu) = \prod_{\delta \supseteq \mu} Z_{\delta} 
	\prod_{\substack{e = (u,v) \\ u \in \partial_{d-1,0}\partial_{d-2,d-1}\mu \\ \col{e}=\{c,c\}}}
	Z_e,
	\label{eq:geom_pic_Z_stab}
\end{equation}
where $\delta \in \widetilde{\mathcal L}_D$, and $c = \col{u}$. 
To unpack \cref{eq:geom_pic_Z_stab}, we note that $\partial_{d-2,d-1}\mu$ comprises the $(d-1)$-simplices of $\widetilde{\mathcal L}$ that contain $\mu$ and $\partial_{d-1,0}\partial_{d-2,d-1}\mu$ consists of the vertices contained in an odd number of these $(d-1)$-simplices.

%%%%%%%%%%%%%%%%%%%%%%%%%%%%%%
%% Toric codes as HCT codes %%
%%%%%%%%%%%%%%%%%%%%%%%%%%%%%%
\subsubsection{Toric codes as HCT codes \label{subsec:toric_hct}}

In this section, we sketch an argument that the HCT code family contains copies of the toric code. 
Let $\mathcal L$ be a $d$-colex defined on a closed manifold. 
Suppose we morph every colorable $d$-ball of color $r=d+1$ in $\mathcal L$.
In the resultant HCT code no stabilizer generators will be associated with $r$-vertices and the only physical qubits will be edge qubits, as all simplex qubits will have been removed during morphing. 
The $X$-type stabilizers of the HCT code can be partitioned into disjoint subgroups
$\langle
	X(v) \mid v \in \widetilde{\mathcal L}_0 : \col{v} = c
\rangle$, 
for $c \in [d]$.
This is because each generator $X(v)$ now acts exclusively on edge qubits of color $\{c,c\}$.

As there are no remaining $r$-vertices, given a $(d-2)$-cell $\mu \in \widetilde{\mathcal L}_{d-2}$, the color of any $(d-1)$-cell $\lambda \in \widetilde{\mathcal L}_{d-1}$ such that $\lambda \supseteq \mu$ is fixed. 
Therefore, the $Z$-type stabilizers of the HCT code can also be partitioned into disjoint subgroups indexed by $c \in [d]$, where the operators in each subgroup only act on qubits of color $\{c,c\}$. 
To summarize, for $c \in [d]$, the stabilizer group of the HCT code can be partitioned into $d$ disjoint subgroups of the form 
\begin{equation}
	\langle X(v), Z(\mu) \mid v \in \widetilde{\mathcal L}_0, \mu \in \widetilde{\mathcal L}_{d-2} : \col{v} = c, \col{\mu} = [d] \setminus \{ c \} \rangle.
	\label{eq:stab_partition}
\end{equation} 

When restricted to the $cc$-edges of $\widetilde{\mathcal L}$, the $X$-type stabilizer generators of the $c$-subgroup have the form $X(v) = \prod_{e \supseteq v} X_e$, i.e., they are toric code $X$-type stabilizers (recall \cref{eq:tc_stab}).
Now consider some $Z$-type stabilizer $Z(\mu)$ in the $c$-subgroup. 
Recalling \cref{eq:geom_pic_Z_stab}, we first write $\partial_{d-2,d-1} \mu = \{ \lambda_1, \lambda_2, \ldots, \lambda_m \}$.
In $\mathcal L$, each $\lambda_i$ must have been contained in the boundary of two colorable $d$-balls of color $d+1$, otherwise the $\lambda_i$ would have been removed during morphing. 
Therefore, $\lambda_i = \mu \ast u_i$, where $\col{u_i}=c$. 
Suppose that we permute our labelling of the $\lambda_i$ such that, in $\mathcal L$, $\lambda_i$ and $\lambda_{i+1}$ (evaluated modulo $m+1$) were on the boundary of the same colorable $d$-ball $\mathcal B^v$, where $\col{v}=r$.  
Then either $(u_i, u_{i+1})$ is in the support of $Z (\mu)$ or $(u_i, w)$ and $(w, u_{i+1})$ are in the support of $Z(\mu)$ for some $w$ that was also contained in $\partial \mathcal B^v$.
Consequently, the support of $Z(\mu)$ is a cycle of $cc$-edges, i.e., it has the form of a toric code $Z$-type stabilizer (recall \cref{eq:tc_stab}).

In 2D we can see that the HCT code family contains toric codes by inspection, e.g., by applying the transformation shown in \cref{fig:morph} to all $r$ colorable $2$-balls. 
In higher dimensions generic examples become more complex, excepting a family of HCT codes derived from color codes with special structure, which we examine in \cref{app:coxeter}. 

%%%%%%%%%%%%%%%%%%%%%%%%%%%%%%%%%%%%%%%%%%%%%%%%%%%
%% Non-Clifford gates in d-dimensional HCT codes %%
%%%%%%%%%%%%%%%%%%%%%%%%%%%%%%%%%%%%%%%%%%%%%%%%%%%
\section{Non-Clifford gates in \texorpdfstring{$d$-dimensional}{d-dimensional} HCT codes \label{app:non_cliff_ddim}}

Color codes have a multitude of transversal logical gates implemented $R_d$ gates, including non-Clifford gates in $d \geq 3$ dimensions~\cite{bombin2006,bombin2007a,bombin2015,kubica2015,kubica2015a,watson2015,bombin2018a}.
To describe the implementation of these gates, we note that the qubits of a color code are bipartite in the following sense. 
We can divide the $d$-simplices of a $d$-colex into two disjoint sets, such that any two $d$-simplices that share a $(d-1)$-simplex are in different sets.
We denote these sets by $T$ and $T^c$, and for a single-qubit unitary $U$ we write $U(T)$ to denote the application of $U$ to all of the qubits in $T$.
Every $d$-dimensional color code has a transversal logical gate implemented by $\widetilde R_d = R_d(T)R_d^\dagger(T^c)$.
The logical action of this gate depends on the topology of the colex.
For example, in tetrahedral color codes $\widetilde R_3$ implements a logical $R_3$ gate~\cite{bombin2007a,bombin2015,bombin2018a}. 
And in the family of hypercubic color codes (see \cref{fig:hypercubic}), $\widetilde R_d$ implements the logical $d$-qubit $MCZ$ gate~\cite{kubica2015}, which we denote by $C_d$. 

HCT codes inherit the logical gates of their parent color codes. 
To find the implementation of a non-Clifford gate in an HCT code, we investigate how morphing changes the implementation of the corresponding color code logical gate. 
This entails understanding the transversal logical gates of ball codes, which is the topic of the next section.

%%%%%%%%%%%%%%%%%%%%%%%%%%%%%
%% MCZ gates in ball codes %%
%%%%%%%%%%%%%%%%%%%%%%%%%%%%%
\subsection{\texorpdfstring{$MCZ$}{MCZ} gates in ball codes \label{subsec:mcz_ball}}

As ball codes are color codes, they have transversal logical gates that are implemented by $R_d$ gates.
Any unitary can be expanded in the Pauli basis, so in principle we can compute the logical action of any transversal gate in a ball code by expanding the transversal gate in the Pauli basis and then grouping terms by logical Pauli.
However, this is cumbersome to do in practice so instead we characterize the transversal logical gates of the ball code family.

First we introduce some notation. 
Let $\kappa \in \mathcal B_k^v$ be a $k$-simplex of a colorable $d$-ball $\mathcal B_v$, where $0 \leq k \leq D$.
We write
\begin{equation}
	\widetilde U(\kappa) = \prod_{\delta \in \St_d (\kappa) \cap T} U(\delta) \prod_{\delta \in \St_d (\kappa) \cap T^c} U^\dagger (\delta).
	\label{eq:tilde_U_kcell}
\end{equation}
The group commutator of two unitary operators $A$ and $B$ is defined to be $K[A, B] = A B A^\dagger B^\dagger$.
We note the identities $K[R_k, X] = R_{k-1}$ and $K[C_k, X] = C_{k-1}$.

\begin{theorem}
	Let $\mathcal B^v$ be a colorable $d$-ball with an associated ball code, where $d \geq 2$, and let $\kappa \in \mathcal B^v_{d-k}$ be a $(d-k)$-simplex of $\mathcal B_v$, where $k \in [d]$.
	The operator $\widetilde R_k (\kappa)$ implements a logical $C_k$ gate on all $k$-tuples of logical qubits whose corresponding edges $e_j \in \mathcal B^v_1$ have the following property:
	\begin{equation}
		e_1 \ast \ldots \ast e_k \ast \kappa = \delta \in \mathcal B^v_d, 
		\label{eq:join_cond}
	\end{equation}
	i.e., the join of the edges and $\kappa$ is a $d$-simplex of $\mathcal B^v$.
	\label{thm:ckz}
\end{theorem}

\begin{proof}
	We prove the theorem by induction on $k$. 
	
	\emph{Base case (k=1):} 
	Suppose we apply the operator $\widetilde R_1 (\lambda)$ for some $\lambda \in \mathcal B^v_{d-1}$.
	Up to an unimportant global phase, $\widetilde R_1 (\lambda)$ is a logical $Z$ operator of the ball code.
	Consider a single-qubit logical $X$ operator $X(e)$, where $e \in \mathcal B^v_1$.
	$\widetilde R_1 (\lambda)$ and $X(e)$ commute if $\St_d (\lambda) \cap \St_d (e) = \emptyset$ or if $e \subseteq \lambda$.
	In the second case we have that $\St_d (\lambda) \subseteq \St_d (e)$ and as $|\St_d (\lambda)| = 2$ the operators commute.
	Crucially, in both cases $(e \ast \lambda)$ is not a $d$-simplex of $\mathcal B^v$. 
	If $e \nsubseteq \lambda$ and their intersection is non-empty, we must have $|\St_d (\lambda) \cap \St_d (e)| = 1$, i.e., both $\lambda$ and $e$ are contained in the same $d$-simplex, or equivalently $(e \ast \lambda) = \delta \in \mathcal B^v_d$.

	\emph{Inductive step:} 
	Suppose we apply the operator $\widetilde R_{k+1} (\mu)$, where $\mu \in \mathcal B^v_{d-k-1}$.
	We first show that $\widetilde R_{k+1} (\mu)$ preserves the stabilizer. 
	We only need to check the single $X$-type stabilizer $X(v)$, as $\widetilde R_{k+1} (\mu)$ trivially commutes with the $Z$-type stabilizers. 
	Consider the group commutator 
	\begin{equation}
		K [\widetilde R_{k+1}(\mu), X(v)] = \widetilde R_{k+1} (\mu) X(v) \widetilde R_{k+1}^\dagger (\mu) X(v) 
		= \widetilde R_k (\mu).
	\end{equation}
	We can partition the support of $\widetilde R_k$ as follows
	\begin{equation}
		\St_d (\mu) = \bigsqcup_{\kappa \in \St_{d-k}^{c} (\mu)} \St_d (\kappa),
		\label{eq:mu_split}
	\end{equation}
	where $\St_{d-k}^c (\mu) = \{ \kappa \in \mathcal B^v_{d-k} \mid \mu \subseteq \kappa, \col{\kappa} = c\}$, and $c$ is a any subset of $[d+1]$ of cardinality $d-k+1$ with $\col{\mu} \subseteq c$.
	Consequently, we can rewrite \cref{eq:group_comm_k} as follows
	\begin{equation}
		K[\widetilde R_{k+1} (\mu), X(e)] = \prod_{\kappa \in \St_{d-k}^{c} (\mu)} R_k (\kappa).
		\label{eq:group_comm_decomp}
	\end{equation} 

	By the inductive assumption, $R_k (\kappa)$ implements a logical $C_k$ gate on all $k$-tuples of logical qubits whose corresponding edges $e_j \in \mathcal B_1^v$ are such that $(e_1 \ast \ldots \ast e_k \ast \kappa) = \delta \in \mathcal B^v_d$.
	Crucially, $\col{e_j} \nsubseteq c$ for all $j \in [k]$. 
	Without loss of generality, suppose that $\col{\mu} = [d-k]$ and $\col{v} = 1$. 
	Then, one possible choice for $c$ is $[d-k+1]$. 
	In this case, $\prod_{\kappa \in \St_{d-k}^{c} (\mu)} R_k (\kappa)$ acts on logical qubits with colors in the set $\{ (1, d-k+2), \ldots, (1, d+1) \}$.
	If we instead choose $c$ to be a different subset of $[d+1]$, then $\prod_{\kappa \in \St_{d-k}^{c} (\mu)} R_k (\kappa)$ would act on logical qubits with a different set of colors. 
	There are $k+1$ possible choices for $c$, and there is no common color of logical qubit in the $k+1$ corresponding sets of logical qubits. 	Therefore, the only way \cref{eq:group_comm_decomp} can hold for all choices of $c$ is if $\prod_{\kappa \in \St_{d-k}^{c} (\mu)} R_k (\kappa) = 1$.

	Now we consider the group commutator of $\widetilde R_{k+1} (\mu)$ with $X(e)$, where $e \in \mathcal B^v_1$.
	If $\St_d (\mu) \cap \St_d (e) = \emptyset$, then $K[\widetilde R_{k+1} (\mu), X(e)] = 1$, as the operators have disjoint support.
	Now suppose that $\St_d (\mu) \cap \St_d (e) \neq \emptyset$.
	If $\St_d (\mu) \subseteq \St_d (e)$, then $e \subseteq \mu$ and 
	\begin{equation}
		K[\widetilde R_{k+1} (\mu), X(e)] = \widetilde R_k (\mu) = 1,
		\label{eq:group_comm_k}
	\end{equation}
	by the same argument as the previous paragraph.
	The final case to consider is when $\St_d (\mu) \cap \St_d (e) \neq \emptyset$ and $\St_d (\mu) \nsubseteq \St_d (e)$.
	Observe that $\St_d (\mu) \cap \St_d (e) = \St_d (e \ast \mu)$, as every $d$-simplex that contains $\mu$ and $e$ must also contain their join. 
	Therefore
	\begin{equation}
		K[\widetilde R_{k+1} (\mu), X(e)] = \widetilde R_k (e \ast \mu),
		\label{eq:group_comm_ckz}
	\end{equation}
	where $(e \ast \mu) = \kappa \in \mathcal B^v_{d-k}$. 
	By the inductive assumption, $\widetilde R_k (\kappa)$ implements a logical $C_k$ gate on all $k$-tuples of logical qubits whose corresponding edges $e_j \in \mathcal B^v_1$ are such that $(e_1 \ast \ldots \ast e_k \ast \kappa) = \delta \in \mathcal B^v_d$.
	This completes the proof, as the group commutator of $C_{k+1}$ and a single-qubit $X$ is $C_k$. 
\end{proof}

\begin{figure}
	\centering
	\subfloat[]{
		\centering
		\includegraphics[width=0.15\linewidth]{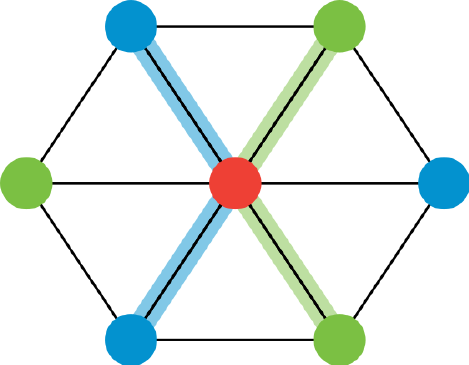}
		\label{sfig:642_canonical_basis}
	}
	\subfloat[]{
		\centering
		\includegraphics[width=0.15\linewidth]{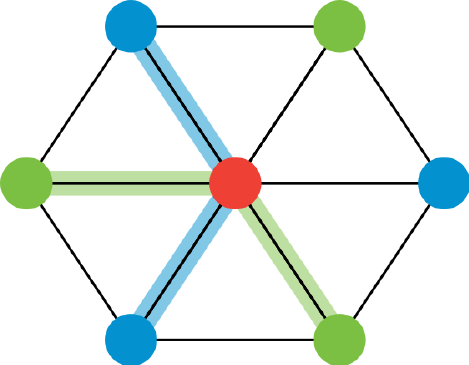}
	}
	\subfloat[]{
		\centering
		\includegraphics[]{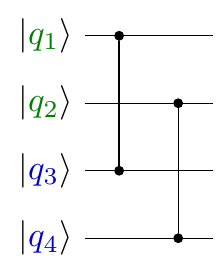}
	}
	\subfloat[]{
		\centering
		\includegraphics[]{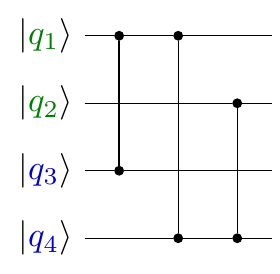}
	}
	\caption{
		Logical gates implemented by $\widetilde R_2 (v)$ in a \nkd{6}{4}{2} code.
		(a) One choice of logical Pauli basis, where for each highlighted edge $e$ we have a logical operator $X(e)$ with the labelling (clockwise from the top left) $\textcolor{blue!80!black}{\overline X_3}$, $\textcolor{green!50!black}{\overline X_1}$, $\textcolor{green!50!black}{\overline X_2}$, and $\textcolor{blue!80!black}{\overline X_4}$. 
		Here two subsets of logical operators have corresponding edges whose join is a triangle. 
		(b) A different choice of logical Pauli basis, where clockwise from the top left we have $\textcolor{blue!80!black}{\overline X_3}$, $\textcolor{green!50!black}{\overline X_2}$, $\textcolor{blue!80!black}{\overline X_4}$, and $\textcolor{green!50!black}{\overline X_1}$.
		Here three subsets of logical operators have corresponding edges whose join is a triangle. 
		(c) and (d) show the logical operator implemented by $\widetilde R_2 (v)$ for (a) and (b), respectively.
	}
	\label{fig:mcz_ex}
\end{figure}

We remark that if we choose $k = d$ in \cref{thm:ckz}, we obtain a logical circuit consisting of $d$-qubit $MCZ$ gates, implemented by $\widetilde R_d (v)$, i.e., the gate acts on all the physical qubits of the code. 
\Cref{fig:mcz_ex} shows some example logical circuits implemented by $\widetilde{R}_2$ for the \nkd{6}{4}{2} code. 
A particularly simple application of \cref{thm:ckz} is to hyperoctahedron codes.
For the $d$-dimensional hyperoctahedron code, it can be easily verified that $\widetilde R_d (v)$ implements a $C_d$ gate on the $d$ logical qubits of the code.
Moreover for $k \in [d]$ and $\kappa \in \mathcal B^v_{d-k}$, $\widetilde R_{k} (\kappa)$ implements a $C_k$ gate on the $k$ logical qubits of color $c \nsubseteq \col{\kappa}$.

\Cref{thm:ckz} allows us to characterize the implementation of the logical HCT code gates that are inherited from the parent color code. 
If the parent color code has a logical gate implemented by $\widetilde R_d$, then HCT codes derived from this code will have the same logical gate implemented by $R_d$, $R_d^\dagger$, and $C_d$ gates. 
For each morphed colorable $d$-ball, \cref{thm:ckz} tells us the specific $C_d$ circuit to apply. 
More concretely, for a ball code with parameters $\nkd{N}{K}{D}$, the $C_d$ circuit will have depth at most $\binom{N}{d}$. 
As the size of any ball-like region is bounded, the implementation of the logical HCT code gate will be constant-depth and hence fault-tolerant.
In \crefrange{app:qrm}{app:stell}, we consider families of HCT codes where the $C_d$ circuits for each ball-like region have simple structure. 

%%%%%%%%%%%%%%%%%%%%%%%%%%%%%%%%%%%%%%%%%%%
%% Further details on morphed QRM codes %%
%%%%%%%%%%%%%%%%%%%%%%%%%%%%%%%%%%%%%%%%%%%
\section{Further details on morphed QRM codes \label{app:qrm}}

We begin this Appendix by briefly reviewing the color code construction of distance-three QRM codes~\cite{kubica2015}.
We start with a $d$-simplex and place another $d$-simplex inside of it. 
We color the vertices of each simplex with colors $c \in [d+1]$ such that each vertex of the same simplex has a different color.
For each vertex $v$ of the internal simplex, we connect $v$ to all vertices of the external simplex whose color is different from $\col{v}$.
The $d$-dimensional color code defined on this lattice is the $d$-dimensional QRM code of distance three. 
\Cref{fig:rm23} shows the lattices for the 2D and 3D codes. 

Now, suppose we morph a colorable $d$-ball of $\mathrm{QRM}(d)$. 
We observe that all ball codes in $\mathrm{QRM}(d)$ are $d$-hyperoctahedron codes. 
This is because a given internal vertex $v$ has $d$ neighbors in the internal simplex and $d$ neighbors in the external simplex, where there are exactly two neighbors of each color, and all neighbors are connected to each other except for those of the same color. 
Therefore the $(d-1)$-dimensional surface of any internal colorable $d$-ball in $\mathrm{QRM} (d)$ is exactly a $d$-hyperoctahedron.
The parameters of the morphed $\mathrm{QRM} (d)$ code are $\nkd{2^d + d - 1}{1}{2}$ and it has a fault-tolerant logical $R_{d}$ gate.
The logical $R_d$ gate of the morphed code has the same implementation as in $\mathrm{QRM} (d)$ for the remaining simplex qubits, along with a $C_d$ gate on the edge qubits. 

\begin{figure}
	\centering
	\subfloat[]{
		\centering
		\includegraphics[width=0.2\linewidth]{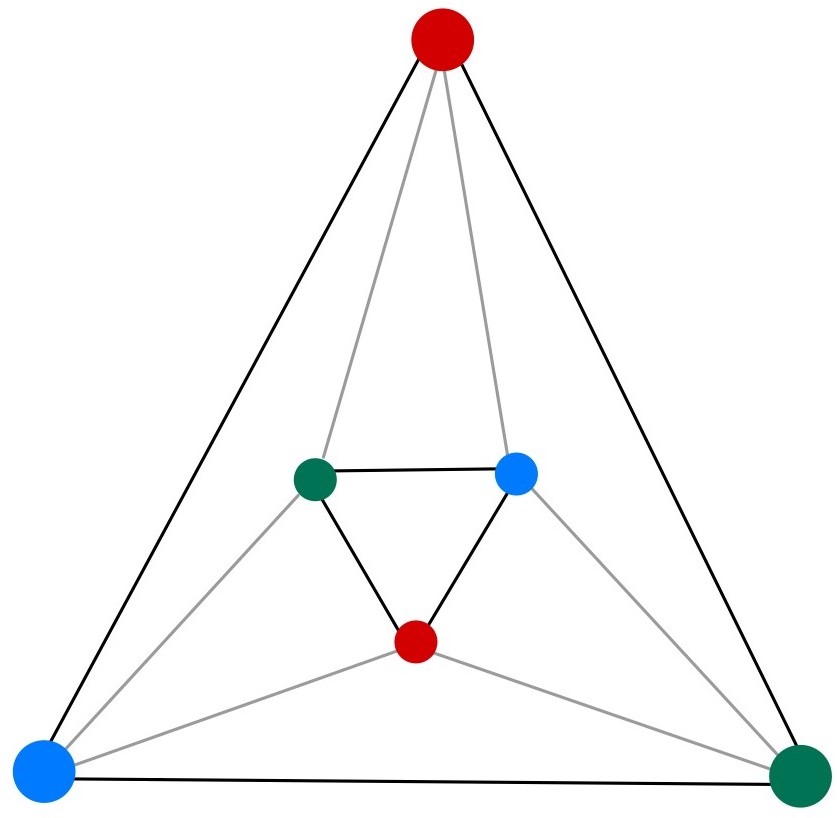}
	}
	\quad
	\subfloat[]{
		\centering
		\includegraphics[width=0.2\linewidth]{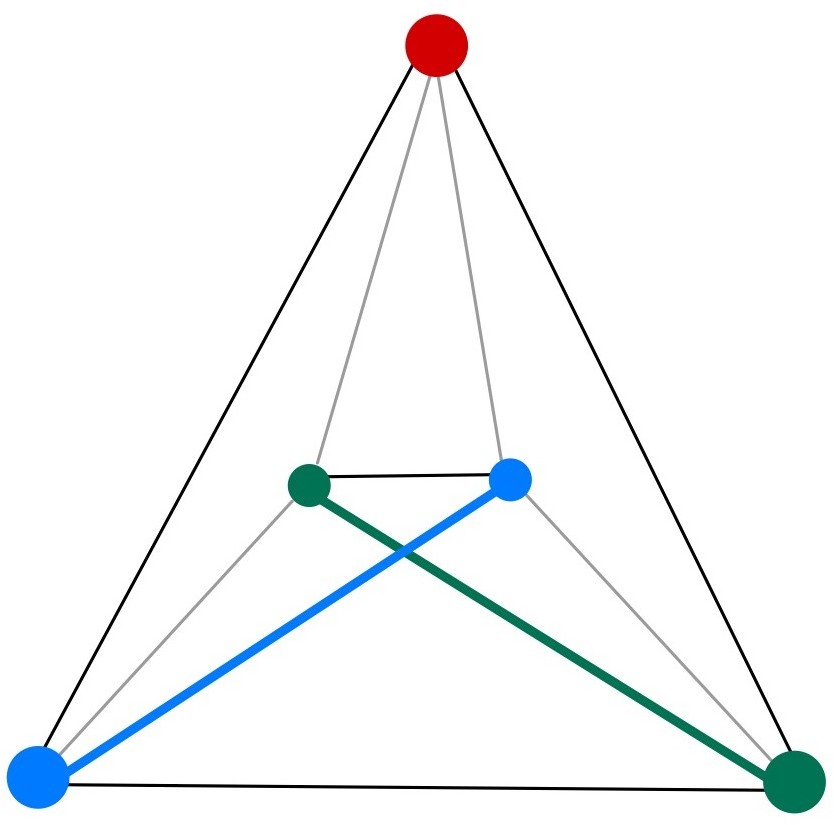}
	}
	\quad
	\subfloat[]{
		\centering
		\includegraphics[width=0.2\linewidth]{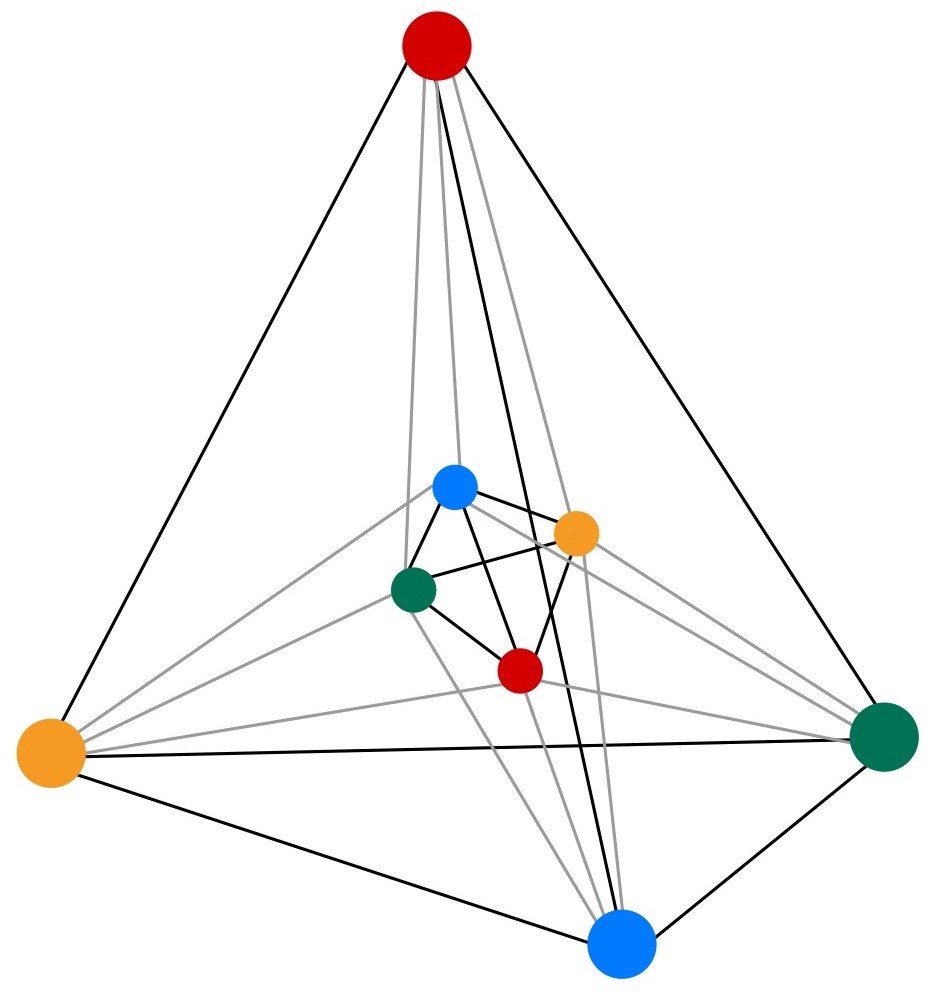}
	}
	\quad
	\subfloat[]{
		\centering
		\includegraphics[width=0.2\linewidth]{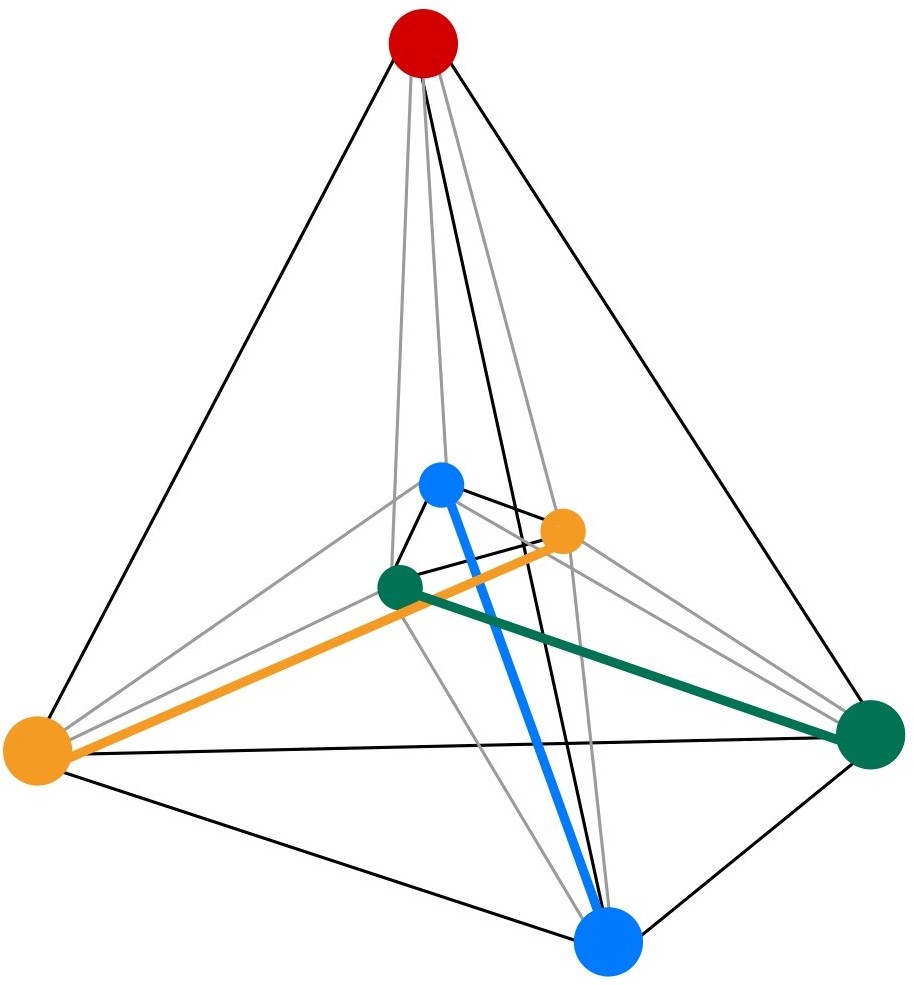}
	}
	\caption{
		QRM codes and their morphed counterparts. 
		(a) $\mathrm{QRM}(2)$, more commonly known as the Steane code. 
		Qubits are on triangles and stabilizer generators are associated with internal vertices.
		(b) The \nkd{5}{1}{2} code produced by morphing the colorable $2$-ball centred at the internal $r$-vertex. 
		Four qubits have been replaced by two (the blue and green edges).
		(c) $\mathrm{QRM}(3)$, more commonly known as 15-qubit Reed-Muller code. 
		Qubits are on tetrahedra and $X$-type ($Z$-type) stabilizer generators are associated with internal vertices (edges). 
		(d) The \nkd{10}{1}{2} code produced by morphing the colorable $3$-ball at the internal $r$-vertex. 
		Eight qubits have been replaced by three (the blue, green and yellow edges).
	}
	\label{fig:rm23}
\end{figure}

For MSD protocols of the type discussed in \cref{subsec:msd_protocol}, the leading order term in the output error is $ap^D$, where $D$ is the distance of the code and $a$ is the number of logical $Z$ operators of weight $D$. 
For $\mathrm{QRM}(d)$, $a=(1-2^{d+1})(1-2^{d})/3$ 
whereas in the corresponding morphed QRM code $a=d$ (as we prove shortly), so there is an exponential separation between the leading order prefactors as the spatial dimension increases.

\begin{lemma}
	The morphed $QRM(d)$ code has $d$ logical $Z$ operators of weight two. 
	\label{lem:logZnum}
\end{lemma}
\begin{proof}
	Any weight two logical $Z$ operator of the morphed $\mathrm{QRM}(d)$ code must act on one or more edge qubits as the distance of $\mathrm{QRM}(d)$ is three.
	Consider an edge of color $c$ and suppose that we apply a $Z$ operator to the qubit on this edge.
	To form a weight two logical $Z$ operator, the only possible choice for the second qubit is the $d$-simplex spanned by the internal vertex of color $c$ and the vertices on the boundary whose colors are not equal to $c$. 
	To see this, observe that the Pauli $Z$ acting on the edge qubit anticommutes with the $X$-type stabilizer generator associated with the internal endpoint of the edge.
	The only other qubit that is contained in the support this $X$-type stabilizer generator and no other $X$-type stabilizer generators is the one described above, as all other $d$-simplices contain internal vertices of color not equal to $c$.
	Therefore any other choice will lead to more unsatisfied stabilizers.
	And any weight two $Z$ operator acting exclusively on edge qubits will not satisfy two stabilizers (associated with the internal endpoints of the edges).
	In the morphed $\mathrm{QRM} (d)$ code there are $d$ edges with associated qubits, so there are $d$ corresponding logical $Z$ operators of weight two.
\end{proof}

%%%%%%%%%%%%%%%%%%%%%%%%%%%%%%%%%%%%%%%%%%%%%%%%%%%%%%%%%%%%%%%%%%%%
%% Constructing color codes and toric codes from Coxeter diagrams %%
%%%%%%%%%%%%%%%%%%%%%%%%%%%%%%%%%%%%%%%%%%%%%%%%%%%%%%%%%%%%%%%%%%%%
\section{Constructing color codes and toric codes from Coxeter diagrams}
\label{app:coxeter}

In this appendix we describe a construction of color codes and toric codes from Coxeter diagrams~\cite{coxeter1973}.
Such an approach has been employed before in~\cite{vuillot2019a,vasmer2019a} and is a generalization of the Schl\"{a}fli symbol construction used in~\cite{breuckmann2016a,breuckmann2017,breuckmann2018,jochym-oconnor2021}. 
Coxeter diagrams are used to represent Coxeter groups, which describe the symmetries of polytopes and tessellations in terms of reflections. 
A tessellation is a gapless covering of a manifold by uniform polytopes such that each adjacent pair of polytopes share a facet (a $(d-1)$-cell).
We say that a tessellation is uniform if all of its vertices are identical, i.e., there is the same combination of polytopes at each vertex. 
And we say that a tessellation is regular it is uniform and all the polytopes are identical.
A Coxeter group is defined by the group presentation
\begin{equation}
	\langle 
		r_1, r_2, \ldots, r_m : (r_i r_j)^{p_{i,j}} = 1
	\rangle,
	\label{eq:cox_gp_defn}
\end{equation}
where $p_{i,i} = 1$ (as the $r_i$ are reflections), $p_{i,j} \geq 2$ for $i \neq j$, and $p_{i,j} = \infty$ implies that the relation is ignored. 

Each Coxeter group has an associated graph called a Coxeter diagram, which describes the spatial relations between a collection of reflecting hyperplanes.
There is a reflecting hyperplane (the $r_i$ in \cref{eq:cox_gp_defn}) for each vertex of the diagram and the edges represent the dihedral angles between different reflecting hyperplanes (corresponding to the $p_{ij}$ in \cref{eq:cox_gp_defn}).
Unmarked edges correspond to dihedral angles of $\pi / 3$ whereas edges marked by $p > 3$ correspond to dihedral angles of $\pi / p$. 
Vertices that are not connected by edges have dihedral angles of $\pi / 2$.
\Cref{fig:coxeter_demo} gives an example using the Coxeter diagram \dynkin[Coxeter,extended,affine mark=*]{C}{2}.

\begin{figure}
	\centering
		\includegraphics[width=0.3\linewidth]{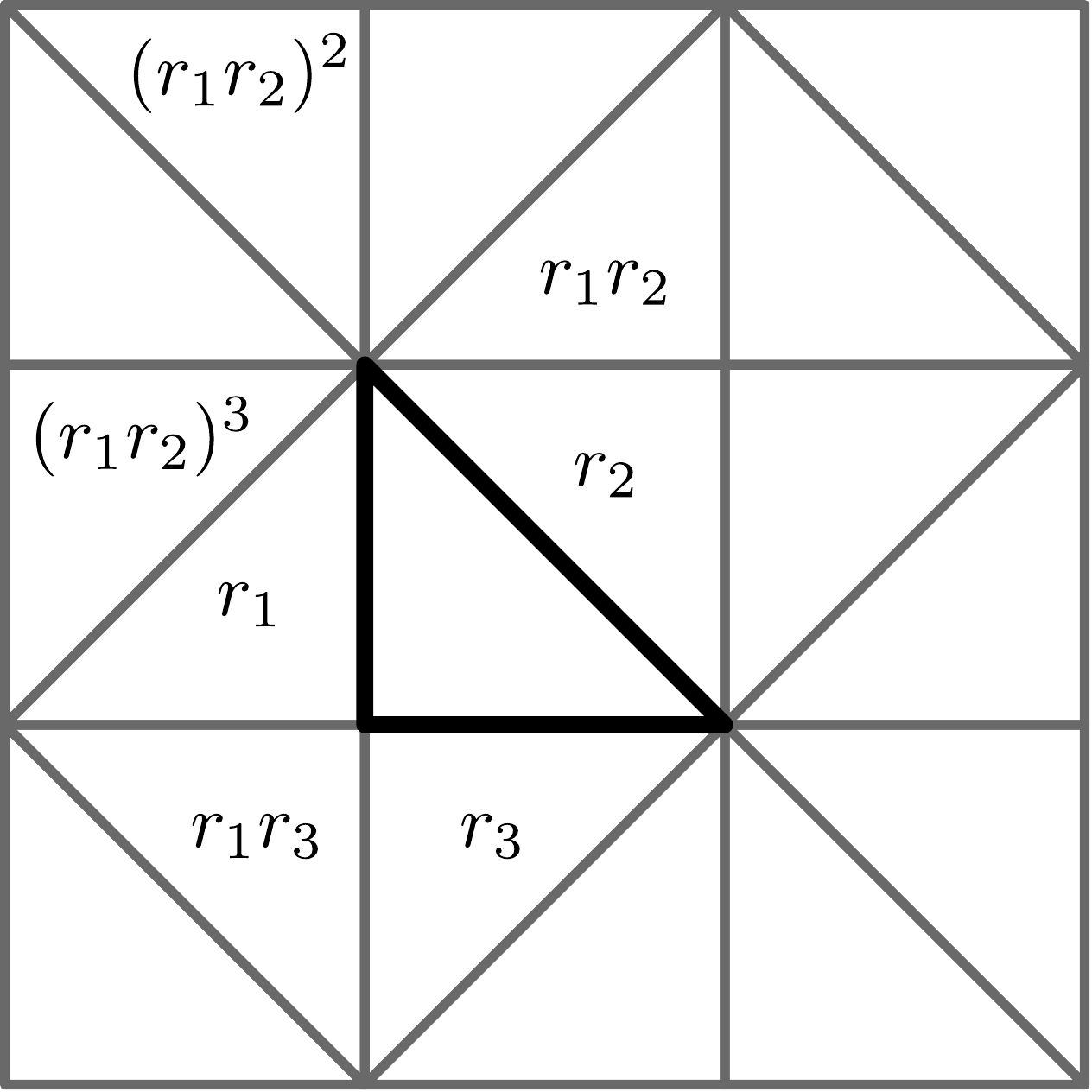}
	\caption{
		Illustrating the Coxeter kaleidoscope.
		The vertices of the Coxeter diagram \dynkin[Coxeter,extended,affine mark=*]{C}{2} correspond to three reflecting hyperplanes $r_1$, $r_2$, and $r_3$ (labelled left-to-right).
		In the figure the reflecting hyperplanes are the sides of the black triangle (the fundamental region). 
		As specified by the diagram, $r_1$ and $r_3$ have dihedral angle $\pi / 2$ whereas $r_1$ and $r_2$ have dihedral angle $\pi / 4$ (as do $r_2$ and $r_3$). 
		We show a subset of the (infinite) reflections of the fundamental region and we label some of these reflections in terms of $r_1$, $r_2$, and $r_3$. 
		From the figure we see that $(r_1 r_3)^2 = 1$ and $(r_1 r_2)^4 = 1$, matching the Coxeter diagram.
	}
	\label{fig:coxeter_demo}
\end{figure}

Given a Coxeter diagram, one can build uniform tilings using Wythoff's construction~\cite{coxeter1973}.
Consider the kaleidoscope corresponding to a given Coxeter diagram (see \cref{fig:coxeter_demo} for an example).
We mark a subset of the vertices of a Coxeter diagram by ringing them.
We then place a generating vertex inside the fundamental domain of the kaleidoscope, where the generating point is equidistant from all the reflecting hyperplanes that correspond to the ringed vertices. 
If a vertex is unringed, then the generating point lies on the corresponding reflecting hyperplane. 
The uniform tessellation is then the reflection of the generating point in all of the reflecting hyperplanes. 
\Cref{fig:wythoff-2d} shows some examples of Wythoff's construction applied to \dynkin[Coxeter,extended,affine mark=*]{C}{2}.

\begin{figure}
	\centering
	\subfloat[{\dynkin[Coxeter,extended,affine mark=*]C2}]{
		\centering
		\includegraphics[width=0.2\linewidth]{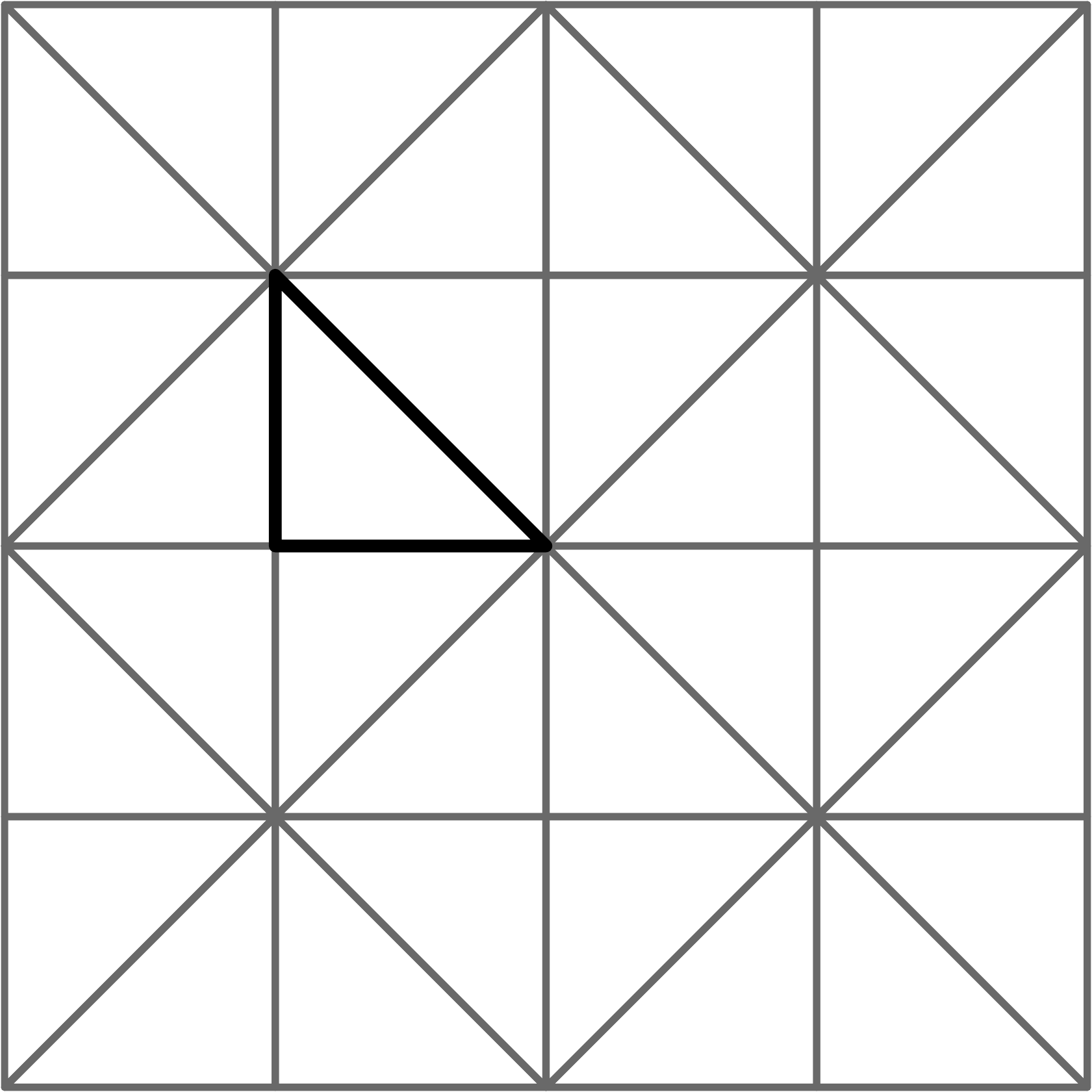}
	}
	\subfloat[
		{
			\begin{dynkinDiagram}[Coxeter,extended,affine mark=*]C2
				\protect\circleRoot 0
			\end{dynkinDiagram}
		}
	]{
		\centering
		\includegraphics[width=0.2\linewidth]{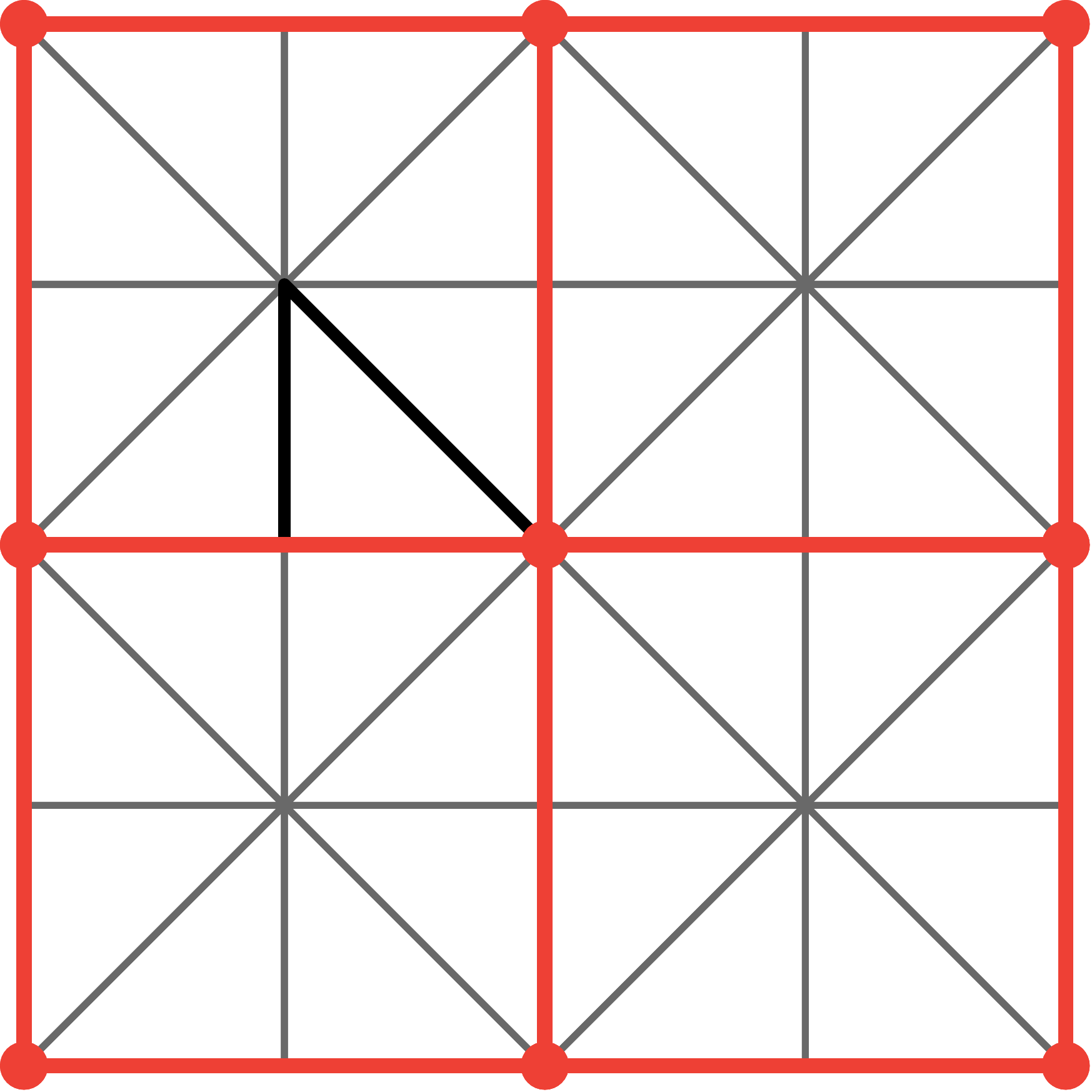}
	}
	\subfloat[
		{
			\begin{dynkinDiagram}[Coxeter,extended,affine mark=*]C2
				\protect\draw (0.7, 0.065) circle (3pt);
				% \protect\circleRoot 2
			\end{dynkinDiagram}
		}
	]{
		\centering
		\includegraphics[width=0.2\linewidth]{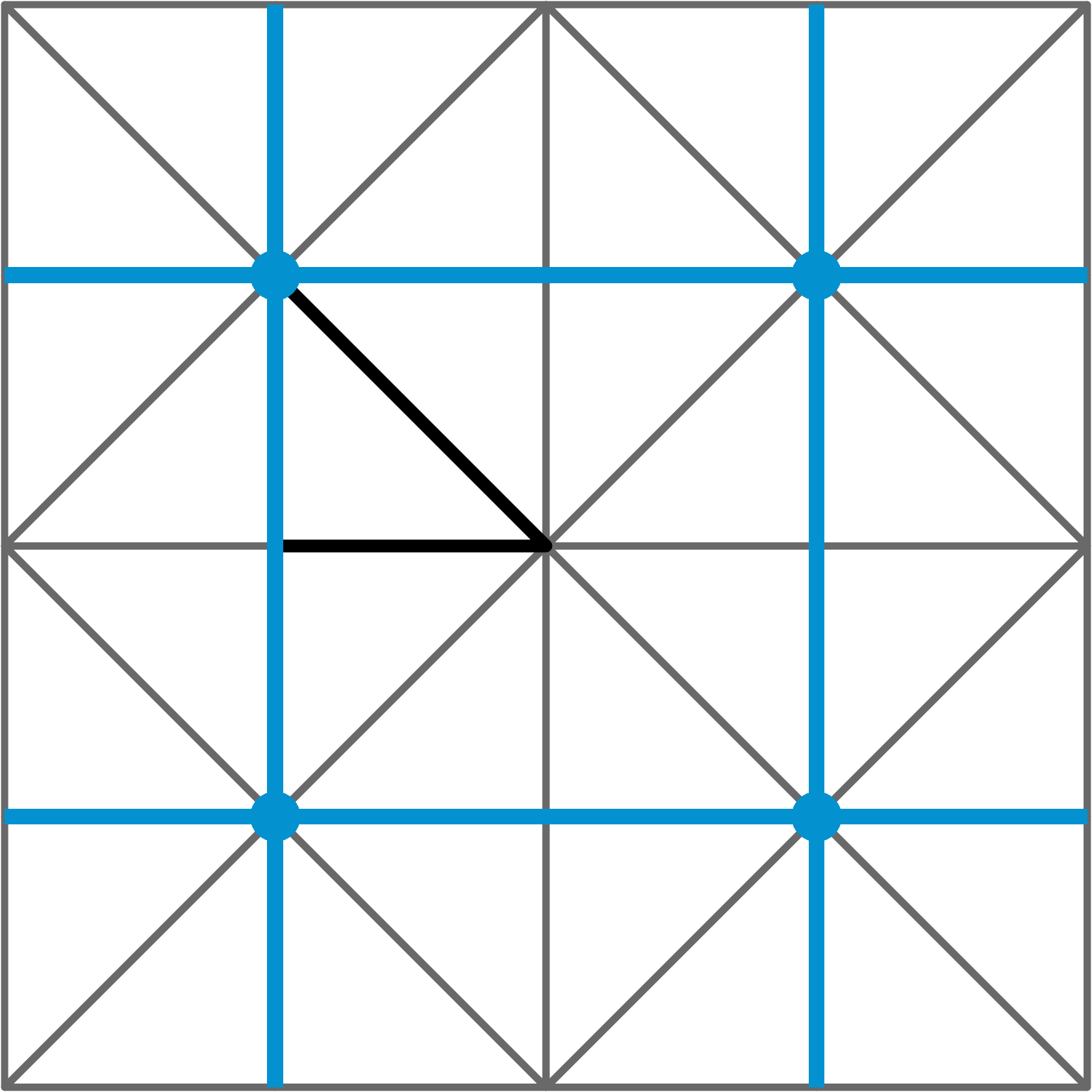}
	}
	\subfloat[
		{
			\begin{dynkinDiagram}[Coxeter,extended,affine mark=*]C2
				\protect\circleRoot 0
				% \protect\circleRoot 1
				% \protect\circleRoot 2
				\protect\draw (0.7, 0.065) circle (3pt);
				\protect\draw (0.35, 0.065) circle (3pt);
			\end{dynkinDiagram}
		}
	]{
		\centering
		\includegraphics[width=0.2\linewidth]{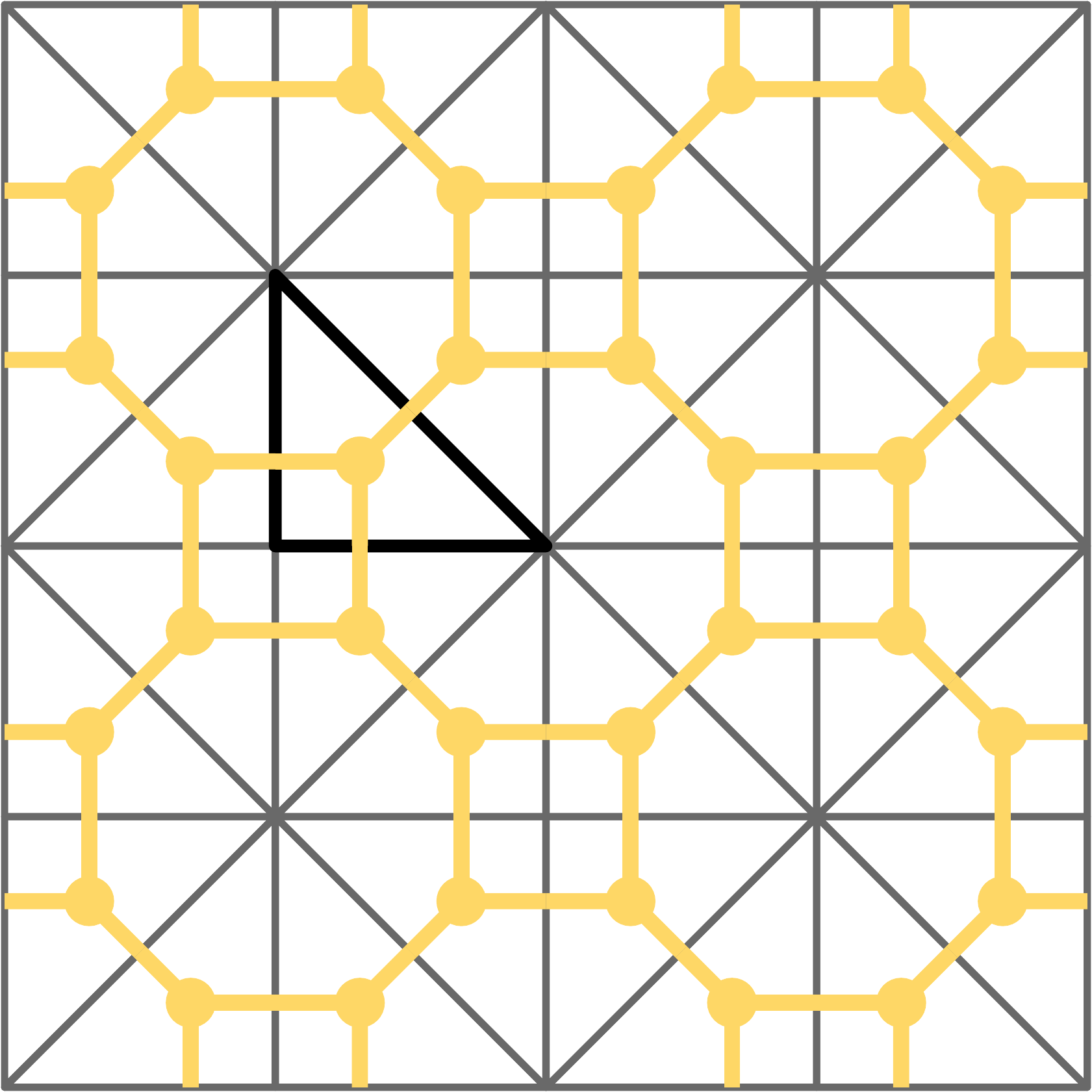}
	}
	\caption{
		Examples of Wythoff's construction applied to the Coxeter diagram \dynkin[Coxeter,extended,affine mark=*]{C}{2}.
		(a) The Coxeter kaleidoscope described in \cref{fig:coxeter_demo}.
		The fundamental region is highlighted in black.
		(b) Ringing the first vertex of the Coxeter diagram corresponds to placing the generating point at one of the corners of the fundamental region (red circle).
		The resultant tessellation is a square tiling.
		(c) By symmetry, ringing the third vertex of the diagram also leads to a square tiling.
		(d) Ringing all three vertices of the Coxeter diagram corresponds to placing the generating vertex at the centre of the fundamental region (equidistant from each reflecting hyperplane).
		In the resultant tessellation, two octahedra and one square meet at each vertex. 
	}
	\label{fig:wythoff-2d}
\end{figure}

\begin{figure}
	\centering
	\subfloat[{\dynkin[Coxeter,extended,affine mark=*]B3}]{
		\centering
		\includegraphics[width=0.2\linewidth]{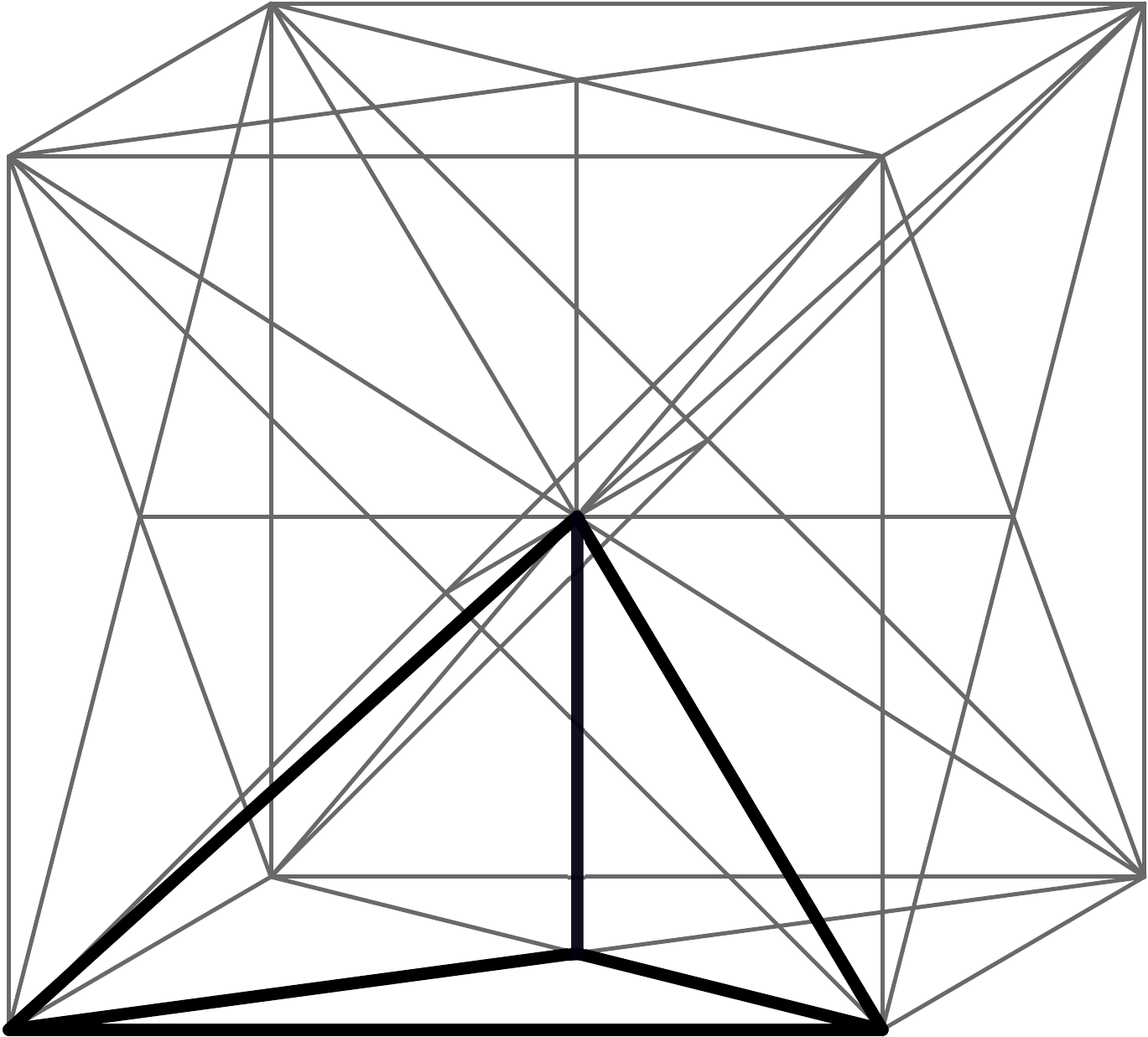}
	}
	\subfloat[
		{
			\begin{dynkinDiagram}[Coxeter,extended,affine mark=*]B3
				% \protect\circleRoot 3
				\protect\draw (0.52, 0.065) circle (3pt);
			\end{dynkinDiagram}
		}
	]{
		\centering
		\includegraphics[width=0.2\linewidth]{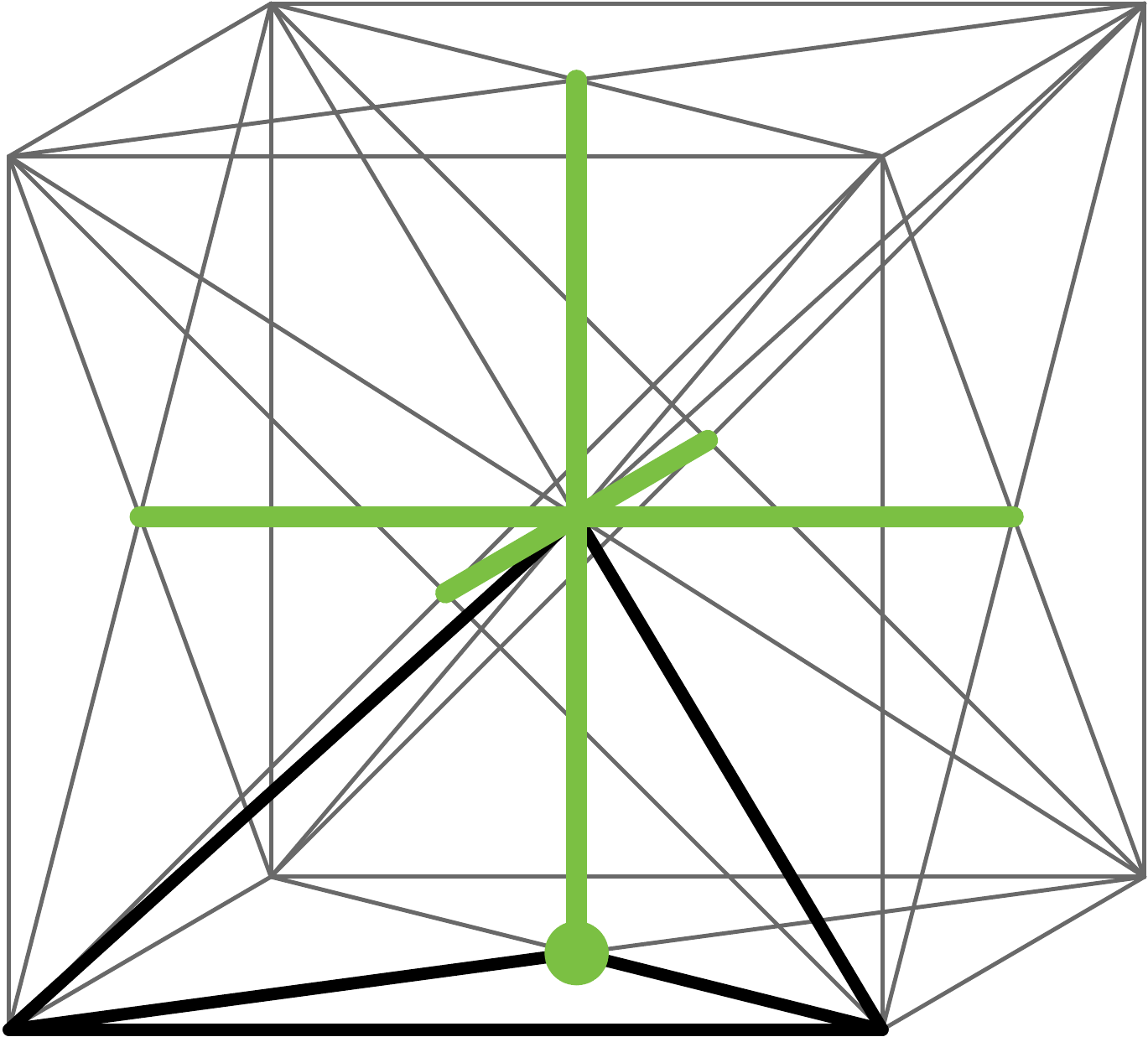}
	}
	\subfloat[
		{
			\begin{dynkinDiagram}[Coxeter,extended,affine mark=*]B3
				\protect\circleRoot 0
				% \protect\draw (0, 0.37) circle (3pt);
			\end{dynkinDiagram}
		}
	]{
		\centering
		\includegraphics[width=0.2\linewidth]{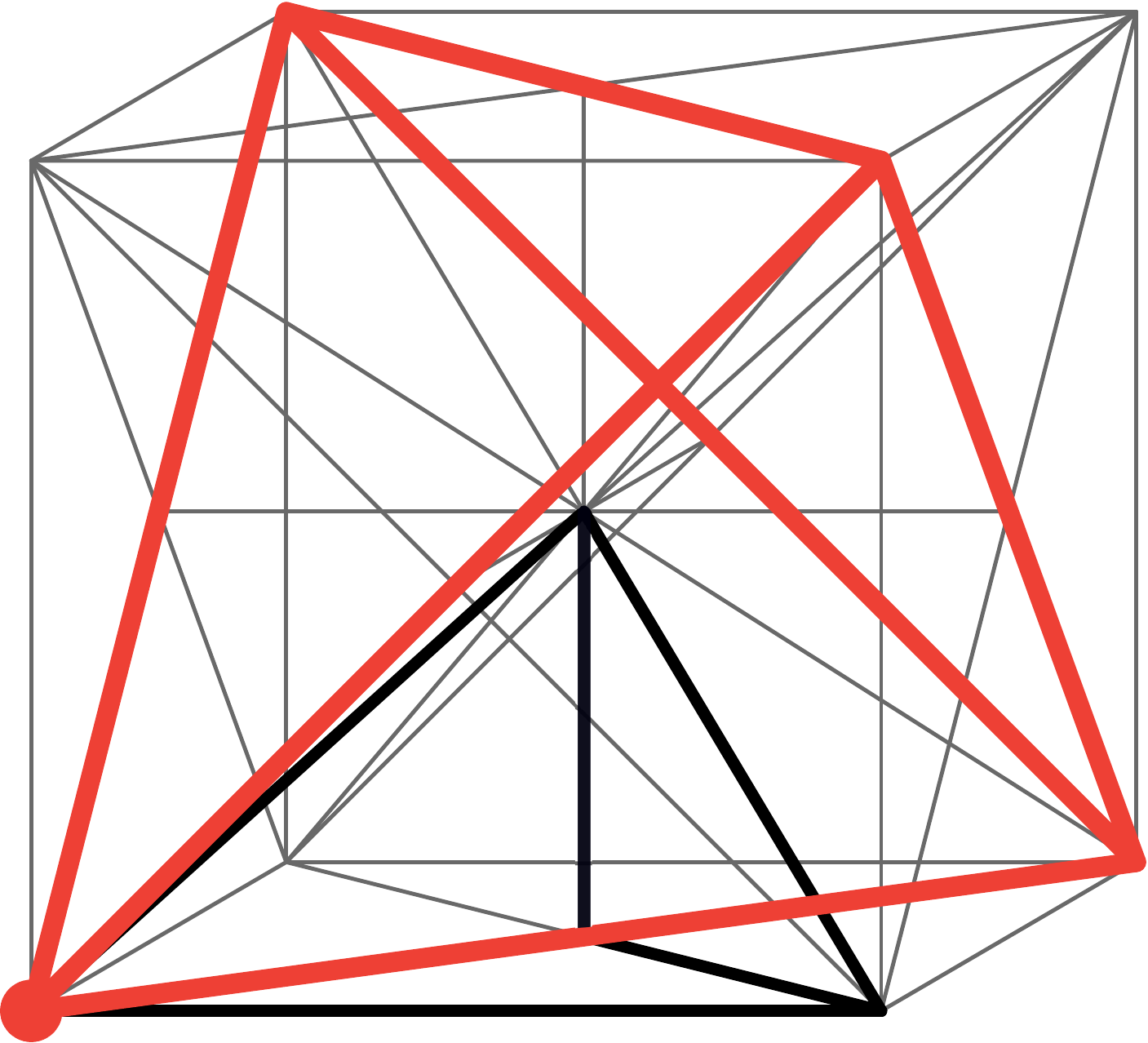}
	}
	\subfloat[
		{
			\begin{dynkinDiagram}[Coxeter,extended,affine mark=*]B3
				% \protect\circleRoot 1
				\protect\draw (0, -0.24) circle (3pt);
			\end{dynkinDiagram}
		}
	]{
		\centering
		\includegraphics[width=0.2\linewidth]{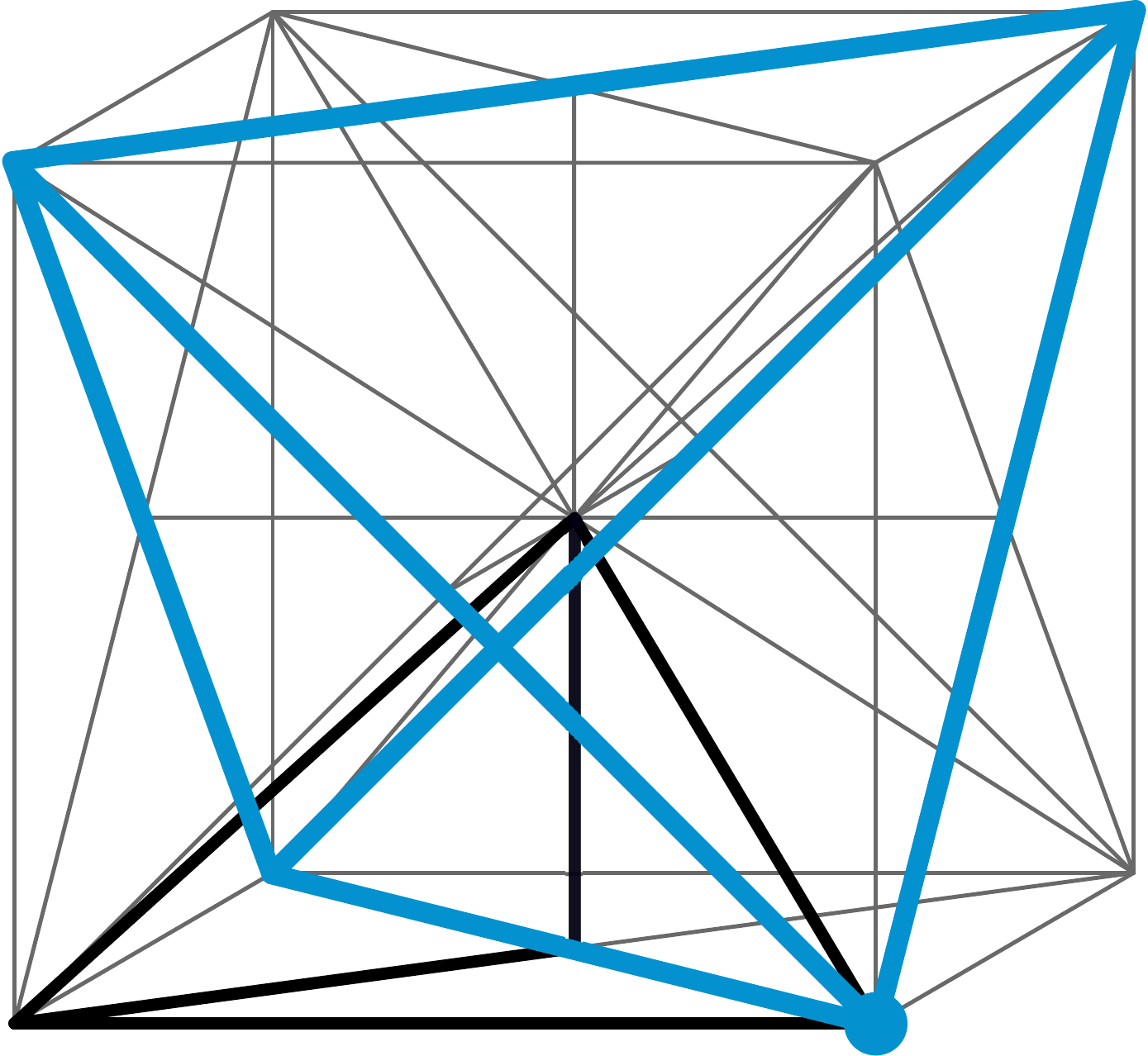}
	}
	\caption{
		Examples of Wythoff's construction applied to the Coxeter diagram \dynkin[Coxeter,extended,affine mark=*]{B}{3}.
		(a) The fundamental region is the tetrahedron highlighted in black.
		The reflecting hyperplanes are the four faces of the tetrahedron.
		(b) Applying Wythoff's construction with the markup 
		\begin{dynkinDiagram}[Coxeter,extended,affine mark=*]B3
			% \protect\circleRoot 3
			\protect\draw (0.52, 0.07) circle (3pt);
		\end{dynkinDiagram}
		produces the cubic honeycomb. 
		(c) Applying Wythoff's construction with the markup
		\begin{dynkinDiagram}[Coxeter,extended,affine mark=*]B3
			\protect\circleRoot 0
			% \protect\draw (0, 0.37) circle (3pt);
		\end{dynkinDiagram}
		produces a uniform tessellation where eight tetrahedra and six octahedra meet at each vertex. 
		This tessellation is sometimes called the tetrahedral-octahedral honeycomb. 
		(d) Gives the same tessellation as in (c) but rotated by $\pi / 2$.
	}
	\label{fig:wythoff-3d}
\end{figure}

The useful thing about Coxeter diagrams in our context is that the kaleidoscope corresponding to a particular Coxeter diagram is exactly a $d$-colex, where $d$ is the number of vertices of the diagram minus one. 
To see this, note that a Coxeter diagram naturally describes a simplicial complex. 
Furthermore, the complex will be $(d+1)$-vertex colorable as we can pick $(d+1)$ colors for the vertices of the fundamental region such that no two vertices have the same color, and the rest of the tessellation is formed by reflections of this fundamental region.
In 2D we have the following (affine) Coxeter diagrams: \dynkin[Coxeter,extended,affine mark=*]{C}{2}, \dynkin[Coxeter,extended,affine mark=*]{A}{2} and \dynkin[Coxeter,extended,affine mark=*]{G}{2}; each of which describes an affine Coxeter group and hence an infinite $2$-colex.
Furthermore, in certain cases different markups of a Coxeter diagram correspond to toric code lattices that are related to the colex via morphing. 
As an example, consider again \dynkin[Coxeter,extended,affine mark=*]{C}{2}.
As shown in \cref{fig:wythoff-2d}, 
\begin{dynkinDiagram}[Coxeter,extended,affine mark=*]C2
	\circleRoot 0
\end{dynkinDiagram}
and
\begin{dynkinDiagram}[Coxeter,extended,affine mark=*]C2
	% \circleRoot 2
	\draw (0.7, 0.08) circle (3pt);
\end{dynkinDiagram}
correspond to square tilings.
These are exactly the lattices we obtain if we morph the colorable $2$-balls with associated $2$-hyperoctahedron codes in the color code lattice described by \dynkin[Coxeter,extended,affine mark=*]{C}{2}.

Let us consider another example, this time in 3D.
In \cref{fig:wythoff-3d} we illustrate the tessellation corresponding to the Coxeter diagram \dynkin[Coxeter,extended,affine mark=*]{B}{3}.
This diagram gives a $3$-colex, whose dual is a uniform tessellation where one square, one truncated octahedron and two truncated cuboctahedra meet at every vertex (a tessellation called the cantitruncated cubic honeycomb).
In this case, the markups
\begin{dynkinDiagram}[Coxeter,extended,affine mark=*]B3
	% \protect\circleRoot 3
	\draw (0.52, 0.08) circle (3pt);
\end{dynkinDiagram}
,
\begin{dynkinDiagram}[Coxeter,extended,affine mark=*]B3
	% \draw (0, 0.39) circle (3pt);
	\protect\circleRoot 0
\end{dynkinDiagram}
and
\begin{dynkinDiagram}[Coxeter,extended,affine mark=*]B3
	\draw (0, -0.235) circle (3pt);
	% \protect\circleRoot 1
\end{dynkinDiagram}
give the cubic honeycomb and two tetrahedral-octahedral honeycombs (see \cref{fig:wythoff-3d}). 
These are exactly the toric code lattices that are produced by morphing the colorable $3$-balls with associated $3$-hyperoctahedron codes in the color code lattice described by \dynkin[Coxeter,extended,affine mark=*]{B}{3}.
Furthermore, these are exactly the toric codes (with the appropriate boundary conditions) that have a transversal logical $C_3$ gate~\cite{vasmer2019}.

One may ask how general the Coxeter diagram approach illustrated in the previous two examples is.
Whilst the markups of Coxeter diagrams are generically related to restricted color code lattices, the most elegant examples are color codes whose structure is described by Coxeter diagrams that are star graphs, e.g., \dynkin[Coxeter,extended,affine mark=*]C2 and \dynkin[Coxeter,extended,affine mark=*]B3. 
The fundamental region of a Coxeter group corresponding to such a diagram is a right-angled simplex, i.e., a simplex where all facets except one meet at right-angles at one vertex (see the fundamental regions in \cref{fig:coxeter_demo,fig:wythoff-3d}). 
In such cases, all the colorable $d$-balls of one color will have associated $d$-hyperoctahedron codes. 
Therefore, we can morph these colorable $d$-balls to obtain $d$ toric codes with a transversal gate implemented by $C_d$.

In 4D, there is a single affine Coxeter diagram with the required star graph structure: \dynkin[Coxeter,extended,affine mark=*]D4.
This diagram describes a $4$-colex dual to a tessellation where four truncated octaplexes and one tesseract meet at each vertex.
An octaplex (or hyper-diamond) is a regular 4-polytope described by the ringed Coxeter diagram 
\begin{dynkinDiagram}[Coxeter]F4
	% \circleRoot 1
	\draw (0, 0.08) circle (3pt);
\end{dynkinDiagram},
where six octahedra meet at each vertex of the polytope and three meet at each edge.
From \dynkin[Coxeter,extended,affine mark=*]D4, we can construct
\begin{dynkinDiagram}[Coxeter,extended,affine mark=*]D4
	\circleRoot 0
\end{dynkinDiagram},
a regular tessellation where every cell is a $4$-hyperoctahedron, and three 4-cells meet at each face. 
Morphing the colorable $4$-balls with associated $4$-hyperoctahedron codes in \dynkin[Coxeter,extended,affine mark=*]D4 produces four toric codes defined on 
\begin{dynkinDiagram}[Coxeter,extended,affine mark=*]D4
	\circleRoot 0
\end{dynkinDiagram}
tessellations with a transversal gate implemented by $C_4$. 
With appropriate boundary conditions, these are the 4D toric codes with a transversal logical $C_4$ gate detailed in~\cite{jochym-oconnor2021}.

There are no further affine Coxeter diagrams in higher-dimensional Euclidean space with the required star graph structure. 
However, Coxeter groups also describe tessellations in hyperbolic space, and there are hyperbolic Coxeter groups with star graph Coxeter diagrams in 2D and 3D~\cite{humphreys1990,davis2008}. 
In 3D hyperbolic space, we have the Coxeter diagram 
\begin{tikzpicture}[baseline]
	\begin{dynkin}[vertical shift]D4
		\node[] at (0.2, 0.25) {\scriptsize 5};
	\end{dynkin}
\end{tikzpicture}.
The color code defined on this $3$-colex has stabilizer weights of 120, 24, 10, 8, 6, and 4.
Wythoff's construction gives us the tessellations detailed in \cref{tab:wythoff}.
\begin{table}
	\begin{tabular}{|| c | c | c | c ||}
		\hline
		Wythoff construction & Name & Structure & Stabilizer weights \\
		\hline
		\begin{tikzpicture}[baseline]
			\begin{dynkin}[vertical shift]D4
				\node[] at (0.2, 0.25) {\scriptsize 5};
				% \circleRoot 1
				\draw (0, 0.07) circle (3pt);
			\end{dynkin}
		\end{tikzpicture} 
		&
		Order-4 dodecahedral
		& 
		8 dodecahedra at each vertex 
		&
		5 and 6 \\ 
		\hline 
		\begin{tikzpicture}[baseline]
			\begin{dynkin}[vertical shift]D4
				\node[] at (0.2, 0.25) {\scriptsize 5};
				% \circleRoot 4
				\draw (0.52, 0.36) circle (3pt);
			\end{dynkin}
		\end{tikzpicture} 
		&
		Alternated order-5 cubic
		& 
		\makecell{20 tetrahedra and \\ 12 icosahedra at each vertex}
		&
		3 and 30 \\ 
		\hline
	\end{tabular}
	\caption{
		Table showing some of the tessellations that can be constructed from the Coxeter diagram
		{\begin{tikzpicture}[baseline]
			\begin{dynkin}[vertical shift]D4
				\protect\node[] at (0.2, 0.25) {\scriptsize 5};
			\end{dynkin}
		\end{tikzpicture}}
		using Wythoff's construction. 
		Toric codes defined on these tessellations have the stabilizer weights shown in the fourth column.
	}
	\label{tab:wythoff}
\end{table}
We expect that morphing the colorable $3$-balls with associated $3$-hyperoctahedron codes in 
\begin{tikzpicture}[baseline]
	\begin{dynkin}[vertical shift]D4
		\node[] at (0.2, 0.25) {\scriptsize 5};
	\end{dynkin}
\end{tikzpicture}
would give toric codes with a transversal gate implemented by $C_3$ gates. 
As hyperbolic toric codes are finite rate, this may be an interesting set of codes to explore.

%%%%%%%%%%%%%%%%%%%%%%%%%%%%%%%%
%% Morphing the `BNB' lattice %% 
%%%%%%%%%%%%%%%%%%%%%%%%%%%%%%%%
\section{Morphing the `BNB' lattice}
\label{app:bnb}

Here we briefly discuss an example of a 3D HCT code that contains both color code and toric code regions.
To construct the code we apply our morphing procedure to the lattice described in~\cite{brown2016a} (which we call the `BNB lattice' after the authors).
The BNB lattice contains colorable $3$-balls with associated $3$-hyperoctahedron codes, but they are not all the same color (in fact they are split evenly between the four colors) and they do not cover all the qubits in the code.
Suppose we start with a color code defined on a BNB lattice that tessellates a closed manifold.
One can verify that when we morph the $3$-hyperoctahedron codes, we obtain a tetrahedral-octahedral honeycomb (with a slight modification; see \cref{fig:bnb}).
In the tetrahedral-octahedral honeycomb, six tetrahedra and eight tetrahedra meet at each vertex.
It can be constructed using Wythoff's construction with the Coxeter diagram 
\begin{dynkinDiagram}[Coxeter,extended,affine mark=*]B3
	\protect\circleRoot 0
\end{dynkinDiagram} 
(see \cref{app:coxeter}).

\begin{figure}
	\centering
	\includegraphics[width=0.4\linewidth]{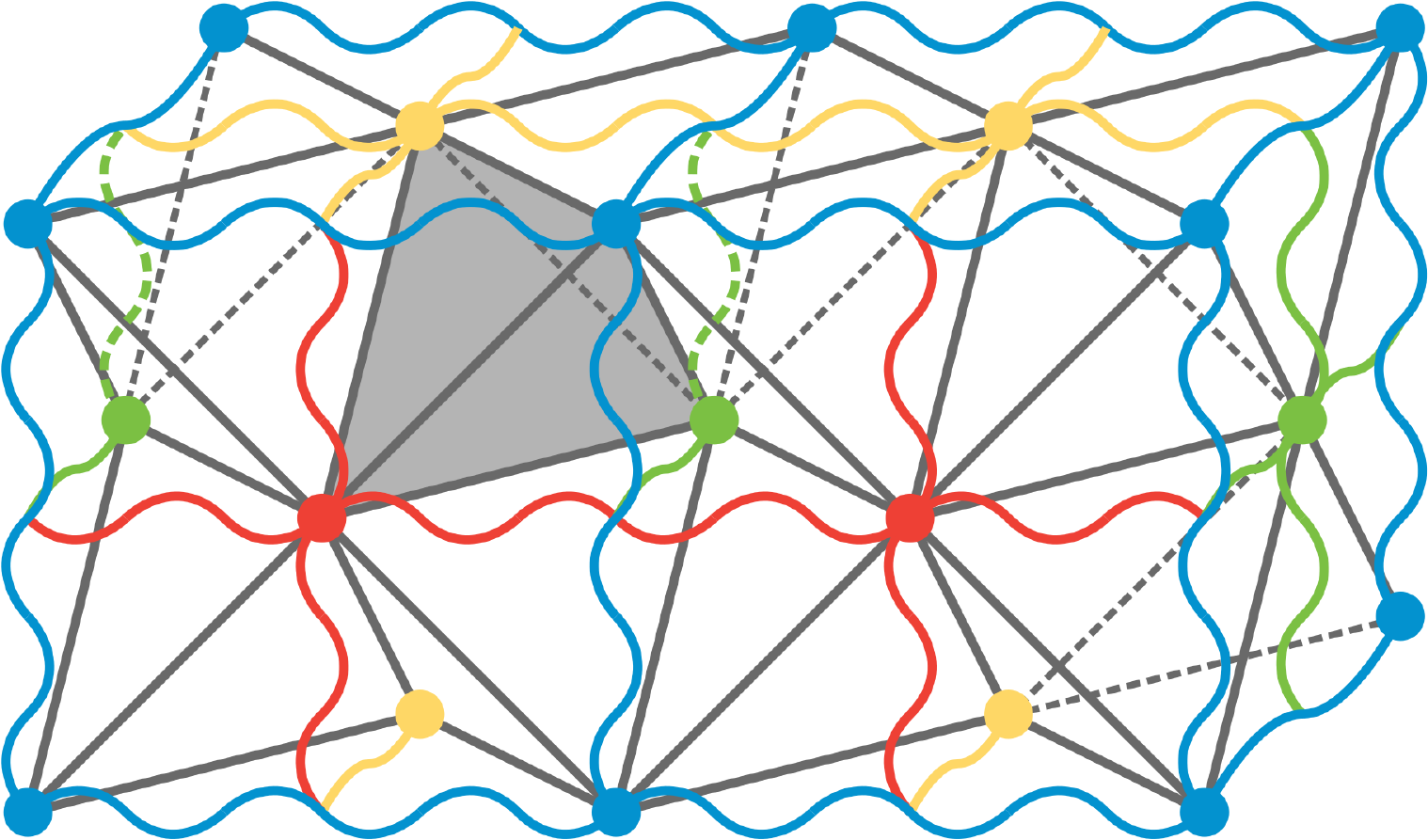}
	\caption{
		A portion of the morphed BNB lattice. 
		The grey edges show the underlying tetrahedral-octahedral lattice. 
		We have one qubit for each edge linking vertices of the same color and one qubit for each tetrahedron (one of which we shade in grey).
		$X$-type stabilizer generators are associated with vertices and $Z$-type stabilizer generators are associated with edges linking vertices of different colors.
	}
	\label{fig:bnb}
\end{figure}

The HCT code defined on the morphed BNB lattice has a particularly simple implementation of a non-Clifford gate.
Namely, we apply $R_3$ and $R_3^\dagger$ to the simplex qubits (following the qubit bipartition of the parent color code) and $C_3$ to the triple of edge qubits at each octahedron.
If we choose a lattice with tetrahedral boundary conditions (as in~\cite{brown2016a}) this gate will implement a logical $R_3$ gate.

%%%%%%%%%%%%%%%%%%%%%%%%%%%%%%%%%%%%
%% Morphing stellated color codes %%
%%%%%%%%%%%%%%%%%%%%%%%%%%%%%%%%%%%%
\section{Morphing stellated color codes}
\label{app:stell}

In this appendix we apply our morphing procedure to stellated color codes~\cite{kesselring2018}.
We focus on distance three stellated color codes based on the 4.8.8 tessellation, though we expect our results would generalize to the entire code family.
Following~\cite{kesselring2018}, we parametrize the codes by an integer $s \geq 3$, the order of the rotational symmetry of the code.
The codes in the sub-family we consider have parameters $\nkd{5s}{s}{3}$.
We explain the construction of distance-three stellated color codes via examples in \cref{fig:stella}.
We note that when $s$ is odd, the stellated codes contain twist defects~\cite{bombin2010a,kesselring2018}.

\begin{figure}
	\centering
	\subfloat[$s=4$]{
		\centering
		\includegraphics[width=0.2\linewidth]{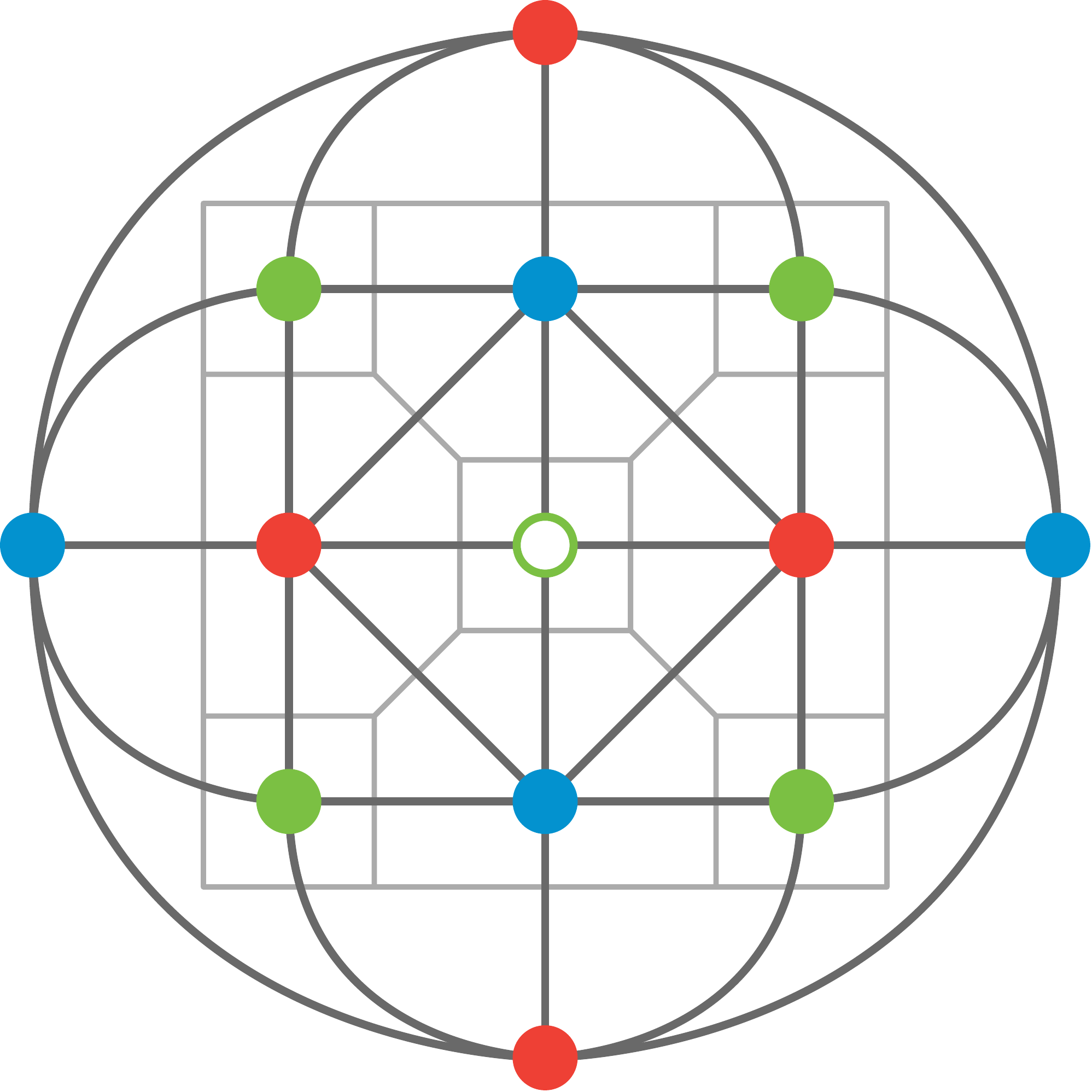}
	}
	\quad
	\subfloat[$s=4$ morphed]{
		\centering
		\includegraphics[width=0.2\linewidth]{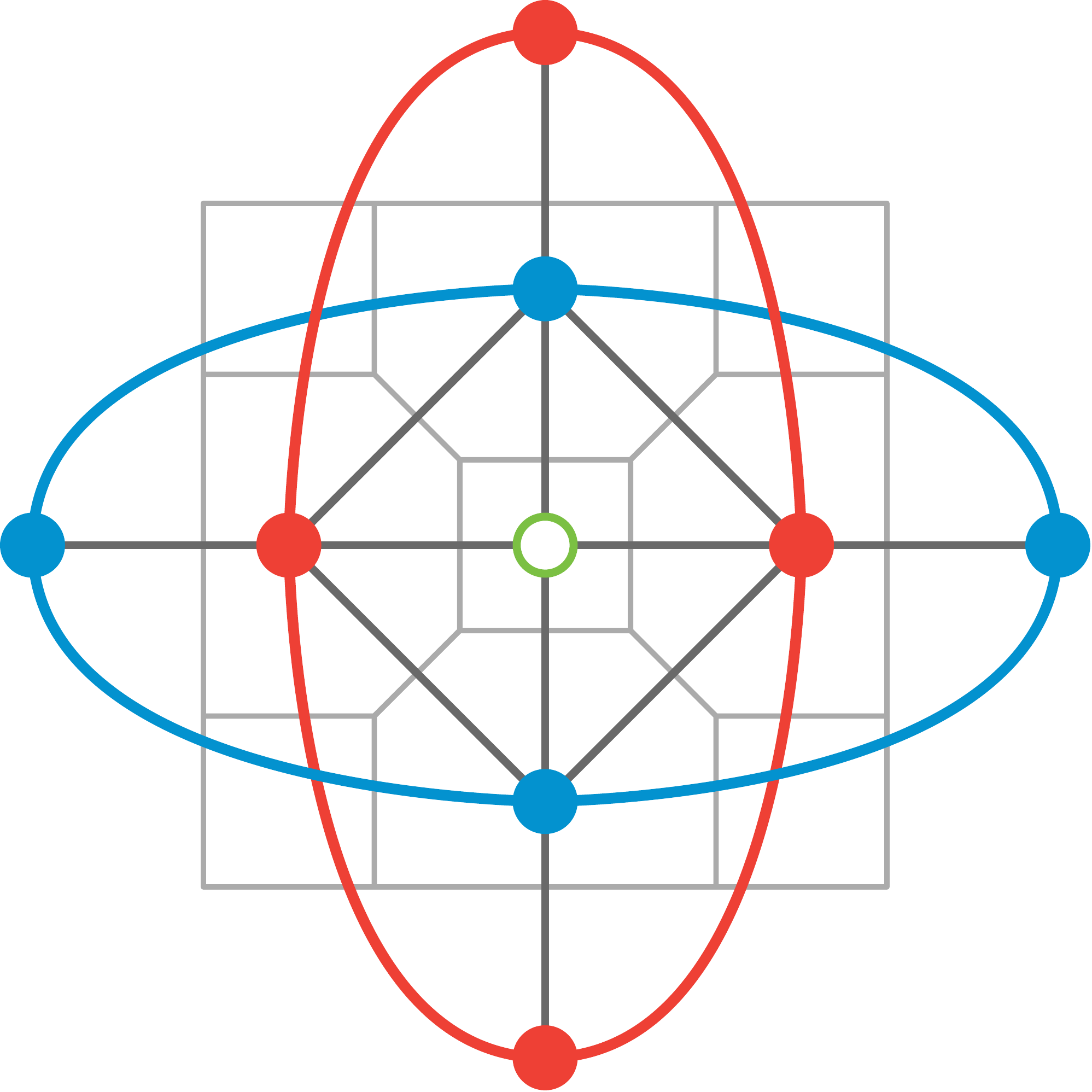}
		\label{sfig:s4}
	}
	\quad
	\subfloat[$s=5$]{
		\centering
		\includegraphics[width=0.2\linewidth]{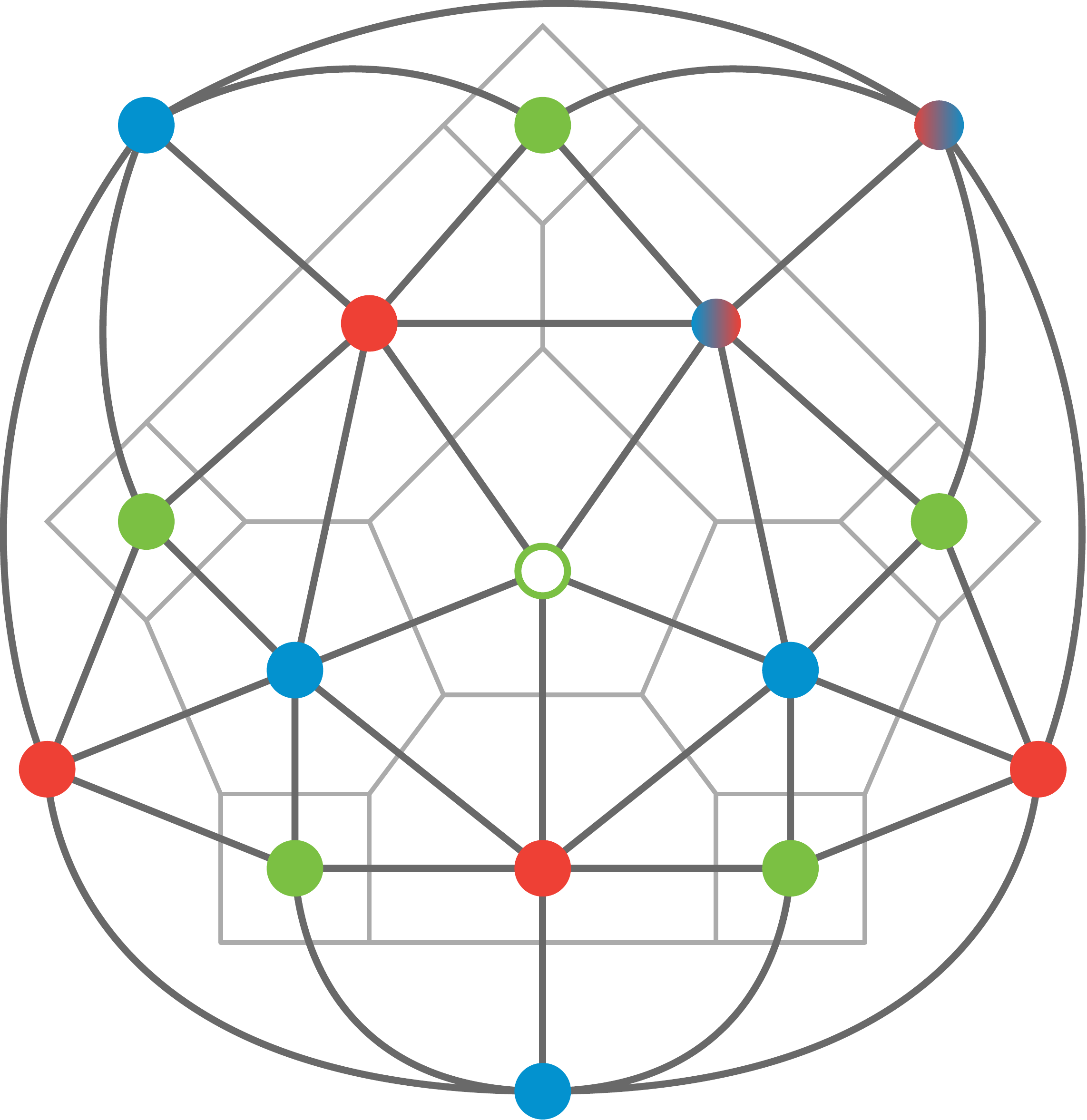}
	}
	\quad
	\subfloat[$s=5$ morphed]{
		\centering
		\includegraphics[width=0.2\linewidth]{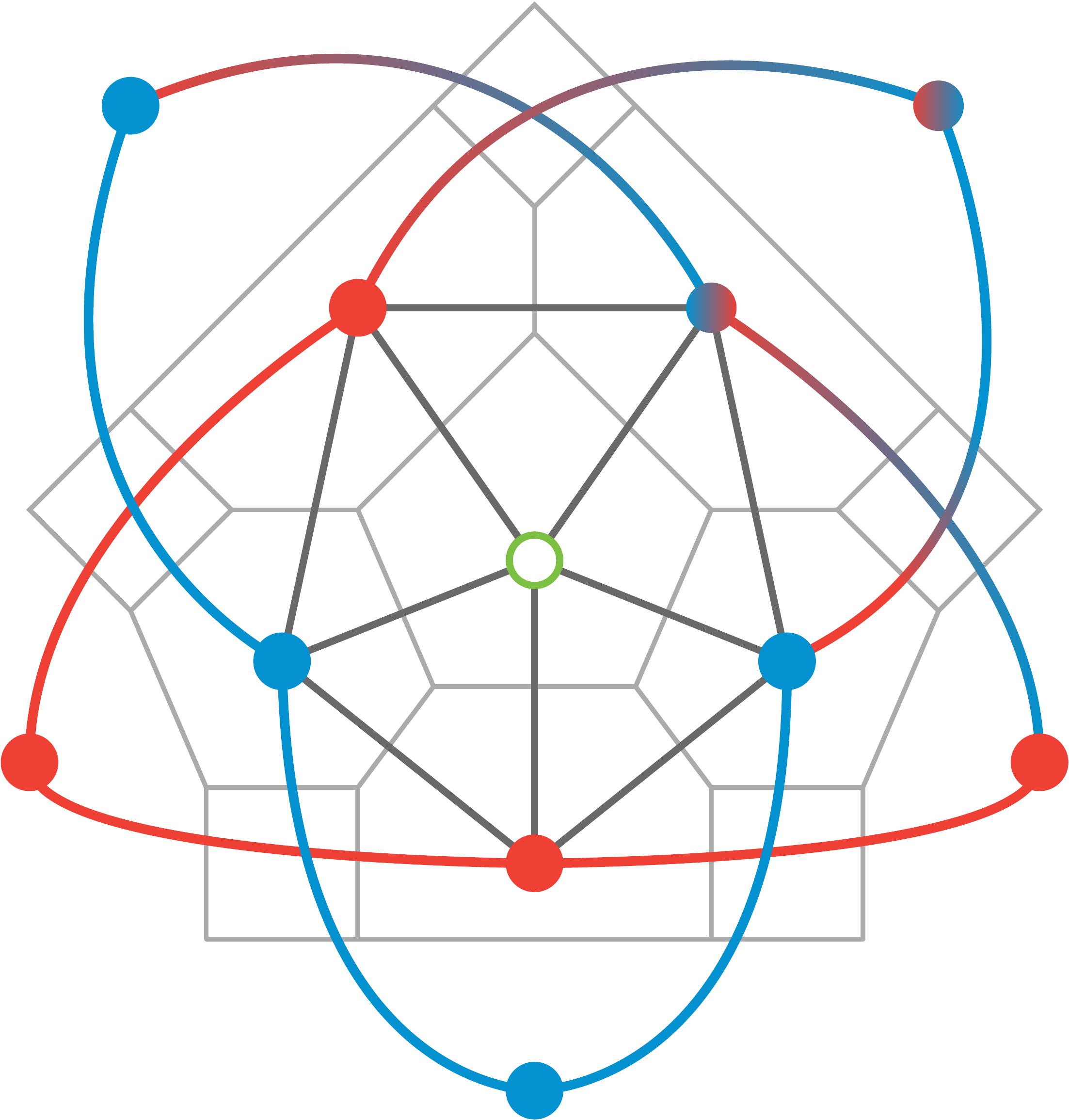}
		\label{sfig:s5}
	}
	\caption{
		Morphing stellated color codes.
		We show the (a)$s=4$ and (c)$s=5$ distance three stellated color codes, which have parameters \nkd{20}{4}{3} and \nkd{25}{5}{3}, respectively.
		We show the primal lattice in light grey for ease of comparison with~\cite{kesselring2018}.
		In (b) and (d) we show the corresponding morphed codes, with parameters \nkd{12}{4}{2} and \nkd{15}{5}{2}.
		We note the central vertex has no associated stabilizer, which we indicate by making it hollow.
		For odd $s$, the color code colorability condition is unsatisfied, indicating the presence of a twist defect.
		We represent this pictorially by blended $r$/$b$ vertices and edges.
	}
	\label{fig:stella}
\end{figure}

We now apply our morphing procedure to distance-three stellated color codes. 
We morph all colorable $2$-balls with associated \nkd{4}{2}{2} codes, giving a family of stellated HCT codes with parameters $\nkd{3s}{s}{2}$.
These codes have a high encoding rate and inherit the tranversal gates of the stellated color codes.
In particular, for even $s$ stellated color codes have a logical $S^{\otimes s} = \widetilde{S}$.
In the morphed stellated color codes, the logical $S^{\otimes s}$ gate is implemented by $S$ and $S^\dagger$ on the color code qubits (as before) and by $C_2$ gates on the pairs of toric code qubits arising from each \nkd{4}{2}{2} code.

One interesting avenue to explore is the generalization of these results to higher dimensions. 
It seems plausible that stellated color codes could exist in 3D, and we expect these codes would have fault-tolerant logical $R_3$ gates. 

%%%%%%%%%%%%%%%%%%%%%%%%%%%%%%
%% HCT code decoder details %%
%%%%%%%%%%%%%%%%%%%%%%%%%%%%%%
\section{HCT code decoder details \label{app:local_lift}}

In this appendix, we explain the local lift step of our 2D HCT code decoder, and we comment on the overall time complexity of our decoder.

Let $\mathcal L$ be a $2$-colex and let $\widetilde{\mathcal L}$ be an HCT code lattice produced from $\mathcal L$ by morphing a subset of the colorable $2$-balls of $\mathcal L$. 
And let $\mu \subseteq \widetilde{\mathcal L}_1$ be the matching produced in the second step of our decoding algorithm (see \cref{subsec:decoder}).
We use $\mu |_v = \{ e \in \mu : v \subseteq e \}$ to denote the restriction of $\mu$ to the neighborhood of a vertex $v$.
In the color code regions of $\widetilde{\mathcal L}$, for each $r$-vertex $v$ in the matching, local lift returns a set of faces $\varphi \subseteq \mathcal L_2 \cap \St_2 (v)$, such that $(\partial_{2,1} \varphi)|_v = \mu|_v$.
In the toric code regions of $\widetilde{\mathcal L}$, we first work out which faces local lift would return if it were to be applied in $\mathcal L$ with the same matching. 
For a given $r$-vertex $v$ on the boundary of a morphed colorable $2$-ball, this set of faces will be equal to the support of a logical $Z$ operator of the ball code.
Local lift then simply returns the edge qubits acted on by this logical operator. 
\Cref{fig:lift_ex} shows an example.

\begin{figure}
	\centering
	\subfloat[]{
		\centering
		\includegraphics[width=0.2\linewidth]{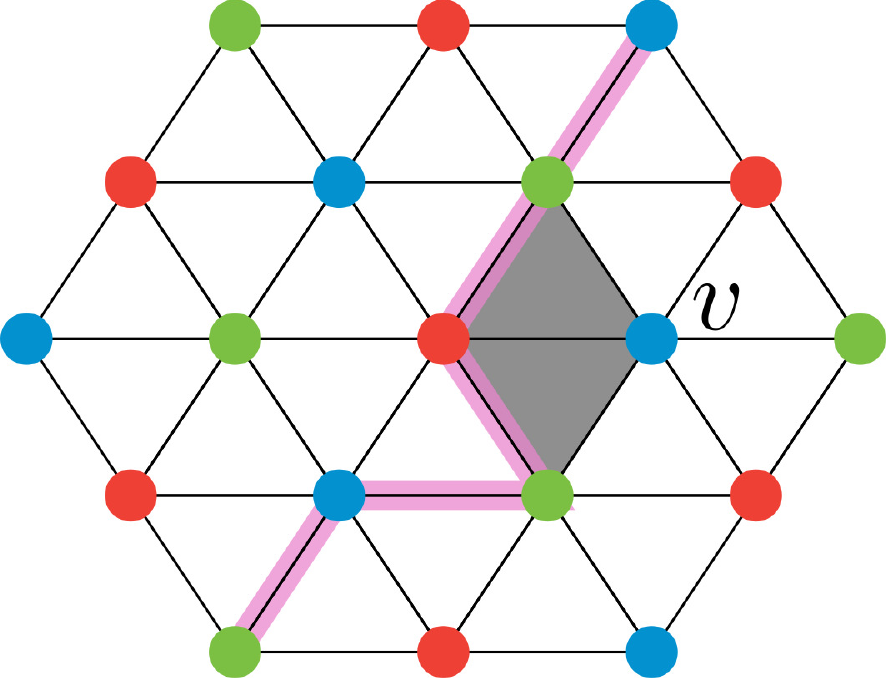}
	}
	\quad
	\subfloat[]{
		\centering
		\includegraphics[width=0.2\linewidth]{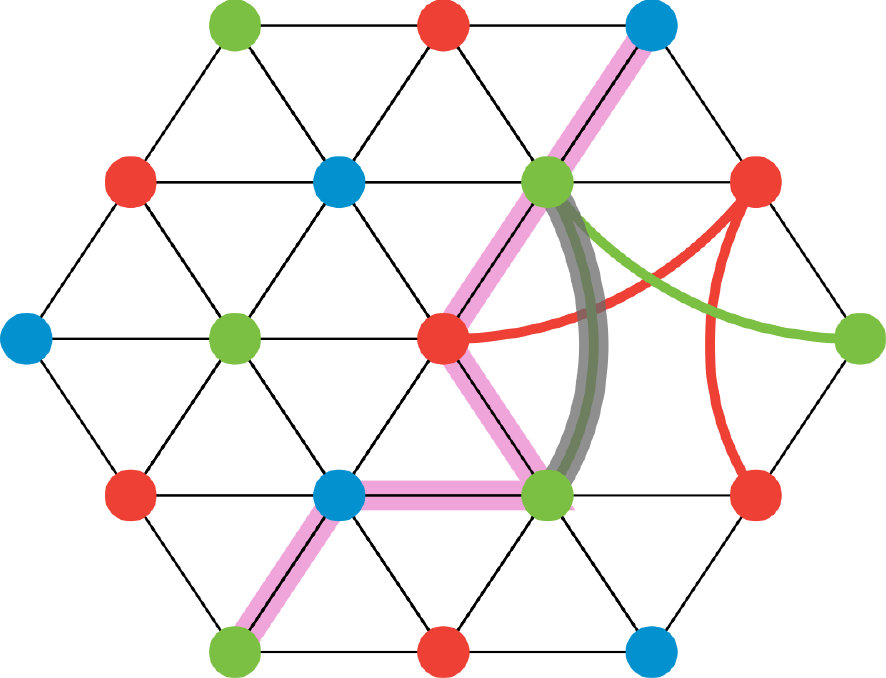}
	}
	\caption{
		The local lift subroutine of our HCT code decoder. 
		(a) 
		We show part of a $2$-colex with a matching (subset of edges) highlighted in pink.
		If we apply local lift at the central $r$-vertex, we obtain the highlighted grey faces.
		These faces have the required property that their boundary equals the matching, when restricted to the neighborhood of $v$.
		(b)
		Applying local lift for the same matching but now in an HCT code lattice. 
		For the ball code defined on $\mathcal B^v$, the $Z$ operator with support equal to the highlighted faces in (a) is logical operator acting on the edge qubit highlighted in grey in (b).
		Therefore, for this matching local lift returns the highlighted edge when it is applied at the central $r$-vertex.
	}
	\label{fig:lift_ex}
\end{figure}

We claim that our decoder has time complexity $O(N^3)$, where $N$ is the number of physical qubits in the code. 
First we note that there are $O(N)$ vertices and $O(N)$ edges in the 2D HCT code lattice.
In the first step of the decoder, we call MWPM twice, giving a time complexity of $O(N^3)$ due to the Blossom algorithm.
In the second step of the decoder we iterate through the edges in the matching, where (in the worst case) there will be $O(N)$ such edges. 
For each edge $e$ the local modification step only requires simple calculations in the neighbourhood of $e$, which is constant size by assumption. 
Therefore each local modification takes constant time and this the second step has time complexity $O(N)$.
The third step of the decoder is similar: we iterate through (at most) $O(N)$ vertices and $O(N)$ edges, applying a constant time local lift at each vertex.
Hence the third step also has time complexity $O(N)$, implying that the overall time complexity is $O(N^3)$.

%%%%%%%%%%%%%%%%%%%%%%%%%%%%%%%%%%
%% Comparing distillation costs %%
%%%%%%%%%%%%%%%%%%%%%%%%%%%%%%%%%%
\section{Comparing distillation costs}
\label{app:cost}

We evaluate the utility of our 10-to-1 MSD protocol by comparing the required distillation cost $C$ to achieve a certain target error rate $p_{\mathrm{targ}}$ assuming an input error rate of $p=0.01$, for three multi-round MSD protocols; see \Cref{tab:cost}.
The MEK protocol~\cite{meier2012} uses 15-to-1 and 10-to-1 distillation and the triorthogonal protocol~\cite{bravyi2012} uses triorthogonal codes.
In the MEK+ protocol, we augment the MEK protocol with 10-to-1 distillation.

\begin{table}	
\begin{tabular}{c c c c c c}
	\hline\hline\noalign{\vskip 1mm}   
	$-\log_{10} p_{\mathrm{targ}}$ & $C_{\Delta}$ & $C_{\mathrm{MEK}}$ & $C_{\mathrm{MEK}+}$ & Sequence & $-\log_{10} p_{\mathrm{actual}}$ \\[1mm]
	\hline\noalign{\vskip 1mm}
		3 & 5.521 & 5.521 & 5.521 & 5 & 3.030 \\
		4 & 17.44 & 17.44 & 17.44 & 15 & 4.443 \\
		5 & 27.86 & 27.86 & 27.86 & 5-5 & 5.104 \\
		6 & 56.07 & 83.99 & 43.39 & 10-5 & 6.969 \\
		7 & 58.30 & 83.99 & 69.41 & 10-10 & 7.923 \\
		8 & 89.26 & 139.3 & 130.2 & 10-15 & 10.32 \\
		9 & 139.3 & 139.3 & 130.2 & 10-15 & 10.32 \\
		10 & 179.4 & 261.7 & 130.2 & 10-15 & 10.32 \\
		11 & 179.4 & 261.7 & 217.0 & 10-5-5 & 12.98 \\
		12 & 187.9 & 418.0 & 217.0 & 10-5-5 & 12.98 \\
		13 & 225.6 & 418.0 & 347.1 & 10-10-5 & 14.89 \\
		14 & 285.6 & 419.9 & 347.1 & 10-10-5 & 14.89 \\
		15 & 315.5 & 696.7 & 555.3 & 10-10-10 & 15.85 \\
		16 & 406.2 & 696.7 & 650.9 & 10-5-15 & 19.36 \\
		17 & 529.5 & 696.7 & 650.9 & 10-5-15 & 19.36 \\
		18 & 574.1 & 1260 & 650.9 & 10-5-15 & 19.36 \\
		19 & 574.1 & 1260 & 650.9 & 10-5-15 & 19.36 \\
		20 & 574.1 & 1260 & 1041 & 10-10-15 & 22.22 \\ 
		21 & 575.9 & 1260 & 1041 & 10-10-15 & 22.22 \\
		22 & 604.3 & 1308 & 1041 & 10-10-15 & 22.22 \\
		23 & 652.3 & 2090 & 1085 & 10-5-5-5 & 25.01 \\
		24 & 731.5 & 2090 & 1085 & 10-5-5-5 & 25.01 \\
		25 & 853.1 & 2090 & 1085 & 10-5-5-5 & 25.01 \\
		26 & 914.0 & 2090 & 1735 & 10-10-5-5 & 28.83 \\
		27 & 947.5 & 2100 & 1735 & 10-10-5-5 & 28.83 \\
		28 & 1015 & 2181 & 1735 & 10-10-5-5 & 28.83 \\
		29 & 1125 & 3483 & 1954 & 10-15-15 & 29.48 \\
		30 & 1301 & 3483 & 2776 & 10-10-10-5 & 30.74 \\ 
	\hline\hline
\end{tabular}
\caption{
	Distillation costs for various MSD protocols to achieve target error rates $p_{\mathrm{targ}}$ assuming an input error rate of $p=0.01$.
	The MEK+ protocol ($C_{\mathrm{MEK}+}$) equals or betters the original MEK protocol~\cite{meier2012} ($C_{\mathrm{MEK}}$) and betters the triorthogonal protocol~\cite{bravyi2012} ($C_{\Delta}$) for $p_{\mathrm{targ}} \in \{ 10^{-6}, 10^{-9}, 10^{-10} \}$.
	The `Sequence' column shows the sequence of distillation protocols used in the MEK+ protocol, where 5 is 10-to-2, 10 is 10-to-1, and 15 is 15-to-1; $p_{\mathrm{actual}}$ is the output error rate achieved by the sequence.
}
\label{tab:cost}
\end{table}

We can also directly compare the distillation cost of the 10-to-1 protocol against the 15-to-1 protocol; see \cref{fig:cost_single}.
Assuming the optimistic noise model for $CCZ$ states described in \cref{subsec:msd_protocol}, we find that the 10-to-1 protocol has a lower cost than the 15-to-1 protocol for target error rates $p_{\mathrm{targ}} \in \{ 10^{-3}, 10^{-5}, 10^{-6}, 10^{-7}, 10^{-12}, 10^{-13}, 10^{-14}, 10^{-15} \}$.
And assuming the pessimistic noise model for $CCZ$ states, we find that the 10-to-1 protocol has a lower cost than the 15-to-1 protocol for target error rates $p_{\mathrm{targ}} \in \{ 10^{-3}, 10^{-5}, 10^{-6} \}$. 

\begin{figure}
	\centering
	\includegraphics[width=0.5\textwidth]{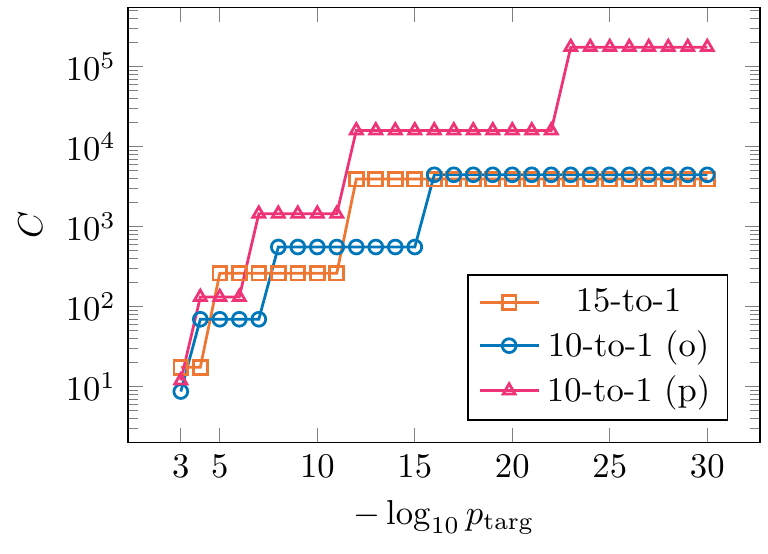}
	\caption{A plot showing the distillation cost $C$ for different target error rates $p_{\mathrm{targ}}$, assuming an input error rate of $p=0.01$. For certain values of $p_{\mathrm{targ}}$, we find that the cost of our 10-to-1 protocol is lower than the cost of the 15-to-1 protocol, where (o) and (p) refer to the optimistic and pessimistic noise models for $CCZ$ states described in \cref{subsec:msd_protocol}, respectively.}
	\label{fig:cost_single}
\end{figure}

% \bibliography{library,library_extra}
\input{main.bbl}

\end{document}

%% file: main.bbl

%% file: main.bbl
\begin{thebibliography}{91}%
\makeatletter
\providecommand \@ifxundefined [1]{%
 \@ifx{#1\undefined}
}%
\providecommand \@ifnum [1]{%
 \ifnum #1\expandafter \@firstoftwo
 \else \expandafter \@secondoftwo
 \fi
}%
\providecommand \@ifx [1]{%
 \ifx #1\expandafter \@firstoftwo
 \else \expandafter \@secondoftwo
 \fi
}%
\providecommand \natexlab [1]{#1}%
\providecommand \enquote  [1]{``#1''}%
\providecommand \bibnamefont  [1]{#1}%
\providecommand \bibfnamefont [1]{#1}%
\providecommand \citenamefont [1]{#1}%
\providecommand \href@noop [0]{\@secondoftwo}%
\providecommand \href [0]{\begingroup \@sanitize@url \@href}%
\providecommand \@href[1]{\@@startlink{#1}\@@href}%
\providecommand \@@href[1]{\endgroup#1\@@endlink}%
\providecommand \@sanitize@url [0]{\catcode `\\12\catcode `\$12\catcode
  `\&12\catcode `\#12\catcode `\^12\catcode `\_12\catcode `\%12\relax}%
\providecommand \@@startlink[1]{}%
\providecommand \@@endlink[0]{}%
\providecommand \url  [0]{\begingroup\@sanitize@url \@url }%
\providecommand \@url [1]{\endgroup\@href {#1}{\urlprefix }}%
\providecommand \urlprefix  [0]{URL }%
\providecommand \Eprint [0]{\href }%
\providecommand \doibase [0]{https://doi.org/}%
\providecommand \selectlanguage [0]{\@gobble}%
\providecommand \bibinfo  [0]{\@secondoftwo}%
\providecommand \bibfield  [0]{\@secondoftwo}%
\providecommand \translation [1]{[#1]}%
\providecommand \BibitemOpen [0]{}%
\providecommand \bibitemStop [0]{}%
\providecommand \bibitemNoStop [0]{.\EOS\space}%
\providecommand \EOS [0]{\spacefactor3000\relax}%
\providecommand \BibitemShut  [1]{\csname bibitem#1\endcsname}%
\let\auto@bib@innerbib\@empty
%</preamble>
\bibitem [{\citenamefont {Aharonov}\ and\ \citenamefont
  {{Ben-Or}}(2008)}]{aharonov2008}%
  \BibitemOpen
  \bibfield  {author} {\bibinfo {author} {\bibfnamefont {D.}~\bibnamefont
  {Aharonov}}\ and\ \bibinfo {author} {\bibfnamefont {M.}~\bibnamefont
  {{Ben-Or}}},\ }\bibfield  {title} {\bibinfo {title} {Fault-{{Tolerant Quantum
  Computation}} with {{Constant Error Rate}}},\ }\href
  {https://doi.org/10.1137/S0097539799359385} {\bibfield  {journal} {\bibinfo
  {journal} {SIAM Journal on Computing}\ }\textbf {\bibinfo {volume} {38}},\
  \bibinfo {pages} {1207} (\bibinfo {year} {2008})}\BibitemShut {NoStop}%
\bibitem [{\citenamefont {Knill}\ \emph {et~al.}(1998)\citenamefont {Knill},
  \citenamefont {Laflamme},\ and\ \citenamefont {Zurek}}]{knill1998a}%
  \BibitemOpen
  \bibfield  {author} {\bibinfo {author} {\bibfnamefont {E.}~\bibnamefont
  {Knill}}, \bibinfo {author} {\bibfnamefont {R.}~\bibnamefont {Laflamme}},\
  and\ \bibinfo {author} {\bibfnamefont {W.~H.}\ \bibnamefont {Zurek}},\
  }\bibfield  {title} {\bibinfo {title} {Resilient {{Quantum Computation}}},\
  }\href {https://doi.org/10.1126/science.279.5349.342} {\bibfield  {journal}
  {\bibinfo  {journal} {Science}\ }\textbf {\bibinfo {volume} {279}},\ \bibinfo
  {pages} {342} (\bibinfo {year} {1998})}\BibitemShut {NoStop}%
\bibitem [{\citenamefont {Kitaev}(1997{\natexlab{a}})}]{kitaev1997}%
  \BibitemOpen
  \bibfield  {author} {\bibinfo {author} {\bibfnamefont {A.~Y.}\ \bibnamefont
  {Kitaev}},\ }\bibfield  {title} {\bibinfo {title} {Quantum computations:
  Algorithms and error correction},\ }\href
  {https://doi.org/10.1070/RM1997v052n06ABEH002155} {\bibfield  {journal}
  {\bibinfo  {journal} {Russian Mathematical Surveys}\ }\textbf {\bibinfo
  {volume} {52}},\ \bibinfo {pages} {1191} (\bibinfo {year}
  {1997}{\natexlab{a}})}\BibitemShut {NoStop}%
\bibitem [{\citenamefont {Gottesman}(2014)}]{gottesman2014}%
  \BibitemOpen
  \bibfield  {author} {\bibinfo {author} {\bibfnamefont {D.}~\bibnamefont
  {Gottesman}},\ }\bibfield  {title} {\bibinfo {title} {Fault-tolerant quantum
  computation with constant overhead},\ }\href@noop {} {\bibfield  {journal}
  {\bibinfo  {journal} {Quantum Information \& Computation}\ }\textbf {\bibinfo
  {volume} {14}},\ \bibinfo {pages} {1338} (\bibinfo {year} {2014})},\ \Eprint
  {https://arxiv.org/abs/1310.2984} {arXiv:1310.2984} \BibitemShut {NoStop}%
\bibitem [{\citenamefont {Fawzi}\ \emph {et~al.}(2018)\citenamefont {Fawzi},
  \citenamefont {Grospellier},\ and\ \citenamefont {Leverrier}}]{fawzi2018a}%
  \BibitemOpen
  \bibfield  {author} {\bibinfo {author} {\bibfnamefont {O.}~\bibnamefont
  {Fawzi}}, \bibinfo {author} {\bibfnamefont {A.}~\bibnamefont {Grospellier}},\
  and\ \bibinfo {author} {\bibfnamefont {A.}~\bibnamefont {Leverrier}},\
  }\bibfield  {title} {\bibinfo {title} {Constant {{Overhead Quantum
  Fault}}-{{Tolerance}} with {{Quantum Expander Codes}}},\ }in\ \href
  {https://doi.org/10.1109/FOCS.2018.00076} {\emph {\bibinfo {booktitle} {2018
  {{IEEE}} 59th {{Annual Symposium}} on {{Foundations}} of {{Computer Science}}
  ({{FOCS}})}}}\ (\bibinfo  {publisher} {{IEEE}},\ \bibinfo {address}
  {{Paris}},\ \bibinfo {year} {2018})\ pp.\ \bibinfo {pages}
  {743--754}\BibitemShut {NoStop}%
\bibitem [{\citenamefont {Litinski}(2019{\natexlab{a}})}]{litinski2019}%
  \BibitemOpen
  \bibfield  {author} {\bibinfo {author} {\bibfnamefont {D.}~\bibnamefont
  {Litinski}},\ }\bibfield  {title} {\bibinfo {title} {A {{Game}} of {{Surface
  Codes}}: Large-{{Scale Quantum Computing}} with {{Lattice Surgery}}},\ }\href
  {https://doi.org/10.22331/q-2019-03-05-128} {\bibfield  {journal} {\bibinfo
  {journal} {Quantum}\ }\textbf {\bibinfo {volume} {3}},\ \bibinfo {pages}
  {128} (\bibinfo {year} {2019}{\natexlab{a}})},\ \Eprint
  {https://arxiv.org/abs/1808.02892} {arXiv:1808.02892} \BibitemShut {NoStop}%
\bibitem [{\citenamefont {Fowler}\ and\ \citenamefont
  {Gidney}(2019)}]{fowler2019}%
  \BibitemOpen
  \bibfield  {author} {\bibinfo {author} {\bibfnamefont {A.~G.}\ \bibnamefont
  {Fowler}}\ and\ \bibinfo {author} {\bibfnamefont {C.}~\bibnamefont
  {Gidney}},\ }\bibfield  {title} {\bibinfo {title} {Low overhead quantum
  computation using lattice surgery},\ }\href@noop {} {\bibfield  {journal}
  {\bibinfo  {journal} {arXiv preprint}\ } (\bibinfo {year} {2019})},\ \Eprint
  {https://arxiv.org/abs/1808.06709} {arXiv:1808.06709} \BibitemShut {NoStop}%
\bibitem [{\citenamefont {Chamberland}\ \emph
  {et~al.}(2020{\natexlab{a}})\citenamefont {Chamberland}, \citenamefont {Noh},
  \citenamefont {{Arrangoiz-Arriola}}, \citenamefont {Campbell}, \citenamefont
  {Hann}, \citenamefont {Iverson}, \citenamefont {Putterman}, \citenamefont
  {Bohdanowicz}, \citenamefont {Flammia}, \citenamefont {Keller}, \citenamefont
  {Refael}, \citenamefont {Preskill}, \citenamefont {Jiang}, \citenamefont
  {{Safavi-Naeini}}, \citenamefont {Painter},\ and\ \citenamefont
  {Brand{\~a}o}}]{chamberland2020b}%
  \BibitemOpen
  \bibfield  {author} {\bibinfo {author} {\bibfnamefont {C.}~\bibnamefont
  {Chamberland}}, \bibinfo {author} {\bibfnamefont {K.}~\bibnamefont {Noh}},
  \bibinfo {author} {\bibfnamefont {P.}~\bibnamefont {{Arrangoiz-Arriola}}},
  \bibinfo {author} {\bibfnamefont {E.~T.}\ \bibnamefont {Campbell}}, \bibinfo
  {author} {\bibfnamefont {C.~T.}\ \bibnamefont {Hann}}, \bibinfo {author}
  {\bibfnamefont {J.}~\bibnamefont {Iverson}}, \bibinfo {author} {\bibfnamefont
  {H.}~\bibnamefont {Putterman}}, \bibinfo {author} {\bibfnamefont {T.~C.}\
  \bibnamefont {Bohdanowicz}}, \bibinfo {author} {\bibfnamefont {S.~T.}\
  \bibnamefont {Flammia}}, \bibinfo {author} {\bibfnamefont {A.}~\bibnamefont
  {Keller}}, \bibinfo {author} {\bibfnamefont {G.}~\bibnamefont {Refael}},
  \bibinfo {author} {\bibfnamefont {J.}~\bibnamefont {Preskill}}, \bibinfo
  {author} {\bibfnamefont {L.}~\bibnamefont {Jiang}}, \bibinfo {author}
  {\bibfnamefont {A.~H.}\ \bibnamefont {{Safavi-Naeini}}}, \bibinfo {author}
  {\bibfnamefont {O.}~\bibnamefont {Painter}},\ and\ \bibinfo {author}
  {\bibfnamefont {F.~G. S.~L.}\ \bibnamefont {Brand{\~a}o}},\ }\bibfield
  {title} {\bibinfo {title} {Building a fault-tolerant quantum computer using
  concatenated cat codes},\ }\href@noop {} {\bibfield  {journal} {\bibinfo
  {journal} {arXiv preprint}\ } (\bibinfo {year} {2020}{\natexlab{a}})},\
  \Eprint {https://arxiv.org/abs/2012.04108} {arXiv:2012.04108} \BibitemShut
  {NoStop}%
\bibitem [{\citenamefont {Kim}\ \emph {et~al.}(2021)\citenamefont {Kim},
  \citenamefont {Lee}, \citenamefont {Liu}, \citenamefont {Pallister},
  \citenamefont {Pol},\ and\ \citenamefont {Roberts}}]{kim2021}%
  \BibitemOpen
  \bibfield  {author} {\bibinfo {author} {\bibfnamefont {I.~H.}\ \bibnamefont
  {Kim}}, \bibinfo {author} {\bibfnamefont {E.}~\bibnamefont {Lee}}, \bibinfo
  {author} {\bibfnamefont {Y.-H.}\ \bibnamefont {Liu}}, \bibinfo {author}
  {\bibfnamefont {S.}~\bibnamefont {Pallister}}, \bibinfo {author}
  {\bibfnamefont {W.}~\bibnamefont {Pol}},\ and\ \bibinfo {author}
  {\bibfnamefont {S.}~\bibnamefont {Roberts}},\ }\bibfield  {title} {\bibinfo
  {title} {Fault-tolerant resource estimate for quantum chemical simulations:
  Case study on {{Li}}-ion battery electrolyte molecules},\ }\href@noop {}
  {\bibfield  {journal} {\bibinfo  {journal} {arXiv preprint}\ } (\bibinfo
  {year} {2021})},\ \Eprint {https://arxiv.org/abs/2104.10653}
  {arXiv:2104.10653} \BibitemShut {NoStop}%
\bibitem [{\citenamefont {Beverland}\ \emph {et~al.}(2021)\citenamefont
  {Beverland}, \citenamefont {Kubica},\ and\ \citenamefont
  {Svore}}]{beverland2021a}%
  \BibitemOpen
  \bibfield  {author} {\bibinfo {author} {\bibfnamefont {M.~E.}\ \bibnamefont
  {Beverland}}, \bibinfo {author} {\bibfnamefont {A.}~\bibnamefont {Kubica}},\
  and\ \bibinfo {author} {\bibfnamefont {K.~M.}\ \bibnamefont {Svore}},\
  }\bibfield  {title} {\bibinfo {title} {Cost of {{Universality}}: A
  {{Comparative Study}} of the {{Overhead}} of {{State Distillation}} and
  {{Code Switching}} with {{Color Codes}}},\ }\href
  {https://doi.org/10.1103/PRXQuantum.2.020341} {\bibfield  {journal} {\bibinfo
   {journal} {PRX Quantum}\ }\textbf {\bibinfo {volume} {2}},\ \bibinfo {pages}
  {020341} (\bibinfo {year} {2021})}\BibitemShut {NoStop}%
\bibitem [{\citenamefont {Egan}\ \emph {et~al.}(2021)\citenamefont {Egan},
  \citenamefont {Debroy}, \citenamefont {Noel}, \citenamefont {Risinger},
  \citenamefont {Zhu}, \citenamefont {Biswas}, \citenamefont {Newman},
  \citenamefont {Li}, \citenamefont {Brown}, \citenamefont {Cetina},\ and\
  \citenamefont {Monroe}}]{egan2021}%
  \BibitemOpen
  \bibfield  {author} {\bibinfo {author} {\bibfnamefont {L.}~\bibnamefont
  {Egan}}, \bibinfo {author} {\bibfnamefont {D.~M.}\ \bibnamefont {Debroy}},
  \bibinfo {author} {\bibfnamefont {C.}~\bibnamefont {Noel}}, \bibinfo {author}
  {\bibfnamefont {A.}~\bibnamefont {Risinger}}, \bibinfo {author}
  {\bibfnamefont {D.}~\bibnamefont {Zhu}}, \bibinfo {author} {\bibfnamefont
  {D.}~\bibnamefont {Biswas}}, \bibinfo {author} {\bibfnamefont
  {M.}~\bibnamefont {Newman}}, \bibinfo {author} {\bibfnamefont
  {M.}~\bibnamefont {Li}}, \bibinfo {author} {\bibfnamefont {K.~R.}\
  \bibnamefont {Brown}}, \bibinfo {author} {\bibfnamefont {M.}~\bibnamefont
  {Cetina}},\ and\ \bibinfo {author} {\bibfnamefont {C.}~\bibnamefont
  {Monroe}},\ }\bibfield  {title} {\bibinfo {title} {Fault-{{Tolerant
  Operation}} of a {{Quantum Error}}-{{Correction Code}}},\ }\href@noop {}
  {\bibfield  {journal} {\bibinfo  {journal} {arXiv preprint}\ } (\bibinfo
  {year} {2021})},\ \Eprint {https://arxiv.org/abs/2009.11482}
  {arXiv:2009.11482} \BibitemShut {NoStop}%
\bibitem [{\citenamefont {Andersen}\ \emph {et~al.}(2020)\citenamefont
  {Andersen}, \citenamefont {Remm}, \citenamefont {Lazar}, \citenamefont
  {Krinner}, \citenamefont {Lacroix}, \citenamefont {Norris}, \citenamefont
  {Gabureac}, \citenamefont {Eichler},\ and\ \citenamefont
  {Wallraff}}]{andersen2020}%
  \BibitemOpen
  \bibfield  {author} {\bibinfo {author} {\bibfnamefont {C.~K.}\ \bibnamefont
  {Andersen}}, \bibinfo {author} {\bibfnamefont {A.}~\bibnamefont {Remm}},
  \bibinfo {author} {\bibfnamefont {S.}~\bibnamefont {Lazar}}, \bibinfo
  {author} {\bibfnamefont {S.}~\bibnamefont {Krinner}}, \bibinfo {author}
  {\bibfnamefont {N.}~\bibnamefont {Lacroix}}, \bibinfo {author} {\bibfnamefont
  {G.~J.}\ \bibnamefont {Norris}}, \bibinfo {author} {\bibfnamefont
  {M.}~\bibnamefont {Gabureac}}, \bibinfo {author} {\bibfnamefont
  {C.}~\bibnamefont {Eichler}},\ and\ \bibinfo {author} {\bibfnamefont
  {A.}~\bibnamefont {Wallraff}},\ }\bibfield  {title} {\bibinfo {title}
  {Repeated quantum error detection in a surface code},\ }\href
  {https://doi.org/10.1038/s41567-020-0920-y} {\bibfield  {journal} {\bibinfo
  {journal} {Nature Physics}\ }\textbf {\bibinfo {volume} {16}},\ \bibinfo
  {pages} {875} (\bibinfo {year} {2020})}\BibitemShut {NoStop}%
\bibitem [{\citenamefont {Chen}\ \emph
  {et~al.}(2021{\natexlab{a}})\citenamefont {Chen}, \citenamefont {Satzinger},
  \citenamefont {Atalaya}, \citenamefont {Korotkov}, \citenamefont {Dunsworth},
  \citenamefont {Sank}, \citenamefont {Quintana}, \citenamefont {McEwen},
  \citenamefont {Barends}, \citenamefont {Klimov}, \citenamefont {Hong},
  \citenamefont {Jones}, \citenamefont {Petukhov}, \citenamefont {Kafri},
  \citenamefont {Demura}, \citenamefont {Burkett}, \citenamefont {Gidney},
  \citenamefont {Fowler}, \citenamefont {Paler}, \citenamefont {Putterman},
  \citenamefont {Aleiner}, \citenamefont {Arute}, \citenamefont {Arya},
  \citenamefont {Babbush}, \citenamefont {Bardin}, \citenamefont {Bengtsson},
  \citenamefont {Bourassa}, \citenamefont {Broughton}, \citenamefont {Buckley},
  \citenamefont {Buell}, \citenamefont {Bushnell}, \citenamefont {Chiaro},
  \citenamefont {Collins}, \citenamefont {Courtney}, \citenamefont {Derk},
  \citenamefont {Eppens}, \citenamefont {Erickson}, \citenamefont {Farhi},
  \citenamefont {Foxen}, \citenamefont {Giustina}, \citenamefont {Greene},
  \citenamefont {Gross}, \citenamefont {Harrigan}, \citenamefont {Harrington},
  \citenamefont {Hilton}, \citenamefont {Ho}, \citenamefont {Huang},
  \citenamefont {Huggins}, \citenamefont {Ioffe}, \citenamefont {Isakov},
  \citenamefont {Jeffrey}, \citenamefont {Jiang}, \citenamefont {Kechedzhi},
  \citenamefont {Kim}, \citenamefont {Kitaev}, \citenamefont {Kostritsa},
  \citenamefont {Landhuis}, \citenamefont {Laptev}, \citenamefont {Lucero},
  \citenamefont {Martin}, \citenamefont {McClean}, \citenamefont {McCourt},
  \citenamefont {Mi}, \citenamefont {Miao}, \citenamefont {Mohseni},
  \citenamefont {Montazeri}, \citenamefont {Mruczkiewicz}, \citenamefont
  {Mutus}, \citenamefont {Naaman}, \citenamefont {Neeley}, \citenamefont
  {Neill}, \citenamefont {Newman}, \citenamefont {Niu}, \citenamefont
  {O'Brien}, \citenamefont {Opremcak}, \citenamefont {Ostby}, \citenamefont
  {Pat{\'o}}, \citenamefont {Redd}, \citenamefont {Roushan}, \citenamefont
  {Rubin}, \citenamefont {Shvarts}, \citenamefont {Strain}, \citenamefont
  {Szalay}, \citenamefont {Trevithick}, \citenamefont {Villalonga},
  \citenamefont {White}, \citenamefont {Yao}, \citenamefont {Yeh},
  \citenamefont {Yoo}, \citenamefont {Zalcman}, \citenamefont {Neven},
  \citenamefont {Boixo}, \citenamefont {Smelyanskiy}, \citenamefont {Chen},
  \citenamefont {Megrant}, \citenamefont {Kelly},\ and\ \citenamefont {{Google
  Quantum AI}}}]{chen2021}%
  \BibitemOpen
  \bibfield  {author} {\bibinfo {author} {\bibfnamefont {Z.}~\bibnamefont
  {Chen}}, \bibinfo {author} {\bibfnamefont {K.~J.}\ \bibnamefont {Satzinger}},
  \bibinfo {author} {\bibfnamefont {J.}~\bibnamefont {Atalaya}}, \bibinfo
  {author} {\bibfnamefont {A.~N.}\ \bibnamefont {Korotkov}}, \bibinfo {author}
  {\bibfnamefont {A.}~\bibnamefont {Dunsworth}}, \bibinfo {author}
  {\bibfnamefont {D.}~\bibnamefont {Sank}}, \bibinfo {author} {\bibfnamefont
  {C.}~\bibnamefont {Quintana}}, \bibinfo {author} {\bibfnamefont
  {M.}~\bibnamefont {McEwen}}, \bibinfo {author} {\bibfnamefont
  {R.}~\bibnamefont {Barends}}, \bibinfo {author} {\bibfnamefont {P.~V.}\
  \bibnamefont {Klimov}}, \bibinfo {author} {\bibfnamefont {S.}~\bibnamefont
  {Hong}}, \bibinfo {author} {\bibfnamefont {C.}~\bibnamefont {Jones}},
  \bibinfo {author} {\bibfnamefont {A.}~\bibnamefont {Petukhov}}, \bibinfo
  {author} {\bibfnamefont {D.}~\bibnamefont {Kafri}}, \bibinfo {author}
  {\bibfnamefont {S.}~\bibnamefont {Demura}}, \bibinfo {author} {\bibfnamefont
  {B.}~\bibnamefont {Burkett}}, \bibinfo {author} {\bibfnamefont
  {C.}~\bibnamefont {Gidney}}, \bibinfo {author} {\bibfnamefont {A.~G.}\
  \bibnamefont {Fowler}}, \bibinfo {author} {\bibfnamefont {A.}~\bibnamefont
  {Paler}}, \bibinfo {author} {\bibfnamefont {H.}~\bibnamefont {Putterman}},
  \bibinfo {author} {\bibfnamefont {I.}~\bibnamefont {Aleiner}}, \bibinfo
  {author} {\bibfnamefont {F.}~\bibnamefont {Arute}}, \bibinfo {author}
  {\bibfnamefont {K.}~\bibnamefont {Arya}}, \bibinfo {author} {\bibfnamefont
  {R.}~\bibnamefont {Babbush}}, \bibinfo {author} {\bibfnamefont {J.~C.}\
  \bibnamefont {Bardin}}, \bibinfo {author} {\bibfnamefont {A.}~\bibnamefont
  {Bengtsson}}, \bibinfo {author} {\bibfnamefont {A.}~\bibnamefont {Bourassa}},
  \bibinfo {author} {\bibfnamefont {M.}~\bibnamefont {Broughton}}, \bibinfo
  {author} {\bibfnamefont {B.~B.}\ \bibnamefont {Buckley}}, \bibinfo {author}
  {\bibfnamefont {D.~A.}\ \bibnamefont {Buell}}, \bibinfo {author}
  {\bibfnamefont {N.}~\bibnamefont {Bushnell}}, \bibinfo {author}
  {\bibfnamefont {B.}~\bibnamefont {Chiaro}}, \bibinfo {author} {\bibfnamefont
  {R.}~\bibnamefont {Collins}}, \bibinfo {author} {\bibfnamefont
  {W.}~\bibnamefont {Courtney}}, \bibinfo {author} {\bibfnamefont {A.~R.}\
  \bibnamefont {Derk}}, \bibinfo {author} {\bibfnamefont {D.}~\bibnamefont
  {Eppens}}, \bibinfo {author} {\bibfnamefont {C.}~\bibnamefont {Erickson}},
  \bibinfo {author} {\bibfnamefont {E.}~\bibnamefont {Farhi}}, \bibinfo
  {author} {\bibfnamefont {B.}~\bibnamefont {Foxen}}, \bibinfo {author}
  {\bibfnamefont {M.}~\bibnamefont {Giustina}}, \bibinfo {author}
  {\bibfnamefont {A.}~\bibnamefont {Greene}}, \bibinfo {author} {\bibfnamefont
  {J.~A.}\ \bibnamefont {Gross}}, \bibinfo {author} {\bibfnamefont {M.~P.}\
  \bibnamefont {Harrigan}}, \bibinfo {author} {\bibfnamefont {S.~D.}\
  \bibnamefont {Harrington}}, \bibinfo {author} {\bibfnamefont
  {J.}~\bibnamefont {Hilton}}, \bibinfo {author} {\bibfnamefont
  {A.}~\bibnamefont {Ho}}, \bibinfo {author} {\bibfnamefont {T.}~\bibnamefont
  {Huang}}, \bibinfo {author} {\bibfnamefont {W.~J.}\ \bibnamefont {Huggins}},
  \bibinfo {author} {\bibfnamefont {L.~B.}\ \bibnamefont {Ioffe}}, \bibinfo
  {author} {\bibfnamefont {S.~V.}\ \bibnamefont {Isakov}}, \bibinfo {author}
  {\bibfnamefont {E.}~\bibnamefont {Jeffrey}}, \bibinfo {author} {\bibfnamefont
  {Z.}~\bibnamefont {Jiang}}, \bibinfo {author} {\bibfnamefont
  {K.}~\bibnamefont {Kechedzhi}}, \bibinfo {author} {\bibfnamefont
  {S.}~\bibnamefont {Kim}}, \bibinfo {author} {\bibfnamefont {A.}~\bibnamefont
  {Kitaev}}, \bibinfo {author} {\bibfnamefont {F.}~\bibnamefont {Kostritsa}},
  \bibinfo {author} {\bibfnamefont {D.}~\bibnamefont {Landhuis}}, \bibinfo
  {author} {\bibfnamefont {P.}~\bibnamefont {Laptev}}, \bibinfo {author}
  {\bibfnamefont {E.}~\bibnamefont {Lucero}}, \bibinfo {author} {\bibfnamefont
  {O.}~\bibnamefont {Martin}}, \bibinfo {author} {\bibfnamefont {J.~R.}\
  \bibnamefont {McClean}}, \bibinfo {author} {\bibfnamefont {T.}~\bibnamefont
  {McCourt}}, \bibinfo {author} {\bibfnamefont {X.}~\bibnamefont {Mi}},
  \bibinfo {author} {\bibfnamefont {K.~C.}\ \bibnamefont {Miao}}, \bibinfo
  {author} {\bibfnamefont {M.}~\bibnamefont {Mohseni}}, \bibinfo {author}
  {\bibfnamefont {S.}~\bibnamefont {Montazeri}}, \bibinfo {author}
  {\bibfnamefont {W.}~\bibnamefont {Mruczkiewicz}}, \bibinfo {author}
  {\bibfnamefont {J.}~\bibnamefont {Mutus}}, \bibinfo {author} {\bibfnamefont
  {O.}~\bibnamefont {Naaman}}, \bibinfo {author} {\bibfnamefont
  {M.}~\bibnamefont {Neeley}}, \bibinfo {author} {\bibfnamefont
  {C.}~\bibnamefont {Neill}}, \bibinfo {author} {\bibfnamefont
  {M.}~\bibnamefont {Newman}}, \bibinfo {author} {\bibfnamefont {M.~Y.}\
  \bibnamefont {Niu}}, \bibinfo {author} {\bibfnamefont {T.~E.}\ \bibnamefont
  {O'Brien}}, \bibinfo {author} {\bibfnamefont {A.}~\bibnamefont {Opremcak}},
  \bibinfo {author} {\bibfnamefont {E.}~\bibnamefont {Ostby}}, \bibinfo
  {author} {\bibfnamefont {B.}~\bibnamefont {Pat{\'o}}}, \bibinfo {author}
  {\bibfnamefont {N.}~\bibnamefont {Redd}}, \bibinfo {author} {\bibfnamefont
  {P.}~\bibnamefont {Roushan}}, \bibinfo {author} {\bibfnamefont {N.~C.}\
  \bibnamefont {Rubin}}, \bibinfo {author} {\bibfnamefont {V.}~\bibnamefont
  {Shvarts}}, \bibinfo {author} {\bibfnamefont {D.}~\bibnamefont {Strain}},
  \bibinfo {author} {\bibfnamefont {M.}~\bibnamefont {Szalay}}, \bibinfo
  {author} {\bibfnamefont {M.~D.}\ \bibnamefont {Trevithick}}, \bibinfo
  {author} {\bibfnamefont {B.}~\bibnamefont {Villalonga}}, \bibinfo {author}
  {\bibfnamefont {T.}~\bibnamefont {White}}, \bibinfo {author} {\bibfnamefont
  {Z.~J.}\ \bibnamefont {Yao}}, \bibinfo {author} {\bibfnamefont
  {P.}~\bibnamefont {Yeh}}, \bibinfo {author} {\bibfnamefont {J.}~\bibnamefont
  {Yoo}}, \bibinfo {author} {\bibfnamefont {A.}~\bibnamefont {Zalcman}},
  \bibinfo {author} {\bibfnamefont {H.}~\bibnamefont {Neven}}, \bibinfo
  {author} {\bibfnamefont {S.}~\bibnamefont {Boixo}}, \bibinfo {author}
  {\bibfnamefont {V.}~\bibnamefont {Smelyanskiy}}, \bibinfo {author}
  {\bibfnamefont {Y.}~\bibnamefont {Chen}}, \bibinfo {author} {\bibfnamefont
  {A.}~\bibnamefont {Megrant}}, \bibinfo {author} {\bibfnamefont
  {J.}~\bibnamefont {Kelly}},\ and\ \bibinfo {author} {\bibnamefont {{Google
  Quantum AI}}},\ }\bibfield  {title} {\bibinfo {title} {Exponential
  suppression of bit or phase errors with cyclic error correction},\ }\href
  {https://doi.org/10.1038/s41586-021-03588-y} {\bibfield  {journal} {\bibinfo
  {journal} {Nature}\ }\textbf {\bibinfo {volume} {595}},\ \bibinfo {pages}
  {383} (\bibinfo {year} {2021}{\natexlab{a}})}\BibitemShut {NoStop}%
\bibitem [{\citenamefont {Chen}\ \emph
  {et~al.}(2021{\natexlab{b}})\citenamefont {Chen}, \citenamefont {Yoder},
  \citenamefont {Kim}, \citenamefont {Sundaresan}, \citenamefont {Srinivasan},
  \citenamefont {Li}, \citenamefont {C{\'o}rcoles}, \citenamefont {Cross},\
  and\ \citenamefont {Takita}}]{chen2021a}%
  \BibitemOpen
  \bibfield  {author} {\bibinfo {author} {\bibfnamefont {E.~H.}\ \bibnamefont
  {Chen}}, \bibinfo {author} {\bibfnamefont {T.~J.}\ \bibnamefont {Yoder}},
  \bibinfo {author} {\bibfnamefont {Y.}~\bibnamefont {Kim}}, \bibinfo {author}
  {\bibfnamefont {N.}~\bibnamefont {Sundaresan}}, \bibinfo {author}
  {\bibfnamefont {S.}~\bibnamefont {Srinivasan}}, \bibinfo {author}
  {\bibfnamefont {M.}~\bibnamefont {Li}}, \bibinfo {author} {\bibfnamefont
  {A.~D.}\ \bibnamefont {C{\'o}rcoles}}, \bibinfo {author} {\bibfnamefont
  {A.~W.}\ \bibnamefont {Cross}},\ and\ \bibinfo {author} {\bibfnamefont
  {M.}~\bibnamefont {Takita}},\ }\bibfield  {title} {\bibinfo {title}
  {Calibrated decoders for experimental quantum error correction},\ }\href@noop
  {} {\bibfield  {journal} {\bibinfo  {journal} {arXiv preprint}\ } (\bibinfo
  {year} {2021}{\natexlab{b}})},\ \Eprint {https://arxiv.org/abs/2110.04285}
  {arXiv:2110.04285} \BibitemShut {NoStop}%
\bibitem [{\citenamefont {Erhard}\ \emph {et~al.}(2021)\citenamefont {Erhard},
  \citenamefont {Poulsen~Nautrup}, \citenamefont {Meth}, \citenamefont
  {Postler}, \citenamefont {Stricker}, \citenamefont {Stadler}, \citenamefont
  {Negnevitsky}, \citenamefont {Ringbauer}, \citenamefont {Schindler},
  \citenamefont {Briegel}, \citenamefont {Blatt}, \citenamefont {Friis},\ and\
  \citenamefont {Monz}}]{erhard2021}%
  \BibitemOpen
  \bibfield  {author} {\bibinfo {author} {\bibfnamefont {A.}~\bibnamefont
  {Erhard}}, \bibinfo {author} {\bibfnamefont {H.}~\bibnamefont
  {Poulsen~Nautrup}}, \bibinfo {author} {\bibfnamefont {M.}~\bibnamefont
  {Meth}}, \bibinfo {author} {\bibfnamefont {L.}~\bibnamefont {Postler}},
  \bibinfo {author} {\bibfnamefont {R.}~\bibnamefont {Stricker}}, \bibinfo
  {author} {\bibfnamefont {M.}~\bibnamefont {Stadler}}, \bibinfo {author}
  {\bibfnamefont {V.}~\bibnamefont {Negnevitsky}}, \bibinfo {author}
  {\bibfnamefont {M.}~\bibnamefont {Ringbauer}}, \bibinfo {author}
  {\bibfnamefont {P.}~\bibnamefont {Schindler}}, \bibinfo {author}
  {\bibfnamefont {H.~J.}\ \bibnamefont {Briegel}}, \bibinfo {author}
  {\bibfnamefont {R.}~\bibnamefont {Blatt}}, \bibinfo {author} {\bibfnamefont
  {N.}~\bibnamefont {Friis}},\ and\ \bibinfo {author} {\bibfnamefont
  {T.}~\bibnamefont {Monz}},\ }\bibfield  {title} {\bibinfo {title} {Entangling
  logical qubits with lattice surgery},\ }\href
  {https://doi.org/10.1038/s41586-020-03079-6} {\bibfield  {journal} {\bibinfo
  {journal} {Nature}\ }\textbf {\bibinfo {volume} {589}},\ \bibinfo {pages}
  {220} (\bibinfo {year} {2021})}\BibitemShut {NoStop}%
\bibitem [{\citenamefont {Bravyi}\ and\ \citenamefont
  {Kitaev}(2005)}]{bravyi2005}%
  \BibitemOpen
  \bibfield  {author} {\bibinfo {author} {\bibfnamefont {S.}~\bibnamefont
  {Bravyi}}\ and\ \bibinfo {author} {\bibfnamefont {A.}~\bibnamefont
  {Kitaev}},\ }\bibfield  {title} {\bibinfo {title} {Universal quantum
  computation with ideal {{Clifford}} gates and noisy ancillas},\ }\href
  {https://doi.org/10.1103/PhysRevA.71.022316} {\bibfield  {journal} {\bibinfo
  {journal} {Physical Review A}\ }\textbf {\bibinfo {volume} {71}},\ \bibinfo
  {pages} {022316} (\bibinfo {year} {2005})}\BibitemShut {NoStop}%
\bibitem [{\citenamefont {Meier}\ \emph {et~al.}(2013)\citenamefont {Meier},
  \citenamefont {Eastin},\ and\ \citenamefont {Knill}}]{meier2012}%
  \BibitemOpen
  \bibfield  {author} {\bibinfo {author} {\bibfnamefont {A.~M.}\ \bibnamefont
  {Meier}}, \bibinfo {author} {\bibfnamefont {B.}~\bibnamefont {Eastin}},\ and\
  \bibinfo {author} {\bibfnamefont {E.}~\bibnamefont {Knill}},\ }\bibfield
  {title} {\bibinfo {title} {Magic-state distillation with the four-qubit
  code},\ }\href@noop {} {\bibfield  {journal} {\bibinfo  {journal} {Quantum
  Information \& Computation}\ }\textbf {\bibinfo {volume} {13}},\ \bibinfo
  {pages} {195–209} (\bibinfo {year} {2013})},\ \Eprint
  {https://arxiv.org/abs/1204.4221} {arXiv:1204.4221} \BibitemShut {NoStop}%
\bibitem [{\citenamefont {Bomb{\'i}n}\ and\ \citenamefont
  {{Martin-Delgado}}(2006)}]{bombin2006}%
  \BibitemOpen
  \bibfield  {author} {\bibinfo {author} {\bibfnamefont {H.}~\bibnamefont
  {Bomb{\'i}n}}\ and\ \bibinfo {author} {\bibfnamefont {M.~A.}\ \bibnamefont
  {{Martin-Delgado}}},\ }\bibfield  {title} {\bibinfo {title} {Topological
  {{Quantum Distillation}}},\ }\href
  {https://doi.org/10.1103/PhysRevLett.97.180501} {\bibfield  {journal}
  {\bibinfo  {journal} {Physical Review Letters}\ }\textbf {\bibinfo {volume}
  {97}},\ \bibinfo {pages} {180501} (\bibinfo {year} {2006})}\BibitemShut
  {NoStop}%
\bibitem [{\citenamefont {Bomb{\'i}n}\ and\ \citenamefont
  {{Martin-Delgado}}(2007{\natexlab{a}})}]{bombin2007a}%
  \BibitemOpen
  \bibfield  {author} {\bibinfo {author} {\bibfnamefont {H.}~\bibnamefont
  {Bomb{\'i}n}}\ and\ \bibinfo {author} {\bibfnamefont {M.~A.}\ \bibnamefont
  {{Martin-Delgado}}},\ }\bibfield  {title} {\bibinfo {title} {Topological
  {{Computation}} without {{Braiding}}},\ }\href
  {https://doi.org/10.1103/PhysRevLett.98.160502} {\bibfield  {journal}
  {\bibinfo  {journal} {Physical Review Letters}\ }\textbf {\bibinfo {volume}
  {98}},\ \bibinfo {pages} {160502} (\bibinfo {year}
  {2007}{\natexlab{a}})}\BibitemShut {NoStop}%
\bibitem [{\citenamefont {Kubica}(2018)}]{kubica2018a}%
  \BibitemOpen
  \bibfield  {author} {\bibinfo {author} {\bibfnamefont {A.~M.}\ \bibnamefont
  {Kubica}},\ }\emph {\bibinfo {title} {The ABCs of the Color Code: A Study of
  Topological Quantum Codes as Toy Models for Fault-Tolerant Quantum
  Computation and Quantum Phases of Matter}},\ \href
  {https://doi.org/10.7907/059V-MG69} {Ph.D. thesis},\ \bibinfo  {school}
  {Caltech} (\bibinfo {year} {2018})\BibitemShut {NoStop}%
\bibitem [{\citenamefont {Kubica}\ \emph {et~al.}(2015)\citenamefont {Kubica},
  \citenamefont {Yoshida},\ and\ \citenamefont {Pastawski}}]{kubica2015a}%
  \BibitemOpen
  \bibfield  {author} {\bibinfo {author} {\bibfnamefont {A.}~\bibnamefont
  {Kubica}}, \bibinfo {author} {\bibfnamefont {B.}~\bibnamefont {Yoshida}},\
  and\ \bibinfo {author} {\bibfnamefont {F.}~\bibnamefont {Pastawski}},\
  }\bibfield  {title} {\bibinfo {title} {Unfolding the color code},\ }\href
  {https://doi.org/10.1088/1367-2630/17/8/083026} {\bibfield  {journal}
  {\bibinfo  {journal} {New Journal of Physics}\ }\textbf {\bibinfo {volume}
  {17}},\ \bibinfo {pages} {083026} (\bibinfo {year} {2015})}\BibitemShut
  {NoStop}%
\bibitem [{\citenamefont {Vasmer}\ and\ \citenamefont
  {Browne}(2019)}]{vasmer2019}%
  \BibitemOpen
  \bibfield  {author} {\bibinfo {author} {\bibfnamefont {M.}~\bibnamefont
  {Vasmer}}\ and\ \bibinfo {author} {\bibfnamefont {D.~E.}\ \bibnamefont
  {Browne}},\ }\bibfield  {title} {\bibinfo {title} {Three-dimensional surface
  codes: Transversal gates and fault-tolerant architectures},\ }\href
  {https://doi.org/10.1103/PhysRevA.100.012312} {\bibfield  {journal} {\bibinfo
   {journal} {Physical Review A}\ }\textbf {\bibinfo {volume} {100}},\ \bibinfo
  {pages} {012312} (\bibinfo {year} {2019})}\BibitemShut {NoStop}%
\bibitem [{\citenamefont {{Jochym-O'Connor}}\ and\ \citenamefont
  {Yoder}(2021)}]{jochym-oconnor2021}%
  \BibitemOpen
  \bibfield  {author} {\bibinfo {author} {\bibfnamefont {T.}~\bibnamefont
  {{Jochym-O'Connor}}}\ and\ \bibinfo {author} {\bibfnamefont {T.~J.}\
  \bibnamefont {Yoder}},\ }\bibfield  {title} {\bibinfo {title} {A
  four-dimensional toric code with non-{{Clifford}} transversal gates},\ }\href
  {https://doi.org/10.1103/PhysRevResearch.3.013118} {\bibfield  {journal}
  {\bibinfo  {journal} {Physical Review Research}\ }\textbf {\bibinfo {volume}
  {3}},\ \bibinfo {pages} {013118} (\bibinfo {year} {2021})},\ \Eprint
  {https://arxiv.org/abs/2010.02238} {arXiv:2010.02238} \BibitemShut {NoStop}%
\bibitem [{\citenamefont {Dennis}\ \emph {et~al.}(2002)\citenamefont {Dennis},
  \citenamefont {Kitaev}, \citenamefont {Landahl},\ and\ \citenamefont
  {Preskill}}]{dennis2002}%
  \BibitemOpen
  \bibfield  {author} {\bibinfo {author} {\bibfnamefont {E.}~\bibnamefont
  {Dennis}}, \bibinfo {author} {\bibfnamefont {A.}~\bibnamefont {Kitaev}},
  \bibinfo {author} {\bibfnamefont {A.}~\bibnamefont {Landahl}},\ and\ \bibinfo
  {author} {\bibfnamefont {J.}~\bibnamefont {Preskill}},\ }\bibfield  {title}
  {\bibinfo {title} {Topological quantum memory},\ }\href
  {https://doi.org/10.1063/1.1499754} {\bibfield  {journal} {\bibinfo
  {journal} {Journal of Mathematical Physics}\ }\textbf {\bibinfo {volume}
  {43}},\ \bibinfo {pages} {4452} (\bibinfo {year} {2002})}\BibitemShut
  {NoStop}%
\bibitem [{\citenamefont {Fowler}(2015)}]{fowler2015}%
  \BibitemOpen
  \bibfield  {author} {\bibinfo {author} {\bibfnamefont {A.~G.}\ \bibnamefont
  {Fowler}},\ }\bibfield  {title} {\bibinfo {title} {Minimum weight perfect
  matching of fault-tolerant topological quantum error correction in average
  {{O}}(1) parallel time},\ }\href@noop {} {\bibfield  {journal} {\bibinfo
  {journal} {Quantum Information \& Computation}\ }\textbf {\bibinfo {volume}
  {15}},\ \bibinfo {pages} {145} (\bibinfo {year} {2015})},\ \Eprint
  {https://arxiv.org/abs/1307.1740} {arXiv:1307.1740} \BibitemShut {NoStop}%
\bibitem [{\citenamefont {Kubica}\ and\ \citenamefont
  {Delfosse}(2019)}]{kubica2019}%
  \BibitemOpen
  \bibfield  {author} {\bibinfo {author} {\bibfnamefont {A.}~\bibnamefont
  {Kubica}}\ and\ \bibinfo {author} {\bibfnamefont {N.}~\bibnamefont
  {Delfosse}},\ }\bibfield  {title} {\bibinfo {title} {Efficient color code
  decoders in \$d\textbackslash geq 2\$ dimensions from toric code decoders},\
  }\href@noop {} {\bibfield  {journal} {\bibinfo  {journal} {arXiv preprint}\ }
  (\bibinfo {year} {2019})},\ \Eprint {https://arxiv.org/abs/1905.07393}
  {arXiv:1905.07393} \BibitemShut {NoStop}%
\bibitem [{\citenamefont {Gottesman}(1997)}]{gottesman1997}%
  \BibitemOpen
  \bibfield  {author} {\bibinfo {author} {\bibfnamefont {D.}~\bibnamefont
  {Gottesman}},\ }\bibfield  {title} {\bibinfo {title} {Stabilizer {{Codes}}
  and {{Quantum Error Correction}}},\ }\href@noop {} {\bibfield  {journal}
  {\bibinfo  {journal} {arXiv preprint}\ } (\bibinfo {year} {1997})},\ \Eprint
  {https://arxiv.org/abs/quant-ph/9705052} {arXiv:quant-ph/9705052}
  \BibitemShut {NoStop}%
\bibitem [{Note1()}]{Note1}%
  \BibitemOpen
  \bibinfo {note} {We observe that morphing is similar to the disentangling
  step of entanglement renormalization~\cite {vidal2007}.}\BibitemShut {Stop}%
\bibitem [{\citenamefont {Steane}(1996)}]{steane1996a}%
  \BibitemOpen
  \bibfield  {author} {\bibinfo {author} {\bibfnamefont {A.~M.}\ \bibnamefont
  {Steane}},\ }\bibfield  {title} {\bibinfo {title} {Multiple-particle
  interference and quantum error correction},\ }\href
  {https://doi.org/10.1098/rspa.1996.0136} {\bibfield  {journal} {\bibinfo
  {journal} {Proceedings of the Royal Society of London. Series A:
  Mathematical, Physical and Engineering Sciences}\ }\textbf {\bibinfo {volume}
  {452}},\ \bibinfo {pages} {2551} (\bibinfo {year} {1996})}\BibitemShut
  {NoStop}%
\bibitem [{\citenamefont {Knill}(2005)}]{knill2005}%
  \BibitemOpen
  \bibfield  {author} {\bibinfo {author} {\bibfnamefont {E.}~\bibnamefont
  {Knill}},\ }\bibfield  {title} {\bibinfo {title} {Quantum computing with
  realistically noisy devices},\ }\href {https://doi.org/10.1038/nature03350}
  {\bibfield  {journal} {\bibinfo  {journal} {Nature}\ }\textbf {\bibinfo
  {volume} {434}},\ \bibinfo {pages} {39} (\bibinfo {year} {2005})}\BibitemShut
  {NoStop}%
\bibitem [{\citenamefont {Knill}(2004)}]{knill2004}%
  \BibitemOpen
  \bibfield  {author} {\bibinfo {author} {\bibfnamefont {E.}~\bibnamefont
  {Knill}},\ }\bibfield  {title} {\bibinfo {title} {Fault-{{Tolerant
  Postselected Quantum Computation}}: Schemes},\ }\href@noop {} {\bibfield
  {journal} {\bibinfo  {journal} {arXiv preprint}\ } (\bibinfo {year}
  {2004})},\ \Eprint {https://arxiv.org/abs/quant-ph/0402171}
  {arXiv:quant-ph/0402171} \BibitemShut {NoStop}%
\bibitem [{\citenamefont {Bravyi}\ and\ \citenamefont
  {Haah}(2012)}]{bravyi2012}%
  \BibitemOpen
  \bibfield  {author} {\bibinfo {author} {\bibfnamefont {S.}~\bibnamefont
  {Bravyi}}\ and\ \bibinfo {author} {\bibfnamefont {J.}~\bibnamefont {Haah}},\
  }\bibfield  {title} {\bibinfo {title} {Magic-state distillation with low
  overhead},\ }\href {https://doi.org/10.1103/PhysRevA.86.052329} {\bibfield
  {journal} {\bibinfo  {journal} {Physical Review A}\ }\textbf {\bibinfo
  {volume} {86}},\ \bibinfo {pages} {052329} (\bibinfo {year}
  {2012})}\BibitemShut {NoStop}%
\bibitem [{\citenamefont {Campbell}\ and\ \citenamefont
  {Howard}(2017{\natexlab{a}})}]{campbell2017a}%
  \BibitemOpen
  \bibfield  {author} {\bibinfo {author} {\bibfnamefont {E.~T.}\ \bibnamefont
  {Campbell}}\ and\ \bibinfo {author} {\bibfnamefont {M.}~\bibnamefont
  {Howard}},\ }\bibfield  {title} {\bibinfo {title} {Unifying {{Gate
  Synthesis}} and {{Magic State Distillation}}},\ }\href
  {https://doi.org/10.1103/PhysRevLett.118.060501} {\bibfield  {journal}
  {\bibinfo  {journal} {Physical Review Letters}\ }\textbf {\bibinfo {volume}
  {118}},\ \bibinfo {pages} {060501} (\bibinfo {year}
  {2017}{\natexlab{a}})}\BibitemShut {NoStop}%
\bibitem [{\citenamefont {Hastings}\ and\ \citenamefont
  {Haah}(2018)}]{hastings2018}%
  \BibitemOpen
  \bibfield  {author} {\bibinfo {author} {\bibfnamefont {M.~B.}\ \bibnamefont
  {Hastings}}\ and\ \bibinfo {author} {\bibfnamefont {J.}~\bibnamefont
  {Haah}},\ }\bibfield  {title} {\bibinfo {title} {Distillation with
  {{Sublogarithmic Overhead}}},\ }\href
  {https://doi.org/10.1103/PhysRevLett.120.050504} {\bibfield  {journal}
  {\bibinfo  {journal} {Physical Review Letters}\ }\textbf {\bibinfo {volume}
  {120}},\ \bibinfo {pages} {050504} (\bibinfo {year} {2018})}\BibitemShut
  {NoStop}%
\bibitem [{\citenamefont {Haah}\ \emph {et~al.}(2017)\citenamefont {Haah},
  \citenamefont {Hastings}, \citenamefont {Poulin},\ and\ \citenamefont
  {Wecker}}]{haah2017}%
  \BibitemOpen
  \bibfield  {author} {\bibinfo {author} {\bibfnamefont {J.}~\bibnamefont
  {Haah}}, \bibinfo {author} {\bibfnamefont {M.~B.}\ \bibnamefont {Hastings}},
  \bibinfo {author} {\bibfnamefont {D.}~\bibnamefont {Poulin}},\ and\ \bibinfo
  {author} {\bibfnamefont {D.}~\bibnamefont {Wecker}},\ }\bibfield  {title}
  {\bibinfo {title} {Magic {{State Distillation}} at {{Intermediate Size}}},\
  }\bibfield  {journal} {\bibinfo  {journal} {Quantum Information and
  Computation}\ }\textbf {\bibinfo {volume} {18}},\ \href
  {https://doi.org/10.26421/QIC18.1-2} {10.26421/QIC18.1-2} (\bibinfo {year}
  {2017}),\ \Eprint {https://arxiv.org/abs/1709.02789} {arXiv:1709.02789}
  \BibitemShut {NoStop}%
\bibitem [{\citenamefont {Haah}\ and\ \citenamefont
  {Hastings}(2018)}]{haah2018a}%
  \BibitemOpen
  \bibfield  {author} {\bibinfo {author} {\bibfnamefont {J.}~\bibnamefont
  {Haah}}\ and\ \bibinfo {author} {\bibfnamefont {M.~B.}\ \bibnamefont
  {Hastings}},\ }\bibfield  {title} {\bibinfo {title} {Codes and {{Protocols}}
  for {{Distilling}} \${{T}}\$, controlled-\${{S}}\$, and {{Toffoli Gates}}},\
  }\href {https://doi.org/10.22331/q-2018-06-07-71} {\bibfield  {journal}
  {\bibinfo  {journal} {Quantum}\ }\textbf {\bibinfo {volume} {2}},\ \bibinfo
  {pages} {71} (\bibinfo {year} {2018})},\ \Eprint
  {https://arxiv.org/abs/1709.02832} {arXiv:1709.02832} \BibitemShut {NoStop}%
\bibitem [{\citenamefont {Steane}(1999)}]{steane1999}%
  \BibitemOpen
  \bibfield  {author} {\bibinfo {author} {\bibfnamefont {A.}~\bibnamefont
  {Steane}},\ }\bibfield  {title} {\bibinfo {title} {Quantum
  {{Reed}}-{{Muller}} codes},\ }\href {https://doi.org/10.1109/18.771249}
  {\bibfield  {journal} {\bibinfo  {journal} {IEEE Transactions on Information
  Theory}\ }\textbf {\bibinfo {volume} {45}},\ \bibinfo {pages} {1701}
  (\bibinfo {year} {1999})}\BibitemShut {NoStop}%
\bibitem [{\citenamefont {Anderson}\ \emph {et~al.}(2014)\citenamefont
  {Anderson}, \citenamefont {{Duclos-Cianci}},\ and\ \citenamefont
  {Poulin}}]{anderson2014}%
  \BibitemOpen
  \bibfield  {author} {\bibinfo {author} {\bibfnamefont {J.~T.}\ \bibnamefont
  {Anderson}}, \bibinfo {author} {\bibfnamefont {G.}~\bibnamefont
  {{Duclos-Cianci}}},\ and\ \bibinfo {author} {\bibfnamefont {D.}~\bibnamefont
  {Poulin}},\ }\bibfield  {title} {\bibinfo {title} {Fault-{{Tolerant
  Conversion}} between the {{Steane}} and {{Reed}}-{{Muller Quantum Codes}}},\
  }\href {https://doi.org/10.1103/PhysRevLett.113.080501} {\bibfield  {journal}
  {\bibinfo  {journal} {Physical Review Letters}\ }\textbf {\bibinfo {volume}
  {113}},\ \bibinfo {pages} {080501} (\bibinfo {year} {2014})}\BibitemShut
  {NoStop}%
\bibitem [{\citenamefont {MacWilliams}\ and\ \citenamefont
  {Sloane}(1977)}]{macwilliams1977}%
  \BibitemOpen
  \bibfield  {author} {\bibinfo {author} {\bibfnamefont {F.~J.}\ \bibnamefont
  {MacWilliams}}\ and\ \bibinfo {author} {\bibfnamefont {N.~J.~A.}\
  \bibnamefont {Sloane}},\ }\href@noop {} {\emph {\bibinfo {title} {The Theory
  of Error Correcting Codes}}},\ North-{{Holland}} Mathematical Library ; v.
  16\ (\bibinfo  {publisher} {{North-Holland Pub. Co.}},\ \bibinfo {address}
  {{Amsterdam}},\ \bibinfo {year} {1977})\BibitemShut {NoStop}%
\bibitem [{\citenamefont {Campbell}(2016)}]{campbellblog}%
  \BibitemOpen
  \bibfield  {author} {\bibinfo {author} {\bibfnamefont {E.~T.}\ \bibnamefont
  {Campbell}},\ }\href@noop {} {\bibinfo {title} {The smallest interesting
  colour code}},\ \bibinfo {howpublished}
  {\url{https://earltcampbell.com/2016/09/26/the-smallest-interesting-colour-code/}}
  (\bibinfo {year} {2016}),\ \bibinfo {note} {accessed: 2021-07-20}\BibitemShut
  {NoStop}%
\bibitem [{\citenamefont {Paetznick}\ and\ \citenamefont
  {Reichardt}(2013)}]{paetznick2013a}%
  \BibitemOpen
  \bibfield  {author} {\bibinfo {author} {\bibfnamefont {A.}~\bibnamefont
  {Paetznick}}\ and\ \bibinfo {author} {\bibfnamefont {B.~W.}\ \bibnamefont
  {Reichardt}},\ }\bibfield  {title} {\bibinfo {title} {Universal
  {{Fault}}-{{Tolerant Quantum Computation}} with {{Only Transversal Gates}}
  and {{Error Correction}}},\ }\href
  {https://doi.org/10.1103/PhysRevLett.111.090505} {\bibfield  {journal}
  {\bibinfo  {journal} {Physical Review Letters}\ }\textbf {\bibinfo {volume}
  {111}},\ \bibinfo {pages} {090505} (\bibinfo {year} {2013})}\BibitemShut
  {NoStop}%
\bibitem [{\citenamefont {Jones}(2013)}]{jones2013a}%
  \BibitemOpen
  \bibfield  {author} {\bibinfo {author} {\bibfnamefont {C.}~\bibnamefont
  {Jones}},\ }\bibfield  {title} {\bibinfo {title} {Low-overhead constructions
  for the fault-tolerant {{Toffoli}} gate},\ }\href
  {https://doi.org/10.1103/PhysRevA.87.022328} {\bibfield  {journal} {\bibinfo
  {journal} {Physical Review A}\ }\textbf {\bibinfo {volume} {87}},\ \bibinfo
  {pages} {022328} (\bibinfo {year} {2013})}\BibitemShut {NoStop}%
\bibitem [{\citenamefont {Litinski}(2019{\natexlab{b}})}]{litinski2019a}%
  \BibitemOpen
  \bibfield  {author} {\bibinfo {author} {\bibfnamefont {D.}~\bibnamefont
  {Litinski}},\ }\bibfield  {title} {\bibinfo {title} {Magic {{State
  Distillation}}: Not as {{Costly}} as {{You Think}}},\ }\href
  {https://doi.org/10.22331/q-2019-12-02-205} {\bibfield  {journal} {\bibinfo
  {journal} {Quantum}\ }\textbf {\bibinfo {volume} {3}},\ \bibinfo {pages}
  {205} (\bibinfo {year} {2019}{\natexlab{b}})},\ \Eprint
  {https://arxiv.org/abs/1905.06903} {arXiv:1905.06903} \BibitemShut {NoStop}%
\bibitem [{\citenamefont {Gidney}\ and\ \citenamefont
  {Fowler}(2019)}]{gidney2019a}%
  \BibitemOpen
  \bibfield  {author} {\bibinfo {author} {\bibfnamefont {C.}~\bibnamefont
  {Gidney}}\ and\ \bibinfo {author} {\bibfnamefont {A.~G.}\ \bibnamefont
  {Fowler}},\ }\bibfield  {title} {\bibinfo {title} {Efficient magic state
  factories with a catalyzed |{{CCZ}}{$>$} to 2|{{T}}{$>$} transformation},\
  }\href {https://doi.org/10.22331/q-2019-04-30-135} {\bibfield  {journal}
  {\bibinfo  {journal} {Quantum}\ }\textbf {\bibinfo {volume} {3}},\ \bibinfo
  {pages} {135} (\bibinfo {year} {2019})},\ \Eprint
  {https://arxiv.org/abs/1812.01238} {arXiv:1812.01238} \BibitemShut {NoStop}%
\bibitem [{\citenamefont {Campbell}\ and\ \citenamefont
  {Howard}(2017{\natexlab{b}})}]{campbell2017}%
  \BibitemOpen
  \bibfield  {author} {\bibinfo {author} {\bibfnamefont {E.~T.}\ \bibnamefont
  {Campbell}}\ and\ \bibinfo {author} {\bibfnamefont {M.}~\bibnamefont
  {Howard}},\ }\bibfield  {title} {\bibinfo {title} {Unified framework for
  magic state distillation and multiqubit gate synthesis with reduced resource
  cost},\ }\href {https://doi.org/10.1103/PhysRevA.95.022316} {\bibfield
  {journal} {\bibinfo  {journal} {Physical Review A}\ }\textbf {\bibinfo
  {volume} {95}},\ \bibinfo {pages} {022316} (\bibinfo {year}
  {2017}{\natexlab{b}})}\BibitemShut {NoStop}%
\bibitem [{\citenamefont {Bomb{\'i}n}(2015)}]{bombin2015}%
  \BibitemOpen
  \bibfield  {author} {\bibinfo {author} {\bibfnamefont {H.}~\bibnamefont
  {Bomb{\'i}n}},\ }\bibfield  {title} {\bibinfo {title} {Gauge color codes:
  Optimal transversal gates and gauge fixing in topological stabilizer codes},\
  }\href {https://doi.org/10.1088/1367-2630/17/8/083002} {\bibfield  {journal}
  {\bibinfo  {journal} {New Journal of Physics}\ }\textbf {\bibinfo {volume}
  {17}},\ \bibinfo {pages} {083002} (\bibinfo {year} {2015})}\BibitemShut
  {NoStop}%
\bibitem [{\citenamefont {Kubica}\ and\ \citenamefont
  {Beverland}(2015)}]{kubica2015}%
  \BibitemOpen
  \bibfield  {author} {\bibinfo {author} {\bibfnamefont {A.}~\bibnamefont
  {Kubica}}\ and\ \bibinfo {author} {\bibfnamefont {M.~E.}\ \bibnamefont
  {Beverland}},\ }\bibfield  {title} {\bibinfo {title} {Universal transversal
  gates with color codes - a simplified approach},\ }\href
  {https://doi.org/10.1103/PhysRevA.91.032330} {\bibfield  {journal} {\bibinfo
  {journal} {Physical Review A}\ }\textbf {\bibinfo {volume} {91}},\ \bibinfo
  {pages} {032330} (\bibinfo {year} {2015})},\ \Eprint
  {https://arxiv.org/abs/1410.0069} {arXiv:1410.0069} \BibitemShut {NoStop}%
\bibitem [{\citenamefont {Watson}\ \emph {et~al.}(2015)\citenamefont {Watson},
  \citenamefont {Campbell}, \citenamefont {Anwar},\ and\ \citenamefont
  {Browne}}]{watson2015}%
  \BibitemOpen
  \bibfield  {author} {\bibinfo {author} {\bibfnamefont {F.~H.~E.}\
  \bibnamefont {Watson}}, \bibinfo {author} {\bibfnamefont {E.~T.}\
  \bibnamefont {Campbell}}, \bibinfo {author} {\bibfnamefont {H.}~\bibnamefont
  {Anwar}},\ and\ \bibinfo {author} {\bibfnamefont {D.~E.}\ \bibnamefont
  {Browne}},\ }\bibfield  {title} {\bibinfo {title} {Qudit color codes and
  gauge color codes in all spatial dimensions},\ }\href
  {https://doi.org/10.1103/PhysRevA.92.022312} {\bibfield  {journal} {\bibinfo
  {journal} {Physical Review A}\ }\textbf {\bibinfo {volume} {92}},\ \bibinfo
  {pages} {022312} (\bibinfo {year} {2015})}\BibitemShut {NoStop}%
\bibitem [{\citenamefont {Bomb{\'i}n}(2018)}]{bombin2018a}%
  \BibitemOpen
  \bibfield  {author} {\bibinfo {author} {\bibfnamefont {H.}~\bibnamefont
  {Bomb{\'i}n}},\ }\bibfield  {title} {\bibinfo {title} {Transversal gates and
  error propagation in {{3D}} topological codes},\ }\href@noop {} {\bibfield
  {journal} {\bibinfo  {journal} {arXiv preprint}\ } (\bibinfo {year}
  {2018})},\ \Eprint {https://arxiv.org/abs/1810.09575} {arXiv:1810.09575}
  \BibitemShut {NoStop}%
\bibitem [{\citenamefont {Criger}\ and\ \citenamefont
  {Terhal}(2016)}]{criger2016}%
  \BibitemOpen
  \bibfield  {author} {\bibinfo {author} {\bibfnamefont {B.}~\bibnamefont
  {Criger}}\ and\ \bibinfo {author} {\bibfnamefont {B.}~\bibnamefont
  {Terhal}},\ }\bibfield  {title} {\bibinfo {title} {Noise {{Thresholds}} for
  the [[4, 2, 2]]-concatenated {{Toric Code}}},\ }\bibfield  {journal}
  {\bibinfo  {journal} {Quantum Information and Computation}\ }\textbf
  {\bibinfo {volume} {16}},\ \href {https://doi.org/10.26421/QIC16.15-16}
  {10.26421/QIC16.15-16} (\bibinfo {year} {2016}),\ \Eprint
  {https://arxiv.org/abs/1604.04062} {arXiv:1604.04062} \BibitemShut {NoStop}%
\bibitem [{\citenamefont {Bomb{\'i}n}\ \emph {et~al.}(2012)\citenamefont
  {Bomb{\'i}n}, \citenamefont {{Duclos-Cianci}},\ and\ \citenamefont
  {Poulin}}]{bombin2012}%
  \BibitemOpen
  \bibfield  {author} {\bibinfo {author} {\bibfnamefont {H.}~\bibnamefont
  {Bomb{\'i}n}}, \bibinfo {author} {\bibfnamefont {G.}~\bibnamefont
  {{Duclos-Cianci}}},\ and\ \bibinfo {author} {\bibfnamefont {D.}~\bibnamefont
  {Poulin}},\ }\bibfield  {title} {\bibinfo {title} {Universal topological
  phase of two-dimensional stabilizer codes},\ }\href
  {https://doi.org/10.1088/1367-2630/14/7/073048} {\bibfield  {journal}
  {\bibinfo  {journal} {New Journal of Physics}\ }\textbf {\bibinfo {volume}
  {14}},\ \bibinfo {pages} {073048} (\bibinfo {year} {2012})}\BibitemShut
  {NoStop}%
\bibitem [{\citenamefont {Delfosse}(2014)}]{delfosse2014}%
  \BibitemOpen
  \bibfield  {author} {\bibinfo {author} {\bibfnamefont {N.}~\bibnamefont
  {Delfosse}},\ }\bibfield  {title} {\bibinfo {title} {Decoding color codes by
  projection onto surface codes},\ }\href
  {https://doi.org/10.1103/PhysRevA.89.012317} {\bibfield  {journal} {\bibinfo
  {journal} {Physical Review A}\ }\textbf {\bibinfo {volume} {89}},\ \bibinfo
  {pages} {012317} (\bibinfo {year} {2014})},\ \Eprint
  {https://arxiv.org/abs/1308.6207} {arXiv:1308.6207} \BibitemShut {NoStop}%
\bibitem [{\citenamefont {Kolmogorov}(2009)}]{kolmogorov2009}%
  \BibitemOpen
  \bibfield  {author} {\bibinfo {author} {\bibfnamefont {V.}~\bibnamefont
  {Kolmogorov}},\ }\bibfield  {title} {\bibinfo {title} {Blossom {{V}}: A new
  implementation of a minimum cost perfect matching algorithm},\ }\href
  {https://doi.org/10.1007/s12532-009-0002-8} {\bibfield  {journal} {\bibinfo
  {journal} {Mathematical Programming Computation}\ }\textbf {\bibinfo {volume}
  {1}},\ \bibinfo {pages} {43} (\bibinfo {year} {2009})}\BibitemShut {NoStop}%
\bibitem [{git()}]{github}%
  \BibitemOpen
  \href@noop {} {}\bibinfo {howpublished}
  {\url{https://github.com/MikeVasmer/hybrid-decoder}}\BibitemShut {NoStop}%
\bibitem [{\citenamefont {Delfosse}\ and\ \citenamefont
  {Nickerson}(2017)}]{delfosse2017}%
  \BibitemOpen
  \bibfield  {author} {\bibinfo {author} {\bibfnamefont {N.}~\bibnamefont
  {Delfosse}}\ and\ \bibinfo {author} {\bibfnamefont {N.~H.}\ \bibnamefont
  {Nickerson}},\ }\bibfield  {title} {\bibinfo {title} {Almost-linear time
  decoding algorithm for topological codes},\ }\href@noop {} {\bibfield
  {journal} {\bibinfo  {journal} {arXiv preprint}\ } (\bibinfo {year}
  {2017})},\ \Eprint {https://arxiv.org/abs/1709.06218} {arXiv:1709.06218}
  \BibitemShut {NoStop}%
\bibitem [{\citenamefont {Brown}\ and\ \citenamefont
  {Williamson}(2020)}]{brown2020b}%
  \BibitemOpen
  \bibfield  {author} {\bibinfo {author} {\bibfnamefont {B.~J.}\ \bibnamefont
  {Brown}}\ and\ \bibinfo {author} {\bibfnamefont {D.~J.}\ \bibnamefont
  {Williamson}},\ }\bibfield  {title} {\bibinfo {title} {Parallelized quantum
  error correction with fracton topological codes},\ }\href
  {https://doi.org/10.1103/PhysRevResearch.2.013303} {\bibfield  {journal}
  {\bibinfo  {journal} {Physical Review Research}\ }\textbf {\bibinfo {volume}
  {2}},\ \bibinfo {pages} {013303} (\bibinfo {year} {2020})},\ \Eprint
  {https://arxiv.org/abs/1901.08061} {arXiv:1901.08061} \BibitemShut {NoStop}%
\bibitem [{\citenamefont {Bonilla~Ataides}\ \emph {et~al.}(2021)\citenamefont
  {Bonilla~Ataides}, \citenamefont {Tuckett}, \citenamefont {Bartlett},
  \citenamefont {Flammia},\ and\ \citenamefont {Brown}}]{bonillaataides2021}%
  \BibitemOpen
  \bibfield  {author} {\bibinfo {author} {\bibfnamefont {J.~P.}\ \bibnamefont
  {Bonilla~Ataides}}, \bibinfo {author} {\bibfnamefont {D.~K.}\ \bibnamefont
  {Tuckett}}, \bibinfo {author} {\bibfnamefont {S.~D.}\ \bibnamefont
  {Bartlett}}, \bibinfo {author} {\bibfnamefont {S.~T.}\ \bibnamefont
  {Flammia}},\ and\ \bibinfo {author} {\bibfnamefont {B.~J.}\ \bibnamefont
  {Brown}},\ }\bibfield  {title} {\bibinfo {title} {The {{XZZX}} surface
  code},\ }\href {https://doi.org/10.1038/s41467-021-22274-1} {\bibfield
  {journal} {\bibinfo  {journal} {Nature Communications}\ }\textbf {\bibinfo
  {volume} {12}},\ \bibinfo {pages} {2172} (\bibinfo {year} {2021})},\ \Eprint
  {https://arxiv.org/abs/2009.07851} {arXiv:2009.07851} \BibitemShut {NoStop}%
\bibitem [{\citenamefont {Poulsen~Nautrup}\ \emph {et~al.}(2017)\citenamefont
  {Poulsen~Nautrup}, \citenamefont {Friis},\ and\ \citenamefont
  {Briegel}}]{PoulsenNautrup2017}%
  \BibitemOpen
  \bibfield  {author} {\bibinfo {author} {\bibfnamefont {H.}~\bibnamefont
  {Poulsen~Nautrup}}, \bibinfo {author} {\bibfnamefont {N.}~\bibnamefont
  {Friis}},\ and\ \bibinfo {author} {\bibfnamefont {H.~J.}\ \bibnamefont
  {Briegel}},\ }\bibfield  {title} {\bibinfo {title} {Fault-tolerant interface
  between quantum memories and quantum processors},\ }\href
  {https://doi.org/10.1038/s41467-017-01418-2} {\bibfield  {journal} {\bibinfo
  {journal} {Nature Communications}\ }\textbf {\bibinfo {volume} {8}},\
  \bibinfo {pages} {1321} (\bibinfo {year} {2017})}\BibitemShut {NoStop}%
\bibitem [{\citenamefont {Shutty}\ and\ \citenamefont
  {Chamberland}(2022)}]{shutty2022finding}%
  \BibitemOpen
  \bibfield  {author} {\bibinfo {author} {\bibfnamefont {N.}~\bibnamefont
  {Shutty}}\ and\ \bibinfo {author} {\bibfnamefont {C.}~\bibnamefont
  {Chamberland}},\ }\bibfield  {title} {\bibinfo {title} {Finding
  fault-tolerant {C}lifford circuits using satisfiability modulo theories
  solvers and decoding merged color-surface codes},\ }\href@noop {} {\bibfield
  {journal} {\bibinfo  {journal} {arXiv preprint}\ } (\bibinfo {year}
  {2022})},\ \Eprint {https://arxiv.org/abs/2201.12450} {arXiv:2201.12450}
  \BibitemShut {NoStop}%
\bibitem [{\citenamefont {Chamberland}\ \emph
  {et~al.}(2020{\natexlab{b}})\citenamefont {Chamberland}, \citenamefont
  {Kubica}, \citenamefont {Yoder},\ and\ \citenamefont
  {Zhu}}]{chamberland2020}%
  \BibitemOpen
  \bibfield  {author} {\bibinfo {author} {\bibfnamefont {C.}~\bibnamefont
  {Chamberland}}, \bibinfo {author} {\bibfnamefont {A.}~\bibnamefont {Kubica}},
  \bibinfo {author} {\bibfnamefont {T.~J.}\ \bibnamefont {Yoder}},\ and\
  \bibinfo {author} {\bibfnamefont {G.}~\bibnamefont {Zhu}},\ }\bibfield
  {title} {\bibinfo {title} {Triangular color codes on trivalent graphs with
  flag qubits},\ }\href {https://doi.org/10.1088/1367-2630/ab68fd} {\bibfield
  {journal} {\bibinfo  {journal} {New Journal of Physics}\ }\textbf {\bibinfo
  {volume} {22}},\ \bibinfo {pages} {023019} (\bibinfo {year}
  {2020}{\natexlab{b}})},\ \Eprint {https://arxiv.org/abs/1911.00355}
  {arXiv:1911.00355} \BibitemShut {NoStop}%
\bibitem [{\citenamefont {Harrington}(2004)}]{harrington2004}%
  \BibitemOpen
  \bibfield  {author} {\bibinfo {author} {\bibfnamefont {J.}~\bibnamefont
  {Harrington}},\ }\emph {\bibinfo {title} {Analysis of Quantum
  Error-Correcting Codes Symplectic Lattice Codes and Toric Codes}},\ \href
  {https://doi.org/10.7907/AHMQ-EG82} {Ph.D. thesis},\ \bibinfo  {school}
  {Caltech} (\bibinfo {year} {2004})\BibitemShut {NoStop}%
\bibitem [{\citenamefont {Wang}\ \emph {et~al.}(2003)\citenamefont {Wang},
  \citenamefont {Harrington},\ and\ \citenamefont {Preskill}}]{wang2003}%
  \BibitemOpen
  \bibfield  {author} {\bibinfo {author} {\bibfnamefont {C.}~\bibnamefont
  {Wang}}, \bibinfo {author} {\bibfnamefont {J.}~\bibnamefont {Harrington}},\
  and\ \bibinfo {author} {\bibfnamefont {J.}~\bibnamefont {Preskill}},\
  }\bibfield  {title} {\bibinfo {title} {Confinement-{{Higgs}} transition in a
  disordered gauge theory and the accuracy threshold for quantum memory},\
  }\href {https://doi.org/10.1016/S0003-4916(02)00019-2} {\bibfield  {journal}
  {\bibinfo  {journal} {Annals of Physics}\ }\textbf {\bibinfo {volume}
  {303}},\ \bibinfo {pages} {31} (\bibinfo {year} {2003})},\ \Eprint
  {https://arxiv.org/abs/quant-ph/0207088} {arXiv:quant-ph/0207088}
  \BibitemShut {NoStop}%
\bibitem [{\citenamefont {Iverson}(2020)}]{iverson2020a}%
  \BibitemOpen
  \bibfield  {author} {\bibinfo {author} {\bibfnamefont {J.~K.}\ \bibnamefont
  {Iverson}},\ }\emph {\bibinfo {title} {Aspects of {{Fault}}-{{Tolerant
  Quantum Computation}}}},\ \href {https://doi.org/10.7907/q5ev-rm81} {Ph.D.
  thesis},\ \bibinfo  {school} {Caltech} (\bibinfo {year} {2020})\BibitemShut
  {NoStop}%
\bibitem [{\citenamefont {Iverson}\ and\ \citenamefont
  {Kubica}(2022)}]{iverson2022}%
  \BibitemOpen
  \bibfield  {author} {\bibinfo {author} {\bibfnamefont {J.}~\bibnamefont
  {Iverson}}\ and\ \bibinfo {author} {\bibfnamefont {A.}~\bibnamefont
  {Kubica}},\ }\href@noop {} {} (\bibinfo {year} {2022}),\ \bibinfo {note} {in
  preparation}\BibitemShut {NoStop}%
\bibitem [{\citenamefont {Raussendorf}\ and\ \citenamefont
  {Harrington}(2007)}]{raussendorf2007a}%
  \BibitemOpen
  \bibfield  {author} {\bibinfo {author} {\bibfnamefont {R.}~\bibnamefont
  {Raussendorf}}\ and\ \bibinfo {author} {\bibfnamefont {J.}~\bibnamefont
  {Harrington}},\ }\bibfield  {title} {\bibinfo {title} {Fault-{{Tolerant
  Quantum Computation}} with {{High Threshold}} in {{Two Dimensions}}},\ }\href
  {https://doi.org/10.1103/PhysRevLett.98.190504} {\bibfield  {journal}
  {\bibinfo  {journal} {Physical Review Letters}\ }\textbf {\bibinfo {volume}
  {98}},\ \bibinfo {pages} {190504} (\bibinfo {year} {2007})}\BibitemShut
  {NoStop}%
\bibitem [{\citenamefont {Wang}\ \emph {et~al.}(2010)\citenamefont {Wang},
  \citenamefont {Fowler}, \citenamefont {Stephens},\ and\ \citenamefont
  {Hollenberg}}]{wang2010}%
  \BibitemOpen
  \bibfield  {author} {\bibinfo {author} {\bibfnamefont {D.~S.}\ \bibnamefont
  {Wang}}, \bibinfo {author} {\bibfnamefont {A.~G.}\ \bibnamefont {Fowler}},
  \bibinfo {author} {\bibfnamefont {A.~M.}\ \bibnamefont {Stephens}},\ and\
  \bibinfo {author} {\bibfnamefont {L.~C.~L.}\ \bibnamefont {Hollenberg}},\
  }\bibfield  {title} {\bibinfo {title} {Threshold error rates for the toric
  and planar codes},\ }\href@noop {} {\bibfield  {journal} {\bibinfo  {journal}
  {Quantum Information \& Computation}\ }\textbf {\bibinfo {volume} {10}},\
  \bibinfo {pages} {456} (\bibinfo {year} {2010})},\ \Eprint
  {https://arxiv.org/abs/0905.0531} {arXiv:0905.0531} \BibitemShut {NoStop}%
\bibitem [{\citenamefont {Bravyi}\ \emph {et~al.}(2014)\citenamefont {Bravyi},
  \citenamefont {Suchara},\ and\ \citenamefont {Vargo}}]{bravyi2014a}%
  \BibitemOpen
  \bibfield  {author} {\bibinfo {author} {\bibfnamefont {S.}~\bibnamefont
  {Bravyi}}, \bibinfo {author} {\bibfnamefont {M.}~\bibnamefont {Suchara}},\
  and\ \bibinfo {author} {\bibfnamefont {A.}~\bibnamefont {Vargo}},\ }\bibfield
   {title} {\bibinfo {title} {Efficient algorithms for maximum likelihood
  decoding in the surface code},\ }\href
  {https://doi.org/10.1103/PhysRevA.90.032326} {\bibfield  {journal} {\bibinfo
  {journal} {Physical Review A}\ }\textbf {\bibinfo {volume} {90}},\ \bibinfo
  {pages} {032326} (\bibinfo {year} {2014})}\BibitemShut {NoStop}%
\bibitem [{\citenamefont {Chubb}(2021)}]{chubb2021}%
  \BibitemOpen
  \bibfield  {author} {\bibinfo {author} {\bibfnamefont {C.~T.}\ \bibnamefont
  {Chubb}},\ }\bibfield  {title} {\bibinfo {title} {General tensor network
  decoding of {{2D Pauli}} codes},\ }\href@noop {} {\bibfield  {journal}
  {\bibinfo  {journal} {arXiv preprint}\ } (\bibinfo {year} {2021})},\ \Eprint
  {https://arxiv.org/abs/2101.04125} {arXiv:2101.04125} \BibitemShut {NoStop}%
\bibitem [{\citenamefont {Katzgraber}\ \emph {et~al.}(2009)\citenamefont
  {Katzgraber}, \citenamefont {Bomb{\'i}n},\ and\ \citenamefont
  {{Martin-Delgado}}}]{katzgraber2009}%
  \BibitemOpen
  \bibfield  {author} {\bibinfo {author} {\bibfnamefont {H.~G.}\ \bibnamefont
  {Katzgraber}}, \bibinfo {author} {\bibfnamefont {H.}~\bibnamefont
  {Bomb{\'i}n}},\ and\ \bibinfo {author} {\bibfnamefont {M.~A.}\ \bibnamefont
  {{Martin-Delgado}}},\ }\bibfield  {title} {\bibinfo {title} {Error
  {{Threshold}} for {{Color Codes}} and {{Random Three}}-{{Body Ising
  Models}}},\ }\href {https://doi.org/10.1103/PhysRevLett.103.090501}
  {\bibfield  {journal} {\bibinfo  {journal} {Physical Review Letters}\
  }\textbf {\bibinfo {volume} {103}},\ \bibinfo {pages} {090501} (\bibinfo
  {year} {2009})}\BibitemShut {NoStop}%
\bibitem [{\citenamefont {Honecker}\ \emph {et~al.}(2001)\citenamefont
  {Honecker}, \citenamefont {Picco},\ and\ \citenamefont
  {Pujol}}]{honecker2001}%
  \BibitemOpen
  \bibfield  {author} {\bibinfo {author} {\bibfnamefont {A.}~\bibnamefont
  {Honecker}}, \bibinfo {author} {\bibfnamefont {M.}~\bibnamefont {Picco}},\
  and\ \bibinfo {author} {\bibfnamefont {P.}~\bibnamefont {Pujol}},\ }\bibfield
   {title} {\bibinfo {title} {Universality {{Class}} of the {{Nishimori Point}}
  in the {{2D}} \$\textbackslash pm {{J}}\$ {{Random}}-{{Bond Ising Model}}},\
  }\href {https://doi.org/10.1103/PhysRevLett.87.047201} {\bibfield  {journal}
  {\bibinfo  {journal} {Physical Review Letters}\ }\textbf {\bibinfo {volume}
  {87}},\ \bibinfo {pages} {047201} (\bibinfo {year} {2001})}\BibitemShut
  {NoStop}%
\bibitem [{\citenamefont {Merz}\ and\ \citenamefont
  {Chalker}(2002)}]{merz2002}%
  \BibitemOpen
  \bibfield  {author} {\bibinfo {author} {\bibfnamefont {F.}~\bibnamefont
  {Merz}}\ and\ \bibinfo {author} {\bibfnamefont {J.~T.}\ \bibnamefont
  {Chalker}},\ }\bibfield  {title} {\bibinfo {title} {Two-dimensional
  random-bond {{Ising}} model, free fermions, and the network model},\ }\href
  {https://doi.org/10.1103/PhysRevB.65.054425} {\bibfield  {journal} {\bibinfo
  {journal} {Physical Review B}\ }\textbf {\bibinfo {volume} {65}},\ \bibinfo
  {pages} {054425} (\bibinfo {year} {2002})}\BibitemShut {NoStop}%
\bibitem [{\citenamefont {Ohzeki}(2009)}]{ohzeki2009}%
  \BibitemOpen
  \bibfield  {author} {\bibinfo {author} {\bibfnamefont {M.}~\bibnamefont
  {Ohzeki}},\ }\bibfield  {title} {\bibinfo {title} {Locations of multicritical
  points for spin glasses on regular lattices},\ }\href
  {https://doi.org/10.1103/PhysRevE.79.021129} {\bibfield  {journal} {\bibinfo
  {journal} {Physical Review E}\ }\textbf {\bibinfo {volume} {79}},\ \bibinfo
  {pages} {021129} (\bibinfo {year} {2009})},\ \Eprint
  {https://arxiv.org/abs/0811.0464} {arXiv:0811.0464} \BibitemShut {NoStop}%
\bibitem [{\citenamefont {Poulin}(2005)}]{poulin2005}%
  \BibitemOpen
  \bibfield  {author} {\bibinfo {author} {\bibfnamefont {D.}~\bibnamefont
  {Poulin}},\ }\bibfield  {title} {\bibinfo {title} {Stabilizer {{Formalism}}
  for {{Operator Quantum Error Correction}}},\ }\href
  {https://doi.org/10.1103/PhysRevLett.95.230504} {\bibfield  {journal}
  {\bibinfo  {journal} {Physical Review Letters}\ }\textbf {\bibinfo {volume}
  {95}},\ \bibinfo {pages} {230504} (\bibinfo {year} {2005})}\BibitemShut
  {NoStop}%
\bibitem [{\citenamefont {Kubica}\ and\ \citenamefont
  {Vasmer}(2021)}]{kubica2021}%
  \BibitemOpen
  \bibfield  {author} {\bibinfo {author} {\bibfnamefont {A.}~\bibnamefont
  {Kubica}}\ and\ \bibinfo {author} {\bibfnamefont {M.}~\bibnamefont
  {Vasmer}},\ }\bibfield  {title} {\bibinfo {title} {Single-shot quantum error
  correction with the three-dimensional subsystem toric code},\ }\href@noop {}
  {\bibfield  {journal} {\bibinfo  {journal} {arXiv preprint}\ } (\bibinfo
  {year} {2021})},\ \Eprint {https://arxiv.org/abs/2106.02621}
  {arXiv:2106.02621} \BibitemShut {NoStop}%
\bibitem [{\citenamefont {Vuillot}\ and\ \citenamefont
  {Breuckmann}(2019)}]{vuillot2019a}%
  \BibitemOpen
  \bibfield  {author} {\bibinfo {author} {\bibfnamefont {C.}~\bibnamefont
  {Vuillot}}\ and\ \bibinfo {author} {\bibfnamefont {N.~P.}\ \bibnamefont
  {Breuckmann}},\ }\bibfield  {title} {\bibinfo {title} {Quantum {{Pin
  Codes}}},\ }\href@noop {} {\bibfield  {journal} {\bibinfo  {journal} {arXiv
  preprint}\ } (\bibinfo {year} {2019})},\ \Eprint
  {https://arxiv.org/abs/1906.11394} {arXiv:1906.11394} \BibitemShut {NoStop}%
\bibitem [{\citenamefont {Bomb{\'i}n}\ and\ \citenamefont
  {{Martin-Delgado}}(2007{\natexlab{b}})}]{bombin2007}%
  \BibitemOpen
  \bibfield  {author} {\bibinfo {author} {\bibfnamefont {H.}~\bibnamefont
  {Bomb{\'i}n}}\ and\ \bibinfo {author} {\bibfnamefont {M.~A.}\ \bibnamefont
  {{Martin-Delgado}}},\ }\bibfield  {title} {\bibinfo {title} {Exact
  topological quantum order in {{D}} = 3 and beyond: Branyons and brane-net
  condensates},\ }\href {https://doi.org/10.1103/PhysRevB.75.075103} {\bibfield
   {journal} {\bibinfo  {journal} {Physical Review B}\ }\textbf {\bibinfo
  {volume} {75}},\ \bibinfo {pages} {075103} (\bibinfo {year}
  {2007}{\natexlab{b}})}\BibitemShut {NoStop}%
\bibitem [{\citenamefont {Kitaev}(1997{\natexlab{b}})}]{kitaev1997a}%
  \BibitemOpen
  \bibfield  {author} {\bibinfo {author} {\bibfnamefont {A.~Y.}\ \bibnamefont
  {Kitaev}},\ }\bibfield  {title} {\bibinfo {title} {Quantum {{Error
  Correction}} with {{Imperfect Gates}}},\ }in\ \href
  {https://doi.org/10.1007/978-1-4615-5923-8_19} {\emph {\bibinfo {booktitle}
  {Quantum {{Communication}}, {{Computing}}, and {{Measurement}}}}},\ \bibinfo
  {editor} {edited by\ \bibinfo {editor} {\bibfnamefont {O.}~\bibnamefont
  {Hirota}}, \bibinfo {editor} {\bibfnamefont {A.~S.}\ \bibnamefont {Holevo}},\
  and\ \bibinfo {editor} {\bibfnamefont {C.~M.}\ \bibnamefont {Caves}}}\
  (\bibinfo  {publisher} {{Springer US}},\ \bibinfo {address} {{Boston, MA}},\
  \bibinfo {year} {1997})\ pp.\ \bibinfo {pages} {181--188}\BibitemShut
  {NoStop}%
\bibitem [{\citenamefont {Kitaev}(2003)}]{kitaev2003}%
  \BibitemOpen
  \bibfield  {author} {\bibinfo {author} {\bibfnamefont {A.~Y.}\ \bibnamefont
  {Kitaev}},\ }\bibfield  {title} {\bibinfo {title} {Fault-tolerant quantum
  computation by anyons},\ }\href
  {https://doi.org/10.1016/S0003-4916(02)00018-0} {\bibfield  {journal}
  {\bibinfo  {journal} {Annals of Physics}\ }\textbf {\bibinfo {volume}
  {303}},\ \bibinfo {pages} {2} (\bibinfo {year} {2003})},\ \Eprint
  {https://arxiv.org/abs/quant-ph/9707021} {arXiv:quant-ph/9707021}
  \BibitemShut {NoStop}%
\bibitem [{\citenamefont {Calderbank}\ and\ \citenamefont
  {Shor}(1996)}]{calderbank1996}%
  \BibitemOpen
  \bibfield  {author} {\bibinfo {author} {\bibfnamefont {A.~R.}\ \bibnamefont
  {Calderbank}}\ and\ \bibinfo {author} {\bibfnamefont {P.~W.}\ \bibnamefont
  {Shor}},\ }\bibfield  {title} {\bibinfo {title} {Good quantum
  error-correcting codes exist},\ }\href
  {https://doi.org/10.1103/PhysRevA.54.1098} {\bibfield  {journal} {\bibinfo
  {journal} {Physical Review A}\ }\textbf {\bibinfo {volume} {54}},\ \bibinfo
  {pages} {1098} (\bibinfo {year} {1996})}\BibitemShut {NoStop}%
\bibitem [{\citenamefont {Krishna}\ and\ \citenamefont
  {Tillich}(2019)}]{krishna2019a}%
  \BibitemOpen
  \bibfield  {author} {\bibinfo {author} {\bibfnamefont {A.}~\bibnamefont
  {Krishna}}\ and\ \bibinfo {author} {\bibfnamefont {J.-P.}\ \bibnamefont
  {Tillich}},\ }\bibfield  {title} {\bibinfo {title} {Towards {{Low Overhead
  Magic State Distillation}}},\ }\href
  {https://doi.org/10.1103/PhysRevLett.123.070507} {\bibfield  {journal}
  {\bibinfo  {journal} {Physical Review Letters}\ }\textbf {\bibinfo {volume}
  {123}},\ \bibinfo {pages} {070507} (\bibinfo {year} {2019})}\BibitemShut
  {NoStop}%
\bibitem [{\citenamefont {Coxeter}(1973)}]{coxeter1973}%
  \BibitemOpen
  \bibfield  {author} {\bibinfo {author} {\bibfnamefont {H.~S.~M.}\
  \bibnamefont {Coxeter}},\ }\href@noop {} {\emph {\bibinfo {title} {Regular
  Polytopes}}}\ (\bibinfo  {publisher} {Dover},\ \bibinfo {address} {New
  York},\ \bibinfo {year} {1973})\BibitemShut {NoStop}%
\bibitem [{\citenamefont {Vasmer}(2019)}]{vasmer2019a}%
  \BibitemOpen
  \bibfield  {author} {\bibinfo {author} {\bibfnamefont {M.}~\bibnamefont
  {Vasmer}},\ }\emph {\bibinfo {title} {Fault-Tolerant Quantum Computing with
  Three-Dimensional Surface Codes}},\ \href
  {https://discovery.ucl.ac.uk/id/eprint/10087863} {Ph.D. thesis},\ \bibinfo
  {school} {University College London} (\bibinfo {year} {2019})\BibitemShut
  {NoStop}%
\bibitem [{\citenamefont {Breuckmann}\ and\ \citenamefont
  {Terhal}(2016)}]{breuckmann2016a}%
  \BibitemOpen
  \bibfield  {author} {\bibinfo {author} {\bibfnamefont {N.~P.}\ \bibnamefont
  {Breuckmann}}\ and\ \bibinfo {author} {\bibfnamefont {B.~M.}\ \bibnamefont
  {Terhal}},\ }\bibfield  {title} {\bibinfo {title} {Constructions and {{Noise
  Threshold}} of {{Hyperbolic Surface Codes}}},\ }\href
  {https://doi.org/10.1109/TIT.2016.2555700} {\bibfield  {journal} {\bibinfo
  {journal} {IEEE Transactions on Information Theory}\ }\textbf {\bibinfo
  {volume} {62}},\ \bibinfo {pages} {3731} (\bibinfo {year} {2016})},\ \Eprint
  {https://arxiv.org/abs/1506.04029} {arXiv:1506.04029} \BibitemShut {NoStop}%
\bibitem [{\citenamefont {Breuckmann}\ \emph {et~al.}(2017)\citenamefont
  {Breuckmann}, \citenamefont {Vuillot}, \citenamefont {Campbell},
  \citenamefont {Krishna},\ and\ \citenamefont {Terhal}}]{breuckmann2017}%
  \BibitemOpen
  \bibfield  {author} {\bibinfo {author} {\bibfnamefont {N.~P.}\ \bibnamefont
  {Breuckmann}}, \bibinfo {author} {\bibfnamefont {C.}~\bibnamefont {Vuillot}},
  \bibinfo {author} {\bibfnamefont {E.}~\bibnamefont {Campbell}}, \bibinfo
  {author} {\bibfnamefont {A.}~\bibnamefont {Krishna}},\ and\ \bibinfo {author}
  {\bibfnamefont {B.~M.}\ \bibnamefont {Terhal}},\ }\bibfield  {title}
  {\bibinfo {title} {Hyperbolic and semi-hyperbolic surface codes for quantum
  storage},\ }\href {https://doi.org/10.1088/2058-9565/aa7d3b} {\bibfield
  {journal} {\bibinfo  {journal} {Quantum Science and Technology}\ }\textbf
  {\bibinfo {volume} {2}},\ \bibinfo {pages} {035007} (\bibinfo {year}
  {2017})}\BibitemShut {NoStop}%
\bibitem [{\citenamefont {Breuckmann}(2018)}]{breuckmann2018}%
  \BibitemOpen
  \bibfield  {author} {\bibinfo {author} {\bibfnamefont {N.~P.}\ \bibnamefont
  {Breuckmann}},\ }\bibfield  {title} {\bibinfo {title} {{{PhD}} thesis:
  Homological {{Quantum Codes Beyond}} the {{Toric Code}}},\ }\href@noop {}
  {\bibfield  {journal} {\bibinfo  {journal} {arXiv preprint}\ } (\bibinfo
  {year} {2018})},\ \Eprint {https://arxiv.org/abs/1802.01520}
  {arXiv:1802.01520} \BibitemShut {NoStop}%
\bibitem [{\citenamefont {Humphreys}(1990)}]{humphreys1990}%
  \BibitemOpen
  \bibfield  {author} {\bibinfo {author} {\bibfnamefont {J.~E.}\ \bibnamefont
  {Humphreys}},\ }\href {https://doi.org/10.1017/CBO9780511623646} {\emph
  {\bibinfo {title} {Reflection {{Groups}} and {{Coxeter Groups}}}}},\ \bibinfo
  {edition} {1st}\ ed.\ (\bibinfo  {publisher} {{Cambridge University Press}},\
  \bibinfo {year} {1990})\BibitemShut {NoStop}%
\bibitem [{\citenamefont {Davis}(2008)}]{davis2008}%
  \BibitemOpen
  \bibfield  {author} {\bibinfo {author} {\bibfnamefont {M.}~\bibnamefont
  {Davis}},\ }\href@noop {} {\emph {\bibinfo {title} {The Geometry and Topology
  of Coxeter Groups}}},\ London {{Mathematical Society}} Monographs Series\
  (\bibinfo  {publisher} {{Princeton University Press}},\ \bibinfo {address}
  {{Princeton}},\ \bibinfo {year} {2008})\BibitemShut {NoStop}%
\bibitem [{\citenamefont {Brown}\ \emph {et~al.}(2016)\citenamefont {Brown},
  \citenamefont {Nickerson},\ and\ \citenamefont {Browne}}]{brown2016a}%
  \BibitemOpen
  \bibfield  {author} {\bibinfo {author} {\bibfnamefont {B.~J.}\ \bibnamefont
  {Brown}}, \bibinfo {author} {\bibfnamefont {N.~H.}\ \bibnamefont
  {Nickerson}},\ and\ \bibinfo {author} {\bibfnamefont {D.~E.}\ \bibnamefont
  {Browne}},\ }\bibfield  {title} {\bibinfo {title} {Fault-tolerant error
  correction with the gauge color code},\ }\href
  {https://doi.org/10.1038/ncomms12302} {\bibfield  {journal} {\bibinfo
  {journal} {Nature Communications}\ }\textbf {\bibinfo {volume} {7}},\
  \bibinfo {pages} {12302} (\bibinfo {year} {2016})}\BibitemShut {NoStop}%
\bibitem [{\citenamefont {Kesselring}\ \emph {et~al.}(2018)\citenamefont
  {Kesselring}, \citenamefont {Pastawski}, \citenamefont {Eisert},\ and\
  \citenamefont {Brown}}]{kesselring2018}%
  \BibitemOpen
  \bibfield  {author} {\bibinfo {author} {\bibfnamefont {M.~S.}\ \bibnamefont
  {Kesselring}}, \bibinfo {author} {\bibfnamefont {F.}~\bibnamefont
  {Pastawski}}, \bibinfo {author} {\bibfnamefont {J.}~\bibnamefont {Eisert}},\
  and\ \bibinfo {author} {\bibfnamefont {B.~J.}\ \bibnamefont {Brown}},\
  }\bibfield  {title} {\bibinfo {title} {The boundaries and twist defects of
  the color code and their applications to topological quantum computation},\
  }\href {https://doi.org/10.22331/q-2018-10-19-101} {\bibfield  {journal}
  {\bibinfo  {journal} {Quantum}\ }\textbf {\bibinfo {volume} {2}},\ \bibinfo
  {pages} {101} (\bibinfo {year} {2018})},\ \Eprint
  {https://arxiv.org/abs/1806.02820} {arXiv:1806.02820} \BibitemShut {NoStop}%
\bibitem [{\citenamefont {Bomb{\'i}n}(2010)}]{bombin2010a}%
  \BibitemOpen
  \bibfield  {author} {\bibinfo {author} {\bibfnamefont {H.}~\bibnamefont
  {Bomb{\'i}n}},\ }\bibfield  {title} {\bibinfo {title} {Topological {{Order}}
  with a {{Twist}}: Ising {{Anyons}} from an {{Abelian Model}}},\ }\href
  {https://doi.org/10.1103/PhysRevLett.105.030403} {\bibfield  {journal}
  {\bibinfo  {journal} {Physical Review Letters}\ }\textbf {\bibinfo {volume}
  {105}},\ \bibinfo {pages} {030403} (\bibinfo {year} {2010})}\BibitemShut
  {NoStop}%
\bibitem [{\citenamefont {Vidal}(2007)}]{vidal2007}%
  \BibitemOpen
  \bibfield  {author} {\bibinfo {author} {\bibfnamefont {G.}~\bibnamefont
  {Vidal}},\ }\bibfield  {title} {\bibinfo {title} {Entanglement
  renormalization},\ }\href {https://doi.org/10.1103/PhysRevLett.99.220405}
  {\bibfield  {journal} {\bibinfo  {journal} {Phys. Rev. Lett.}\ }\textbf
  {\bibinfo {volume} {99}},\ \bibinfo {pages} {220405} (\bibinfo {year}
  {2007})}\BibitemShut {NoStop}%
\end{thebibliography}
